\documentclass[prl,superscriptaddress,twocolumn,longbibliography,reprint]{revtex4-2}

\usepackage{dcolumn,array}
\usepackage{bm}
\usepackage{subfigure}
\usepackage{amsmath,amsfonts,amssymb,graphics,graphicx,epsfig,color,times,natbib}
\usepackage{tikz}
\usepackage{pgfplots}
\usepackage{verbatim}
\usepackage{url}
\usepackage[unicode=true,
 bookmarks=true,bookmarksnumbered=false,bookmarksopen=false,
 breaklinks=false,pdfborder={0 0 1},backref=false,colorlinks=true]
 {hyperref}
\hypersetup{
 linkcolor=magenta, urlcolor=blue, citecolor=blue, pdfstartview={FitH}, hyperfootnotes=false, unicode=true}

\newcommand{\ipic}[3][-0.5]{\raisebox{#1\height}{\scalebox{#3}{\includegraphics{#2}}}}

\newcommand{\tr}{\mbox{tr}}

\date{\today}

\usepackage{mathtools}
\usepackage{amsthm}
\newtheorem{theorem}{Theorem}

\newtheorem{lemma}{Lemma}

\begin{document}


\title{The Presence and Absence of Barren Plateaus in Tensor-network Based Machine Learning}

\author{Zidu Liu}\thanks{These authors contributed equally to this work.}
\author{Li-Wei Yu}\thanks{These authors contributed equally to this work.}
 \affiliation{Center for Quantum Information, IIIS, Tsinghua University, Beijing 100084, People's Republic of China}
 
 \author{L.-M. Duan}\email{lmduan@tsinghua.edu.cn}
\affiliation{Center for Quantum Information, IIIS, Tsinghua University, Beijing 100084, People's Republic of China}
\author{Dong-Ling Deng}
\email{dldeng@tsinghua.edu.cn}
 \affiliation{Center for Quantum Information, IIIS, Tsinghua University, Beijing 100084, People's Republic of China}
\affiliation{Shanghai Qi Zhi Institute, 41th Floor, AI Tower, No. 701 Yunjin Road, Xuhui District, Shanghai 200232, China}

\begin{abstract}
Tensor networks are efficient representations of high-dimensional tensors with widespread applications in quantum many-body physics. Recently, they have been adapted to the field of machine learning, giving rise to an emergent research frontier that has attracted considerable attention. Here, we study the trainability of tensor-network based machine learning models by exploring the landscapes of different loss functions, with a focus on the matrix product states (also called tensor trains) architecture. In particular, we rigorously prove that barren plateaus (i.e., exponentially vanishing gradients) prevail in the training process of the machine learning algorithms with global loss functions. 
Whereas, for local loss functions the gradients with respect to variational parameters near the local observables do not vanish as the system size increases. Therefore, the barren plateaus are absent in this case and the corresponding models could be efficiently trainable. Our results reveal a crucial aspect of tensor-network based machine learning in a rigorous fashion,  which provide a valuable guide for both practical applications and theoretical studies in the future. 
\end{abstract}
\date{\today}
\maketitle

The interplay between machine learning and physics would benefit both fields \cite{Dunjko2018Machine,Sarma2019Machine,Carleo2019Machine}. On the one hand, machine learning tools and techniques can be utilized to tackle intricate problems in physics. Along this line, notable progresses have been made recently and machine learning has cemented its role in a wide spectrum of physical problems \cite{Dunjko2018Machine,Sarma2019Machine,Carleo2019Machine}, ranging from classifying different phases of matter \cite{Wang2016Discovering,Carrasquilla2017Machine,Zhang2017Quantum,Zhang2018Machine,Rodriguez2019Identifying,Scheurer2020Unsupervised,Yu2021Unsupervised}, quantum nonlocality detection \cite{Deng2017MachineBN}, quantum tomography \cite{Torlai2018Neuralnetwork}, and topological quantum compiling \cite{Zhang2020Topological}, to black hole detection \cite{Pasquato2016Detecting},   gravitational lenses \cite{Hezaveh2017Fast} and wave analysis \cite{Rahul2013Application,Abbott2016Observation}, glassy dynamics \cite{Schoenholz2016Structural}, and material design \cite{Kalinin2015Big}, etc. On the other hand, ideas and concepts 
originated  in the physics domain also hold the intriguing potentials to enhance, speed up, or innovate machine learning. Within this vein, a number of quantum learning algorithms have been developed \cite{Biamonte2017Quantum,Harrow2009Quantum,Lloyd2014Quantum,Dunjko2016Quantum,Amin2018Quantum,Gao2018Quantum,Lloyd2018Quantum,Hu2019Quantum,Schuld2019Quantum,Liu2021Rigorous}, which may offer exponential quantum advantages over their classical counterparts. In addition, recent works have also exploited  a variety of physical concepts, such as entanglement \cite{Deng2017Quantum,Levine2019Quantum}, locality \cite{Lin2017Why}, and renormalization group \cite{Mehta2014Exact,Koch2018Mutual}, to gain valuable insights on understanding why deep learning  works so well. Here, we study the trainability of tensor-network based machine learning models, which likewise draw crucial inspiration from physics, through exploring the landscapes of different loss functions (see Fig. \ref{fig:1} for a pictorial illustration).

\begin{figure}
\centering
\includegraphics[scale=0.3]{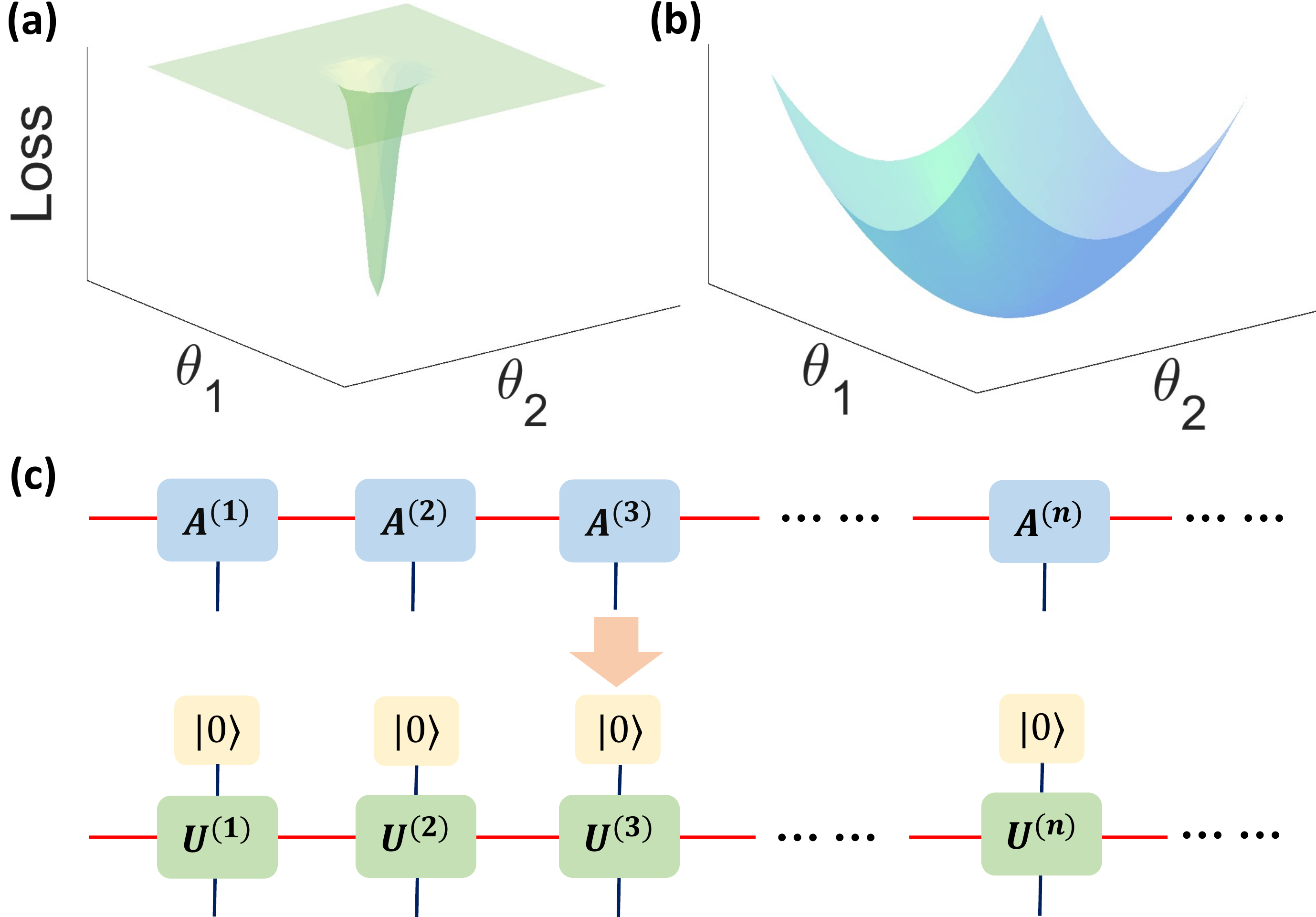}
\caption{Landscapes of different loss functions and the unitary embedding of matrix product states. (a) For global loss functions, the landscape hosts a barren plateau, where the gradient along any reasonable direction vanishes. This will hinder the training of the corresponding models by gradient-based algorithms. Here, $\theta_{1}$ and $\theta_2$ represent two variational parameters. (b) For local loss functions, the landscape is free of barren plateaus and the corresponding models are efficiently trainable. (c)  A graphical illustration of the unitary embedding of matrix product states.}
\label{fig:1}
\end{figure}

Tensor network is one of the most powerful tools for studying quantum many-body systems \cite{Orus2019Tensor,Biamonte2019Lectures,Cirac2020Matrix}. Inspired by their success in quantum physics, recently there has been a huge surge of interest in adapting them to machine learning \cite{Cichocki2014Tensor,Cichocki2016Tensor,Cichocki2016Tensor2,Stoudenmire2016Supervised,Novikov2016Exponential,Liu2018Entanglement,Glasser2020From,Liu2019Machine,Levine2018Deep,Cheng2021Supervised,Stoudenmire2018Learning,Han2018Unsupervised,Liu2021Tensor,Chen2018Equivalence,Levine2019Quantum,Bhatia2019Matrix,Efthymiou2019Tensornetwork,Huggins2019Towards,Sun2020Tangent,Wang2020Anomaly,Wall2021Generative,Costa2021Tensortrain,Kardashin2021Quantum}. Indeed, tensor networks have been invoked in various machine learning scenarios, including dimensionality reduction \cite{Cichocki2016Tensor}, image recognition \cite{Stoudenmire2016Supervised,Novikov2016Exponential,Liu2018Entanglement,Glasser2020From}, generative models \cite{Han2018Unsupervised,Liu2021Tensor}, natural language processing \cite{Guo2018Matrix,Meichanetzidis2020Quantum}, anomaly detection \cite{Wang2020Anomaly}, etc. Tensor-network based machine learning models bear several intriguing features from both theoretical and practical perspectives. At the theoretical level, their expressive power can be naturally characterized by the entanglement structure of the underlying tensor-network quantum states. This gives rise to a possible way to determine their applicability to a given learning task via analyzing the entanglement properties \cite{Convy2021Mutual,Lu2021Tensor}. In addition, tensor networks provide a convenient framework to study how and why certain quantum learning models would exhibit exponential advantages over their classical analogues \cite{Gao2018Quantum,Levine2019Quantum,Gao2021Enhancing}. Using tensor networks, recently a separation in expressive power between Bayesian networks and their quantum extensions has been rigorously proved and shown to be originated from quantum nonlocality and quantum contextuality \cite{Gao2021Enhancing}. At the practical level, numerical techniques used for tensor networks, such as the canonical form and renormalization \cite{Stoudenmire2016Supervised,Stoudenmire2018Learning},  are also useful and inspiring for optimizing and training machine learning models. We note that a number of open-source libraries have been released  in the community \cite{Abadi2016TensorFlow,roberts2019tensornetwork}, which have boosted and will continue to nourish the development of tensor-network based machine learning. This emergent research direction is growing rapidly, with notable progresses made from various aspects. Yet, undoubtedly it is still in its infancy and many important issues remain largely unexplored. 

In classical machine learning, a notorious obstacle for training artificial neural networks concerns the barren plateau phenomenon, where the gradient of the loss function along any direction vanishes exponentially with the problem size \cite{Pascanu2013Difficulty}. Recently, barren plateaus have also been shown to exist for many quantum learning models based on variational quantum circuits and the related topics are still under active study at the current stage \cite{Mcclean2018Barren,Wang2020Noise,Cerezo2021Higher,Sharma2020Trainability,Arrasmith2020Effect,Holmes2021Barren,Marrero2020Entanglement,Patti2020Entanglement,Pesah2020Absence,Arrasmith2021Equivalence,Cerezo2021Cost,Uvarov2021Barren,Grant2019Initialization,Zhao2021Analyzing}.  In this paper, we investigate the presence and absence of barren plateaus for tensor-network based machine learning models, which is a crucial but hitherto unexplored issue in the literature. We focus on the matrix product states (MPS) architecture, which is a special case of tensor networks in one dimension. Through exploring the landscapes of different loss functions, we  prove rigorously that barren plateaus arise generally for MPS-based learning models with global loss functions, rendering their training process inefficient by gradient-based algorithms and the architecture unscalable. In contrast, for local loss functions the gradients with respect to variational parameters near the local observables do not vanish and is independent of the system size. As a result, no barren plateau appears in this case and the corresponding models are efficiently trainable. We also prove that for local loss functions the gradients decays exponentially with the distance between the region where the local observable acts and the site that hosts the derivative parameter. This reveals the locality property of tensor networks from a new perspective and would be valuable for reducing the computational cost in training corresponding models. In addition,  we carry out numerical simulations to show that the above asymptotic results holds as well for MPS-based learning models with modest sizes in practice.

{\it Notations and framework.}---An arbitrary MPS with the periodic boundary condition takes the form \cite{Schollwock2011Densitymatrix}: $|\psi\rangle=\sum_{j_1, \dots, j_n}\text{Tr}\left[A_{j_1}^{(1)}A_{j_2}^{(2)}\cdots A_{j_n}^{(n)}\right]|j_1,\cdots j_n\rangle$, where $|j_k\rangle$ denotes the local state of the $k$-th physical site with physical dimension $d$, $A^{k}_{j_k}$'s are $D\times D$ matrices with $D$ representing the virtual bond dimension.  Any such MPS can be embedded by the $Dd\times Dd$ unitary matrices $U$ \cite{Perez2007Matrix,Gross2010Quantum,Haferkamp2021Emergent}, as shown in Fig.~\ref{fig:1} (c). For the unitary embedding MPS $|\psi\rangle$, the values of $ \langle\psi|\psi\rangle$ are exponentially concentrated around one \cite{Haferkamp2021Emergent}.
We consider loss functions expressed by the unitary embedding MPS, with the bond matrices randomly initialized \cite{Haferkamp2021Emergent,Kliesch2019Guaranteed}. Those random unitary operators with the measure $dU$ form the approximate unitary $2$-designs
\cite{Renes2004Symmetric,Dankert2009Exact,Harrow2009Random}, which indicates that the first and second moments are approximately the same as the corresponding moments with respect to the Haar measure $dU_H$, i.e., $M_1(dU)=M_1(dU_H)$, $M_2(dU)=M_2(dU_H)$. For convenience, we will use a diagrammatic language to describe tensor networks and carry out related calculations \cite{Biamonte2019Lectures}. As an example, the first and second moments with respect to the Haar measure \cite{Collins2006Integration} of the unitary group $U(N)$ are given by the  Weingarten functions with the following graph \cite{barren1d_supplementary}: 
\begin{eqnarray}
&&\ipic{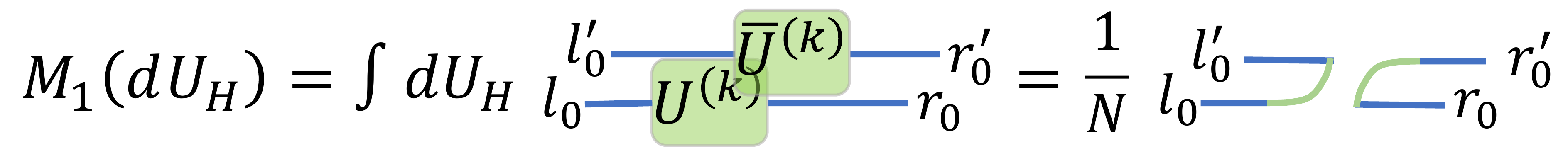}{0.2},\\
&&\ipic{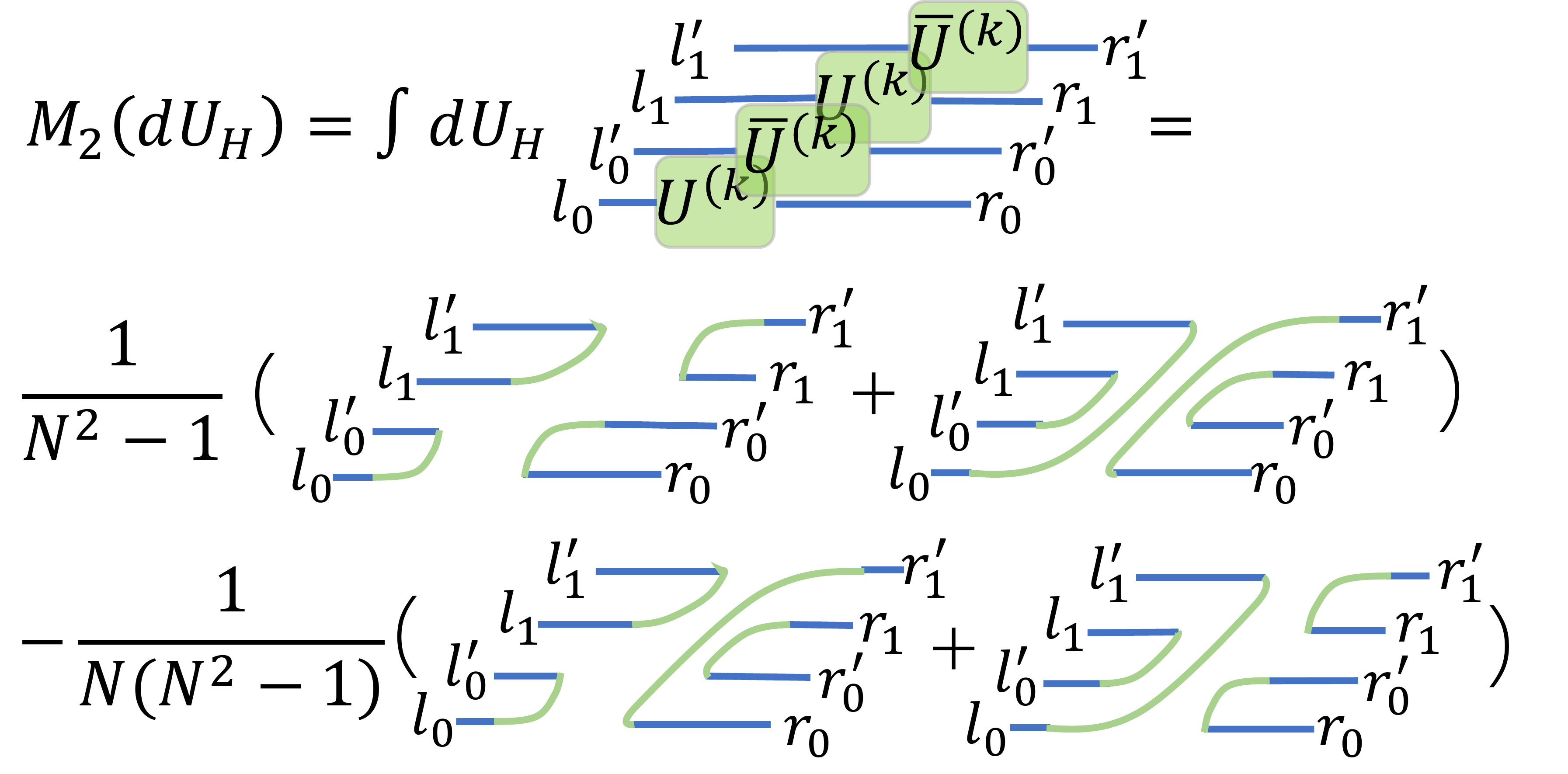}{0.2},
\end{eqnarray}
where the connected lines represent the contractions of the corresponding indices.

{\it The presence of barren plateaus for global loss functions.}---We start with a rigorous proof of the presence of barren plateaus for MPS-based machine learning with global loss functions. In quantum machine learning, global loss functions are widely used in various scenarios and the optimization of these functions plays a vital role for many learning tasks \cite{Romero2017Quantum,Li2017Efficient,Huang2019Near,Carolan2020Variational,Cirstoiu2020Variational,Cerezo2020Variationalquantum}. For simplicity and technique convenience, here we consider the following loss function as an example:
\begin{equation}
 \mathcal{L}_g =1-\|\langle\psi(\Theta)|\phi\rangle\|^2,   
\end{equation}
where $|\phi\rangle$ is a normalized $n$-qudit target quantum state, and $|\psi({{\Theta}})\rangle$ denotes the $n$-site unitary embedded MPS parameterized by a set of real numbers $\Theta=\{\theta_k^{(i)}\}$ with $\theta_k^{(i)}$ representing the $k$-th parameter on the $i$-th site. We mention that this loss function may have applications in quantum state preparation or tomography \cite{Sugiyama2013Precision}. We now present our first theorem. 


\begin{theorem}\label{theorem:global} Define the derivative of the above loss function with respect to the variational parameter $\theta_k^{(i)}$ by $\partial_k^{(i)}\mathcal{L}_g$. Then  $\forall \;\theta_k^{(i)} \in \Theta$, $\partial_k^{(i)}\mathcal{L}_g$ obeys the following inequality:
\begin{equation}\label{Prob:global}
{\rm Pr}\left(|\partial_k^{(i)}\mathcal{L}_g|>\epsilon\right)\leq \epsilon^{-2} \mathcal{O}(d^{-n}),
\end{equation}
where $\rm Pr(\cdot)$ represents the probability.
\end{theorem}
\begin{proof} We give the main idea here. The detailed proof is technically involved and thus left to the Supplementary Materials \cite{barren1d_supplementary}. 
 By using the fact that each local random unitary for embedding the MPS is approximately $2$-design \cite{Haferkamp2021Emergent}, we can first prove that the average of the gradient over the whole parameter space vanishes, i.e.,  $\langle\partial_k^{(i)}\mathcal{L}_g\rangle=0\, (\forall \theta_k^{(i)} \in \Theta)$. In the next step, we calculate the variance of the derivative: $\text{Var}(\partial_k^{(i)} \mathcal{L}_g) \equiv \langle (\partial_k^{(i)} \mathcal{L}_g)^2 \rangle - \langle \partial_k^{(i)} \mathcal{L}_g\rangle^2=\langle (\partial_k^{(i)} \mathcal{L}_g)^2 \rangle$, which can be upper bounded by an exponentially small number as the system size increases. At last, by using the Chebyshev's inequality \cite{Tchebichef1867Des}, we conclude that the probability with the absolute value of the gradient being greater than a constant $\epsilon$ is bounded by an exponentially small number. This leads to Eq.~(\ref{Prob:global}) and completes the proof.
\end{proof}



The above theorem indicates that in the process of training the MPS-based machine learning models with a global loss function,  the probability that the gradient along any direction is non-zero to any precision $\epsilon$  vanishes exponentially in terms of the number of physical sites, given that each unitary matrix $U_{i_k}^{(k)}$ in the embedded MPS is  randomly initialized. In other words, as the system size $n$ increases the gradient vanishes exponentially almost everywhere in the parameter space, giving rise to a barren plateau for the landscape of the loss function. The presence of barren plateaus requires an exponentially large precision and iteration steps to navigate through the landscape, thus rendering any gradient-based algorithms inefficient and impractical in training the corresponding models when scaled up to large system sizes. In fact, even for some other optimization approaches, such as these Hessian-based  \cite{Cerezo2021Higher} or gradient free \cite{Arrasmith2020Effect} ones, the barren plateaus may still pose a serious challenge to escape from. We mention that, although the rigorous proof of Theorem \ref{theorem:global} is only given for a specific loss function for  technical convenience, similar conclusion would hold for general global loss functions. This is supported by our numerical simulations on the  Kullback-Leibler (KL) divergence \cite{Kullback1951Information} in following paragraphs.

We remark that the mechanism for the existence of the barren plateaus in our MPS-based model is intrinsically different from that for learning models based on quantum  variational circuits \cite{Mcclean2018Barren}, despite the fact that the derivative decays exponentially with respect to the system size in both scenarios. In the case of variational quantum circuits, the exponential decay originates from the $2$-moment of $U(2^n)$ with respect to the Haar measure. This is reflected by the underlying fact  that the random quantum circuits are approximate $2$-design for $U(2^n)$. In contrast, for the case of MPS-based models the exponential decay is essentially due to the $2$-moment of $U(Dd)^{\otimes{n}}$, which only requires that each random embedding  $U(Dd)$ is approximately unitary $2$-design.

{\it The absence of barren plateaus for local loss functions.}---  Local loss functions are those constructed with the local observables. They are also widely used in different contexts, including quantum classifiers \cite{Lu2020Quantum}, classical shadow \cite{Huang2020Predicting}, and solving the ground state of many-body Hamiltonians \cite{Carleo2017Solving}. Without loss of generality, we consider the following local loss function as an example:
\begin{equation}
\label{local_loss}
    \mathcal{L}_l = \langle \psi (\Theta) |\hat{O}_m |\psi(\Theta) / \langle \psi (\Theta) |\psi(\Theta)\rangle,
\end{equation}
where $\hat{O}_m$ is an arbitrary local operator acting on the $m$-th site of the MPS. Here, we further assume $\text{Tr}(\hat{O}_m) = 0$ for  technical convenience. This does not diminish the generality of $\hat{O}_m$ since one can always add an irrelevant constant. We now present our second theorem.

\begin{theorem}\label{theorem:local1}
Define the derivative of the local loss function $\mathcal{L}_l$ with respect to the parameter $\theta_k^{(i)}$ by $\partial_k^{(i)} \mathcal{L}_l$. Then $\forall\; \theta_k^{(i)}\in \Theta$, the variance of $\partial_k^{(i)} \mathcal{L}_l$  scales as
\begin{equation}
\label{theorem 2.1}
{\rm Var}(\partial_{k}^{(i)} \mathcal{L}_l) \sim \mathcal{O}\left({\rm Tr}(\hat{O}^2_m)\frac{P(D,d)}{Q(D,d)}\right),
\end{equation}
where $P(D,d)$ and $Q(D,d)$ are certain polynomials of $D$ and $d$ with constant degrees \cite{barren1d_supplementary}.
In addition, the upper bound of the variance of $\partial_k^{(i)} \mathcal{L}_l$ decays exponentially with respect to the distance $\Delta\equiv |i-m|$:
\begin{equation}
\label{theorem 2.2}
{\rm Var}(\partial_{k}^{(i)} \mathcal{L}_l) \leq \mathcal{O} (d^{-\Delta}).
\end{equation}
\end{theorem}

\begin{proof} We sketch the major steps here and leave the technical details to the Supplementary  Materials \cite{barren1d_supplementary}. We use similar techniques as in the proof of Theorem \ref{theorem:global}. By using the $2$-design property of unitary embedded MPS, it is straightforward to prove $\langle \partial_k^{(i)} \mathcal{L}_l\rangle=0$. We then calculate the variance ${\rm Var}(\partial_{k}^{(i)} \mathcal{L}_l)$ and obtain the desired results. Specifically, we first consider the on-site case,  i.e., the derivative and the local operator act on the same physical site. We find that for the connection of the two dangling legs in the term $\int dU_H \otimes_{i=1} ^{N-1} (\delta_{{l_d}, {\bar{l}_d}} U^{(i)}_{r_{d,D},l_{d,D}}\bar{U}^{(i)}_{\bar{r}_{d,D},\bar{l}_{d,D}} |0\rangle_{r_d}^{(i)}|\bar{0}\rangle_{\bar{r}_d}^{(i)})^{\otimes 2}$ yield a non-vanishing term with respect to the system size $n$, which  leads to Eq.~(\ref{theorem 2.1}) after a lengthy calculation. Then we consider the off-site case  
and assume $\Delta \leq \lfloor \frac{n}{2} \rfloor$ since a periodic boundary condition is used. By using the unitary $2$-design, we integrate the unitaries between sites $i$ and $m$ and find that this integration only contributes a $d^{-\Delta}$ decaying term, which leads to the Ineq.~(\ref{theorem 2.2}) and completes the proof. 
\end{proof}

\begin{figure}
\centering
\includegraphics[scale=0.3]{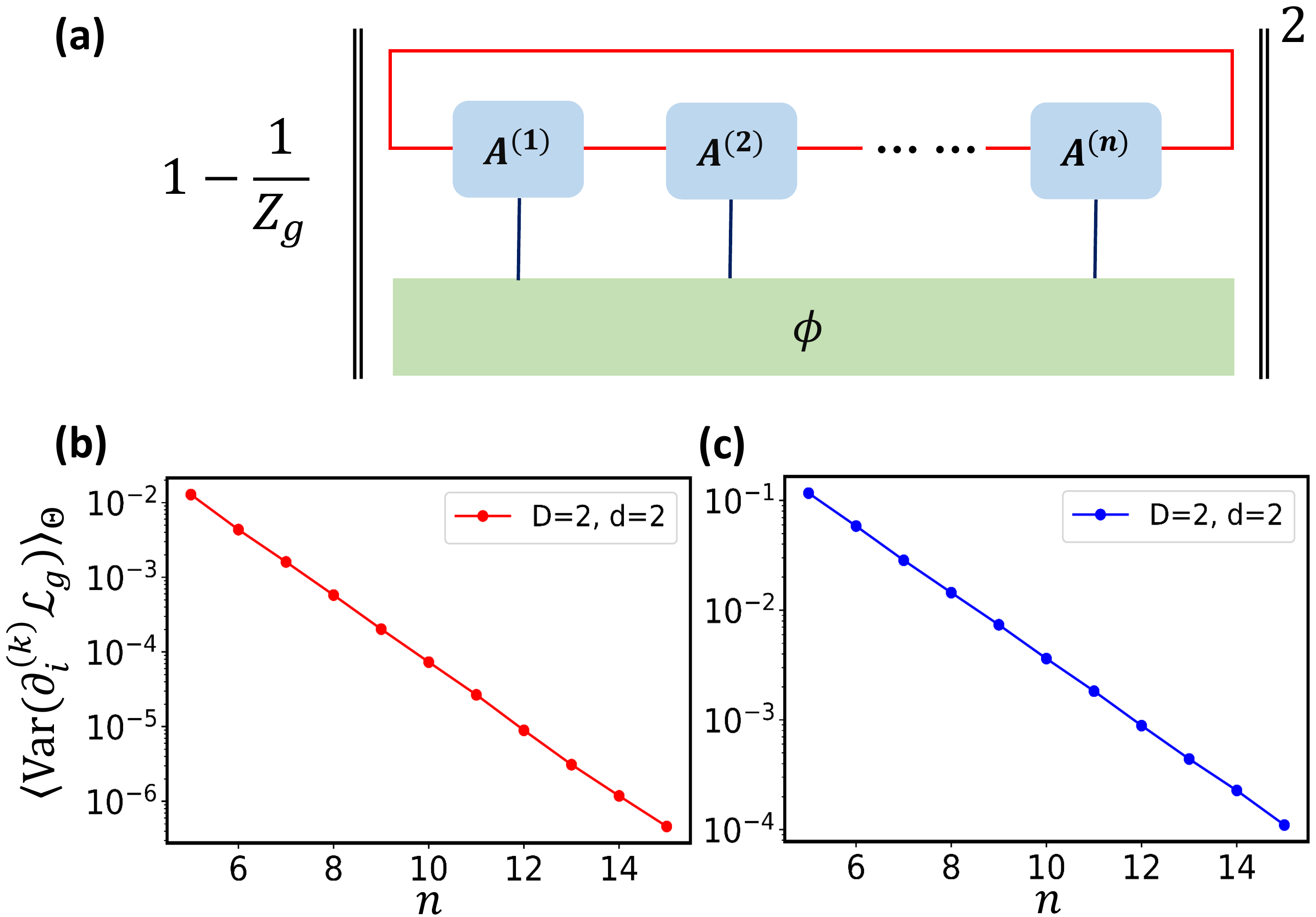}
\caption{Numerical results for different global loss functions. (a) A graphical  representation for the normalized  square  loss $\mathcal{L}_{g,1} = 1-|\langle \phi | \psi(\Theta)\rangle|^2/Z_g$. (b) and (c) plot the average variances of $\partial_k^{(i)} \mathcal{L}_g$ versus the system size $n$ for $\mathcal{L}_{g,1}$ and the KL divergence loss function, respectively. Here, the average is taken over all the variational parameters of the MPS, and we denote it as $\langle \text{Var}(\partial_k^{(i)} \mathcal{L}_g)\rangle_{\Theta}$. We fix both the virtual and physical dimension to $D=d=2$ and the MPS is randomly initialized \cite{barren1d_supplementary}. As the system size $n$ is increasing, $\langle \text{Var}(\partial_k^{(i)} \mathcal{L}_g)\rangle_{\Theta}$  decays exponentially for both loss functions.}
\label{fig:2}
\end{figure}

Theorem~\ref{theorem:local1} implies that the scaling behaviour of the derivatives for local loss functions along any direction is independent of the system size $n$, both upper and lower bounded by a nonzero polynomial fraction of the virtual and physics  dimensions. As a result, there exist no barren plateaus in the landscapes of these functions and the corresponding models are efficiently trainable with gradient-based algorithms. This shows with an analytical guarantee the trainability advantage of choosing local losses for MPS-based learning. Moreover, the second part of this theorem indicates that the derivative is upper bounded by an exponentially small number with respect to $\Delta$. This is attributed to the locality property of MPS. Yet, we stress that this exponential decay cannot be derived directly from the entanglement area law \cite{Eisert2010Colloquium} satisfied by tensor-network states. There exist different quantum states with area-law entanglement but support long-range correlations \cite{Verstraete2006Criticality}. Theorem~\ref{theorem:local1} asserts that although long-range correlations could exist in the variational MPS used, the dependence of the local loss function decays exponentially as $\Delta$ increases. In practice, this would help reduce the computational cost in training MPS-based models, since only the derivatives with respect to nearby parameters play a role and are needed for updating.



{\it Numerical results.}---To verify that the above scaling results are valid for MPS-based learning models with modest system sizes and different loss functions, we carry out some numerical simulations by using the open-source TensorFlow \cite{Abadi2016TensorFlow} and TensorNetwork library \cite{roberts2019tensornetwork}. For the case of global loss, we consider two different kind of loss functions: the normalized square loss defined as $\mathcal{L}_{g,1} = 1-|\langle \phi | \psi(\Theta)\rangle|^2/Z_g$, where $Z_g = \langle \psi(\Theta) | \psi(\Theta) \rangle$ is the normalization factor which concentrate  exponentially  around one \cite{Haferkamp2021Emergent}; and the KL divergence \cite{mackay2003information} defined by $\mathcal{L}_{g,2} = D_{KL}(Q(\phi)||P(\phi,\psi))$, where $P$ and $Q$ denote the probability distributions (see \cite{barren1d_supplementary} for details).  Our results are plotted in Fig. \ref{fig:2}.  From this figure, we see clearly that the variance of the derivatives  along any direction decays exponentially for both global loss functions, with a modest system size ranging from five to fifteen qubits. This is in precise agreement with the analytical results in Theorem \ref{theorem:global} and manifests an undesirable  drawback for choosing global loss functions in MPS-based learning. 
For the case of local loss, we carry out numerical simulations for the loss function defined in Eq.~(\ref{local_loss}) with varying system sizes. Our results are plotted in Fig. \ref{fig:3}.  From Fig. \ref{fig:3} (b), it is evident that the variance of the derivative is independent of the system size, which agrees precisely with the analytical results in Eq. (\ref{theorem 2.1}). In addition,  Fig. \ref{fig:3} (c) clearly shows that the variance decays exponentially with the distance $\Delta$, which verifies the Ineq. (\ref{theorem 2.2}) in Theorem \ref{theorem:local1}.

\begin{figure}
\centering
\includegraphics[scale=0.3]{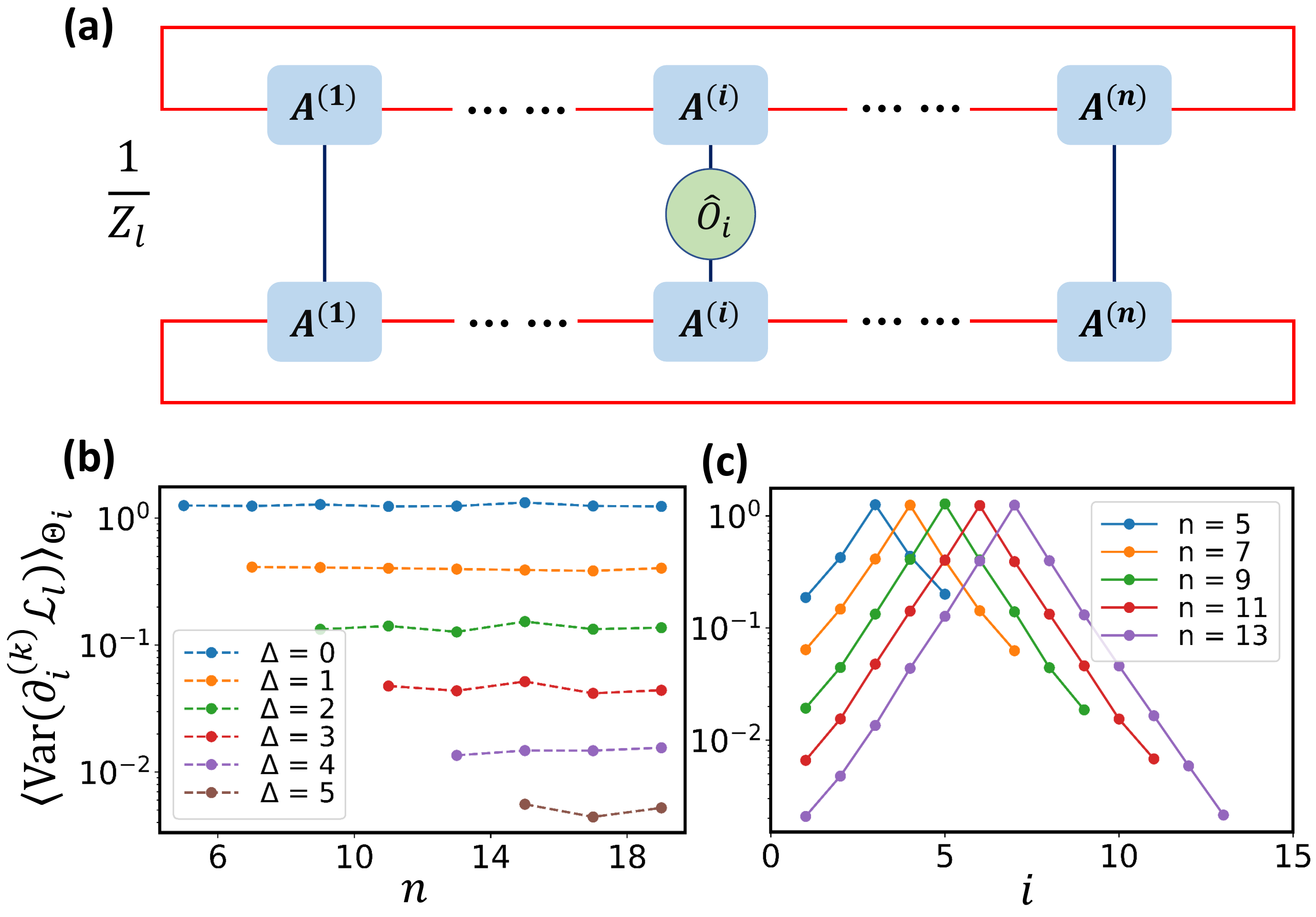}
\caption{
Numerical results for the local loss function defined in Eq. (\ref{local_loss}). (a) A graphical representation of this loss function, where $Z_l = \langle \psi (\Theta) | \psi(\Theta) \rangle$ denotes the normalization factor. (b) The average variance of $\partial_k^{(i)} \mathcal{L}_g$ at the $i$-th site versus different system size $n$. We set the physical  and virtual bond dimensions as $D=d= 2$ and assume the local observable acts at site $m=\frac{n+1}{2}$. The average is taken over all variational parameters at the $i$-th site, and we denote it as $\langle \text{Var}(\partial_k^{(i)}\mathcal{L}_l)\rangle_{\Theta_i}$. (c) The average variance of $\partial_k^{(i)} \mathcal{L}_l$ with fixed $m$ and varying $i$. Here, the parameters are chosen the same as in (b). As the distance $\Delta=|i-m|$ increases,  $\langle \text{Var}(\partial_k^{(i)}\mathcal{L}_l)\rangle_{\Theta_i}$ decays exponentially for different system sizes \cite{barren1d_supplementary}.
}
\label{fig:3}
\end{figure}

{\it Discussion and conclusion.}---Although our discussion is mainly focused on the MPS-based machine learning models, the results obtained in both theorems and the techniques for proving them carry over  straightforwardly  to the more general tensor-network scenario. Yet, it is worthwhile to mention a subtle distinction: for MPS with a fixed bond dimension, computing the local loss functions and derivatives could be very efficient since contraction for MPS is efficient; In sharp contrast, for general tensor networks contraction is $\#$P-complete \cite{Biamonte2015Tensor} and the calculation of loss functions and  derivatives could be exponentially hard. Consequently, although we can prove the absence of barren plateaus for general tensor-network based learning models with local losses, training them may still be infeasible when scaling up due to the intrinsic  difficulty in computing the derivatives. We mention that a variety of strategies, such as specific parameter initialization \cite{Grant2019Initialization} and pre-training \cite{Verdon2019Learning}, have been proposed to escape barren plateaus for variational quantum circuits. These strategies may also apply for tensor-network based learning and a further investigation along this line is highly desirable.


In summary, we have rigorously proved the presence and absence of the barren plateaus in  the landscapes of different loss functions for tensor-network based machine learning models. In particular, we proved that for global loss functions the derivatives along any direction vanish exponentially as the system size increases, giving rise to barren plateaus and rendering their training process inefficient by gradient-based algorithms. While for local loss functions, the gradients with respect to the parameters near the targeting local observables do not vanish and thus no barren plateau occurs. In addition, we proved that the gradient decays exponentially with the distance between the site where the local observable acts on and the site that hosts the derivative parameter. This sheds new light on the locality property of MPS and is of independent interest. We carried out numerical simulations to benchmark the validity of the analytical results for modest system sizes and different loss functions. Our results uncover some critical features for tensor-network based machine learning, which would benefit future studies from both theoretical and practical prospects.

We thank Pan Zhang, Patrick Coles, Lei Wang, Miles Stoudenmire, Xiaopeng Li, Xun Gao, Shunyao Zhang, and Wenjie Jiang for helpful discussions and correspondences. This work was supported by the Frontier Science Center for Quantum Information of the Ministry of Education of China, Tsinghua University Initiative Scientific Research Program, the Beijing Academy of Quantum Information Sciences, the National key Research and Development Program of China (2016YFA0301902), and the National Natural Science Foundation of China (Grants No. 12075128 and No. 11905108). D.-L. D. also acknowledges additional support from the Shanghai Qi Zhi Institute.

The source code for this work will be uploaded to github upon its publication.

\bibliography{BP_ML_MPS,DengQAIGroup}

\clearpage
\onecolumngrid
\makeatletter
\setcounter{figure}{0}
\setcounter{equation}{0}
\renewcommand{\thefigure}{S\@arabic\c@figure}
\renewcommand \theequation{S\@arabic\c@equation}
\renewcommand \thetable{S\@arabic\c@table}

\DeclarePairedDelimiter\ceil{\lceil}{\rceil}
\DeclarePairedDelimiter\floor{\lfloor}{\rfloor}

\newcommand{\taui}{{\tau}^{z}_{i}}
\newcommand{\ZZ}{\mathcal{Z}}
\newcommand{\CC}{\mathbb{C}}
\newcommand{\nn}{\mathbb{N}}
\newcommand{\rr}{\mathbb{R}}
\renewcommand\thefigure{S\arabic{figure}}
\renewcommand\thetable{S\arabic{table}}
\renewcommand\theequation{S\arabic{equation}}
	\newtheorem{corollary}{Corollary}

%
%
%
%
%
%
%
%
%
%
%
\usetikzlibrary{
  shapes,
  shapes.geometric,
	trees,
	matrix,
  positioning,
    pgfplots.groupplots,
  }
\newcommand{\ipict}[3][-0.4]{\raisebox{#1\height}{\scalebox{#3}{\includegraphics{#2}}}}
\newcommand{\ipicg}[3][-0.32]{\raisebox{#1\height}{\scalebox{#3}{\includegraphics{#2}}}}
\hypersetup{
    colorlinks=true, linktocpage=true, pdfstartpage=3, pdfstartview=FitV,
    breaklinks=true, pdfpagemode=UseNone, pageanchor=true, pdfpagemode=UseOutlines,
    plainpages=false, bookmarksnumbered, bookmarksopen=true, bookmarksopenlevel=1,
    hypertexnames=true, pdfhighlight=/O,
    urlcolor=blue, linkcolor=blue, citecolor=black,
}

%
%
%
%
%
%
%
%
%
%
%
%
%
%
%
%
%

\begin{center} 
	{\large \bf Supplementary Material for:  The Presence and Absence of Barren Plateaus in Tensor-network Based Machine Learning}
\end{center}

\section{Details of proof techniques}

\subsection{Unitary designs for the random unitary matrices}
In this section, we introduce the concept of unitary $t$-design. For any measure $dU$ on the unitary group $U(N)$, the $t$-th moments $M_t(dU)$ of $dU$ are defined by the following integral \cite{Renes2004Symmetric,Dankert2009Exact,Harrow2009Random}
 \begin{equation}
 M_t(dU)=\int _{U(N)}dU \prod_{\lambda=1}^t U_{i_\lambda,j_\lambda}\bar{U}_{i'_\lambda,j'_\lambda},
 \end{equation}
 where $U_{i,j}$ represents the $(i,j)$-element of the unitary matrix $U$, and $\bar{U}$ represents the complex conjugation of $U$.
 Generally, the $t$-th moment with respect to the Haar measure $dU_H$ on the unitary group takes the following formula \cite{Collins2006Integration}
\begin{equation}
M_t(dU_H)=\int _{U(N)}dU \prod_{\lambda=1}^t U_{i_\lambda,j_\lambda}\bar{U}_{i'_\lambda,j'_\lambda}=
\sum_{\sigma,\tau}\delta_{i_1,i'_{\sigma(1)}}\cdots\delta_{i_t,i'_{\sigma(t)}} \delta_{j_1,j'_{\tau(1)}}\cdots\delta_{j_t,j'_{\tau(t)}} {\rm Wg}(\tau\sigma^{-1},N),
\end{equation} 
 where ${\rm Wg}(\cdot)$ represents the  Weingarten function, $\sigma$ and $\tau$ are the permutation operators in the symmetric group $S_t$. The  Weingarten function is defined with respect to the permutation operator $\sigma\in S_t$ and the dimension $N$ of the unitary group,
 \begin{equation}
{\rm Wg}(\sigma,N)=\frac{1}{(t!)^2}\sum_{\eta\vdash t, \\ l(\eta)\leq N} \frac{\chi^\eta(1)^2\chi^\eta(\sigma)}{s_{\eta,N}(1)}, 
 \end{equation}
 where the sum is over all the cases of non-negative integer partitions $\eta$ of $t$ with length $l(\eta)\leq N$. Here the partition $\eta$ of integer $t$ (abbreviated by $\eta\vdash t$) means a non-increasing sequence of non-negative integers $\eta=(\eta_1,\eta_2,\dots)$ satisfying $\sum_i\eta_i=t$. The length $l(\eta)$ means the number of non-zero values in the sequence $\eta$.   $\chi^\eta$ represents the irreducible character of the symmetry group $S_t$ with respect to the partition $\eta$, $s_{\eta,d}$ is the Schur polynomial of $\eta$ evaluated at $\underbrace{(1,1,\dots, 1)}_{N}$.
 
 {\bf The Schur polynomial $s_{\eta,d}$.}--- The Schur polynomial $s_{\eta,d}$ at $\underbrace{(1,1,\dots, 1)}_{N}$ is equivalent to the dimension of irreducible representation of $U(N)$ that corresponds to the partition $\eta$, which takes the following formula
 \begin{equation}
     s_{\eta,N}=s_\eta(1,1,\dots,1)=\prod_{1\leq i<j\leq N}\frac{\lambda_i-\lambda_j+j-i}{j-i}.
 \end{equation}
 
 {\bf The irreducible character $\chi^\eta$ of the symmetric group $S_t$.}--- In the case of the identity permutation, the irreducible character value $\chi^\eta(1)$ is equal to the dimension of the irreducible representation of the symmetric group $S_t$ indexed by $\eta$, which is given by the celebrated {\it hook length} formula
 \begin{equation}
     \chi^\eta(1) =\frac{|\eta|!}{\prod_{i,j}h_{i,j}^\eta},
 \end{equation}
 where  $h^\eta_{i,j}$ denotes the hook length of the cell $(i,j)$ in a Young diagram with respect to the partition $\eta$.
 
 In the case of the non-trivial partition $\eta$, the character value $\chi^\eta(1)$ of the symmetric group $S_t$ can be evaluated based on the Murnaghan-Nakayama rule, which is a combinatorial approach that is used for calculating the character.

 {\bf Unitary t-design.}--- The measure $dU$ is called the unitary $t$-design, if the $k$-th moment with respect to $dU$ equals to the $k$-th moment with respect to the Haar measure $dU_H$ for all the positive integers $k\leq t$,
\begin{equation}
    M_1(dU) = M_1(dU_H),\quad M_2(dU) = M_2(dU_H),\,\, \cdots,\,\, M_t(dU) = M_t(dU_H).  
\end{equation}

Based on the theory of the unitary $t$-design, here we calculate the probability distribution of the values of the loss functions associated with the unitary embedding matrix product states (MPS).  To calculate the average and variance of the loss functions with respect to the whole parameter space, we only need to care about the 1-moment and 2-moment of the random unitary matrices. It has been shown that the random unitary matrices form the approximate unitary $2$-design \cite{Harrow2009Random,Haferkamp2021Emergent}.  The approximate unitary 2-design indicates that the first and second moments are approximately the same as the corresponding moments over the Haar measure $dU_H$, i.e., $M_1(dU)=M_1(dU_H)$, $M_2(dU)=M_2(dU_H)$. The first and second moments over the Haar measure are given by the Weingarten functions with the following formula
\begin{eqnarray}
&& M_1(dU_H)=\int_{U(N)}dU_H\;U_{l_0,r_0}\bar{U}_{l_0',r_0'}=\frac{1}{N}\delta_{l_0l_0'}\delta_{r_0r_0'},\label{1-design}\\
&&\begin{aligned}M_2(dU_H)
   =&\int_{U(N)}dU_H\;U_{l_0,r_0}U_{l_1,r_1}{\bar{U}_{l_0',r_0'}}\bar{U}_{l_1',r_1'},\\
=& \frac{1}{N^2-1}(\delta_{l_0l_0'}\delta_{l_1l_1'}\delta_{r_0r_0'}\delta_{r_1r_1'}+\delta_{l_0l_1'}\delta_{l_1l_0'}\delta_{r_0r_1'}\delta_{r_1r_0'})\\
& -\frac{1}{N(N^2-1)}(\delta_{l_0l_0'}\delta_{l_1l_1'}\delta_{r_0r_1'}\delta_{r_1r_0'}
 +\delta_{l_0l_1'}\delta_{l_1l_0'}\delta_{r_0r_0'}\delta_{r_1r_1'}), 
\end{aligned}\label{2-design}
\end{eqnarray}
where $U_{l,r}$ represents the left and right dangling leg of the unitary tensor $U$, and $\bar{U}_{l,r}$ represents the complex conjugation of $U_{l,r}$.  
For the convenient illustrations of the following analytical proofs, we represent the 1-moment and the 2-moment in Eqs.~(\ref{1-design}, \ref{2-design}) by the following graph:
\begin{eqnarray}
&&\ipic{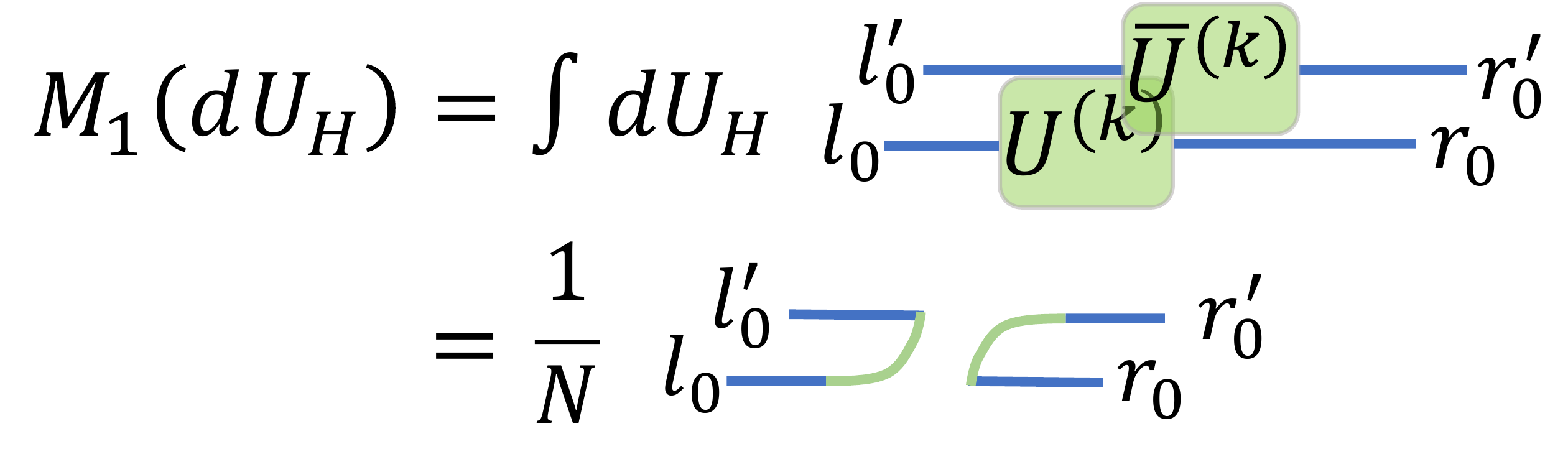}{0.2},\label{graph-1-design}\\
&&\ipic{Unitary_2_design.pdf}{0.2}.\label{graph-2-design}
\end{eqnarray}
In the calculation of the 2-moment in Eqs.~(\ref{2-design}) and (\ref{graph-2-design}),  for the convenience of our latter proofs,  we define the connection $\delta_{l_0 l'_0}\delta_{l_1 l'_1}$ as the symmetric connection (denoted by $S$), and $\delta_{l_0 l'_1}\delta_{l_1 l'_0}$ as the anti-symmetric connection (denoted by $A$). 

\subsection{Parameterization of the unitary embedding matrix product states}
An arbitrary MPS with the periodic boundary condition takes the form,
\begin{equation}
|\psi\rangle=\sum_{i_1, \dots, i_n}\tr\left[A_{i_1}^{(1)}A_{i_2}^{(2)}\cdots A_{i_n}^{(n)}\right]|i_1,\cdots i_n\rangle,
\end{equation}
where $|j_k\rangle$ means the local state of the $k$-th physical site with dimension $d$, $A^{k}_{j_k}$'s are the $D\times D$ matrices, $D$ represents the bond dimension.  Any such MPS can be unitarily embedded by the $Dd\times Dd$ unitary matrices $U$ with the graphical expression \cite{Perez2007Matrix,Gross2010Quantum,Haferkamp2021Emergent}
\begin{equation}\label{eq:stateUprep}
	\ipic{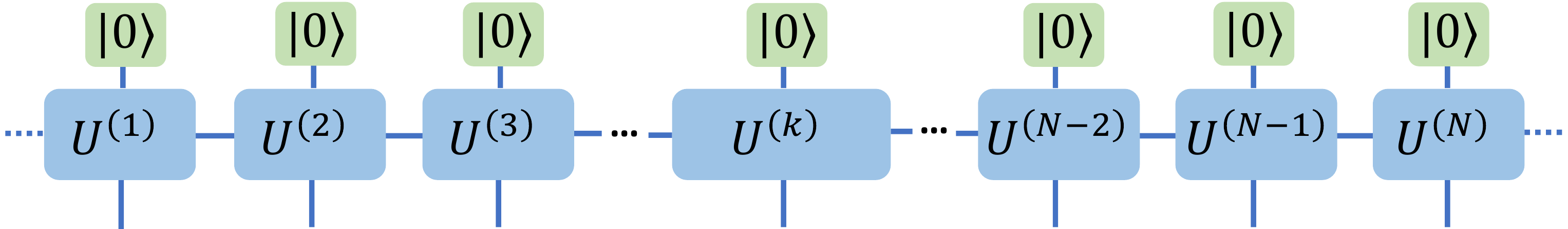}{0.3} .
\end{equation}
The unitary $U^{(i)} \in U(Dd)$ is represented in $\mathbb{C}^{D}\otimes\mathbb{C}^{d}$, where $D$ means the bond dimension, and $|0\rangle\in \mathbb{C}^d$ represents the measurement state in the subsystem $\mathbb{C}^d$.

Here the parameters in the unitary matrices $U^{(i)}$ are randomly initialized. For the convenience of illustrating our results, without loss of generality, here we parameterize the $Dd\times Dd$  matrices  $U^{(i)}$ by the production of a set of exponential unitaries, as
\begin{equation}\label{U_Decomposition}
\begin{aligned}
&U^{(i)}(\bm{\theta}^{(i)})=\prod_{\xi=1}^{{\rm Poly}(Dd)}e^{i\theta_{\xi}^{(i)}G^{(i)}_\xi}=U_{-}^{(i)}U_{+}^{(i)},\\
&U_{-}^{(i)} = \prod_{\xi=1}^{\lambda}e^{i\theta_{\xi}^{(i)}G^{(i)}_\xi},\;
U_{+} ^{(i)}= \prod_{\xi=\lambda+1}^{{\rm Poly}(Dd)}e^{i\theta_{\xi}^{(i)}G^{(i)}_\xi},
\end{aligned}
\end{equation}
where $\bm{\theta}^{(i)}=\{\theta_\xi^{(i)}\}$ represents a set of  real random numbers that parameterize the unitary $U^{(i)}$,  and the set $\{G^{(i)}_\xi\}$ represents a set of Hermitian operators (the number is of order $\mathcal{O}({\rm Poly}(Dd))$) that is required to ensure the universality and randomness of $U^{(i)}$ and $U_-^{(i)}$ (or $U_+^{(i)}$) in the unitary group $U(Dd)$. This setting is inspired by the quantum circuit model proposed in \cite{Mcclean2018Barren},  
 hence we can reasonably make the assumptions that $U^{(i)}$ approximately forms $2$-design, and at least one of the $U_-^{(i)}$ and $U_+^{(i)}$ forms the approximate unitary $2$-design \cite{Harrow2009Random,Mcclean2018Barren}.



\subsection{Cauchy-Schawrz inequality for the  tensor networks}
Here we introduce the Cauchy-Schawrz inequality for the  tensor networks \cite{Haferkamp2021Emergent,Kliesch2019Guaranteed} that will be used in the following proofs.
	\begin{lemma}\label{CS_Inequality}
		Define a tensor network by $(T,C)$  with $J\geq 2$ tensors $T=(t^j)_{\in \{1,\cdots, J\}}$. If no tensor  self-contracts  in the contraction $C$, then the following inequality holds,
		\begin{equation}
		|C(T)|\leq \prod_{j=1}^J \left|\left|t^j\right|\right|_F,
		\end{equation}
		where $||.||_F$ represents the Frobenius norm of the vectorized tensor $t^j$.
	\end{lemma}


\section{Proof of theorem 1}
In this section, we will consider the loss function with the following form:
\begin{equation}
\label{global_loss_eq}
    \mathcal{L}_g = 1-|\langle \psi (\Theta) | \phi \rangle |^2,
\end{equation}
where $|\psi(\Theta) \rangle $ represents the parameterized matrix product state  and $|\phi \rangle$ represents an arbitrary constant quantum state.

According to the graphical representation of the unitary embedded matrix product state in Eq.~(\ref{eq:stateUprep}), we can then represent the loss function of Eq.~(\ref{global_loss_eq}) in the following graphical formula:
\begin{equation}\label{loss_global}
    \mathcal{L}_g = 1 - \ipic{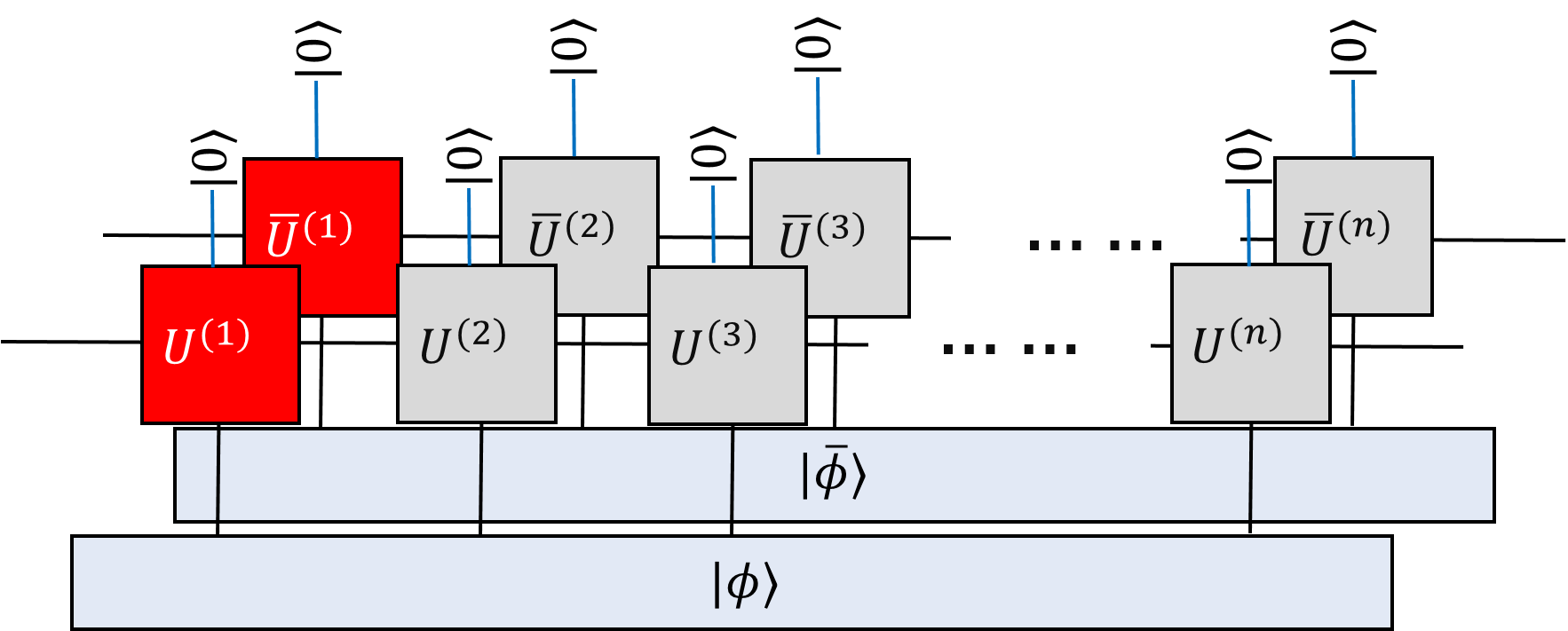}{0.3}.
\end{equation}
Based on the global loss function $\mathcal{L}_g$, we will calculate the mean value and the variance of the gradient $\partial_k\mathcal{L}_g$.
Without loss of generality, here we suppose that the derivative parameter of $\partial_k\mathcal{L}_g$ locates on the first site (denoted by the red color in Eq.~(\ref{loss_global}) ).

\subsection{Mean value of the derivative of global loss function}

Let us first introduce the expectation values of the derivative of the global loss function. For the first site that hosts the derivative parameter, 
we have:
\begin{equation}
\label{derivative_mean}
    \partial_k (U^{(1)} \otimes \bar{U}^{(1)}) = -\text{i} (U_-^{(1)} G^{(1)}_k U_+^{(1)}) \otimes (\bar{U}_-^{(1)} \bar{U}_+^{(1)}) + \text{i} (U_-^{(1)} U_+^{(1)})\otimes (\bar{U}_-^{(1)}\bar{G}^{(1)}_k \bar{U}_+^{(1)}),
\end{equation}
where we express the unitary operator $U^{(1)} = U_-^{(1)} U_+^{(1)} = U_1^{(1)} U_2^{(1)} \dots U_{\text{Poly}(Dd)}^{(1)}$ as in Eq.~(\ref{U_Decomposition}), and $G^{(1)}_k$ is the Hermitian operator of the corresponding derivative unitary operator $U_k^{(1)} = e^{-\text{i} G^{(1)}_k \theta_k}$, $\bar{G}^{(1)}_k$ is the complex conjugation of $G^{(1)}_k$. For latter convenience, here we omit  some unnecessary indices for the notations, such as $U^{(i)}\rightarrow U$,  $\bar{U}^{(i)}\rightarrow \bar{U}$, $G^{(1)}_k\rightarrow G$, and $\bar{G}^{(1)}_k\rightarrow \bar{G}$. We then express the Eq.~(\ref{derivative_mean}) in a compact form:
\begin{equation}
\label{derivative_mean_compact}
    \partial_k (U^{(1)} \otimes \bar{U}^{(1)}) = \sum_{\alpha = 0,1} (-1)^{\alpha} \text{i} U_- G^\alpha U_+ \otimes \bar{U}_- \bar{G}^{1-\alpha} \bar{U}_+,
\end{equation}
where $\alpha$ represents the power of $G$ ($\bar{G}$), i.e., $G^0$ means the identity operator, $G^1=G$. 

 Integrating the right side of Eq.~(\ref{derivative_mean_compact}) with respect to the full unitary space, one can obtain the following graphical expression:
\begin{equation}
\label{derivative_mean_compact_graph}
   \int dU_-\int dU_+\left(\sum_{\alpha = 0,1} (-1)^{\alpha} \text{i} U_- G^\alpha U_+ \otimes \bar{U}_- \bar{G}^{1-\alpha} \bar{U}_+\right) =\text{i}  \ipict{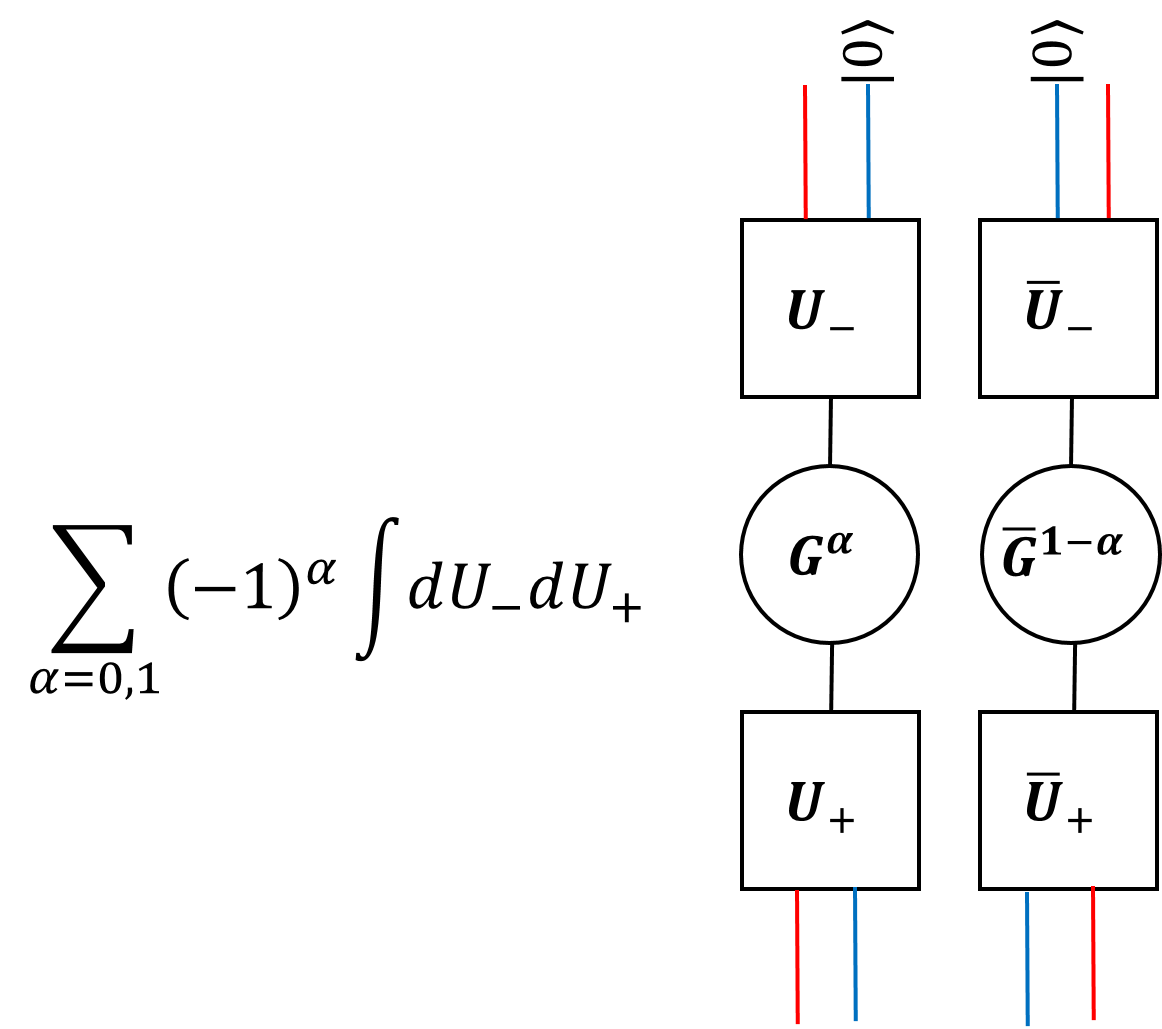}{0.3},
\end{equation}
where the red lines indicate the bond dimension $D$, the blue lines indicate the physical dimension $d$, and the black lines indicate the multiplication of physical dimension and bond dimension $d\times D$. 

Here for the random matrix product states, one can reasonably assume that the unitary embedded operators $\{U^{(i)}\}$ are random enough, where each unitary operator can be treated as a local deep quantum circuit, and it is reasonable to assume that at least one of the $U_-$ and $U_+$ forms $1$-design. Let us consider the first case, where $U_-$ forms the unitary $1$-design. By integrating the $U_-$ operator with respect to the Haar measure $dU_-$ according to the rule in Eq.~(\ref{graph-1-design}), we obtain the following graph:
\begin{equation}
\label{derivative_mean_compact_u-}
    \ipic{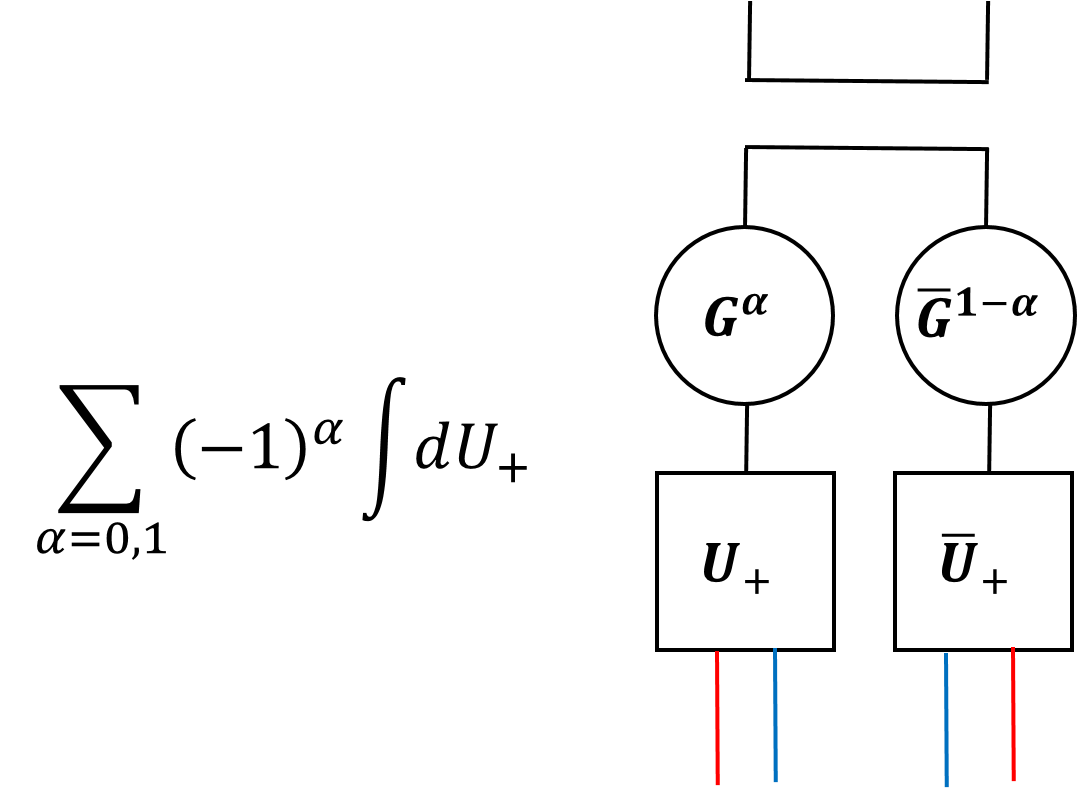}{0.3}.
\end{equation}
We then integrate all the unitary operators from the $2$-th site to the  $n$-th site and ignore the trivial connected parts, the dangling legs  are contracted and Eq.~(\ref{derivative_mean_compact_u-}) becomes:
\begin{equation}
\label{derivative_mean_compact_u-_int}
    \ipic{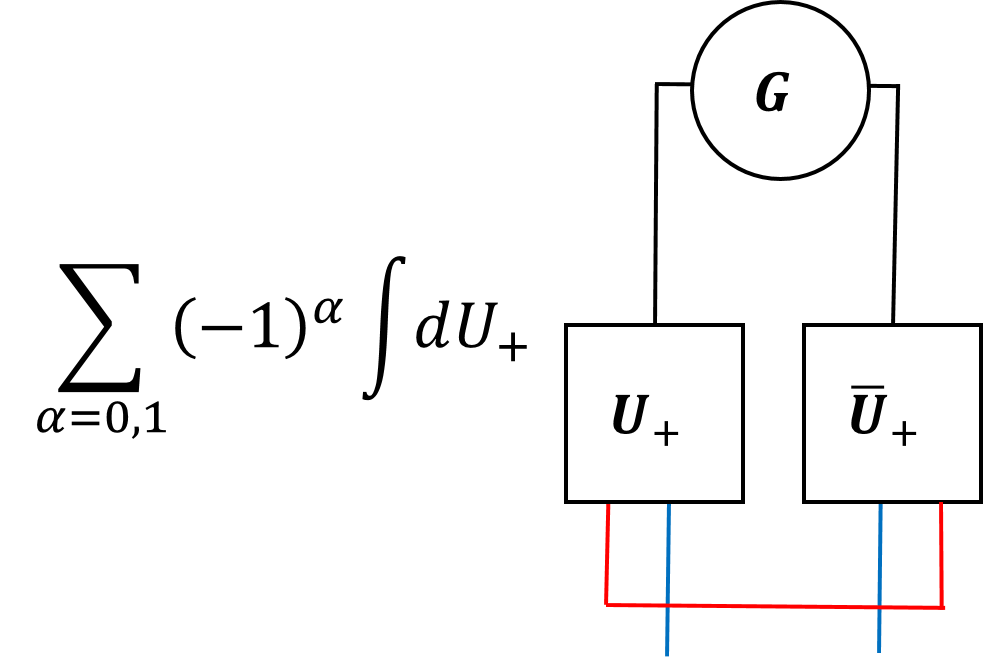}{0.3},
\end{equation}
where we use the following relations:
\begin{equation}
    \ipic{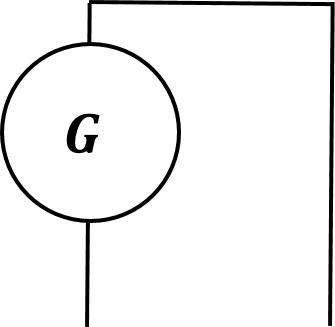}{0.3} = \ipic{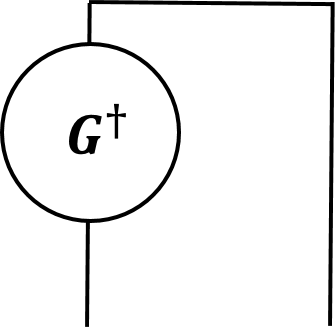}{0.3} = \ipic{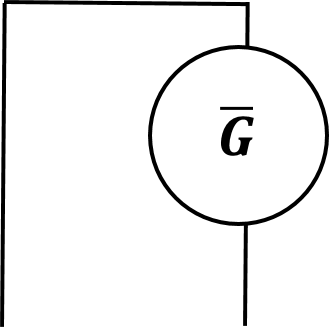}{0.3},
\end{equation}
and ignore the constant term. The dangling blue legs in Eq.~(\ref{derivative_mean_compact_u-_int}) are connected to the constant state $|\phi\rangle$. The integrands for $\alpha = 0$ and $\alpha = 1$ cancel with each other, which means that the expectation value of the derivative of the loss function is $0$.

In the case that $U_+$ forms the unitary 1-design, one can integrate the $U_+$ part with respect to the full unitary space. Then the integration of  the derivative site (red-colored site in Eq.~(\ref{loss_global}) ) becomes:
\begin{equation}
\label{derivative_mean_compact_u+}
        \ipic{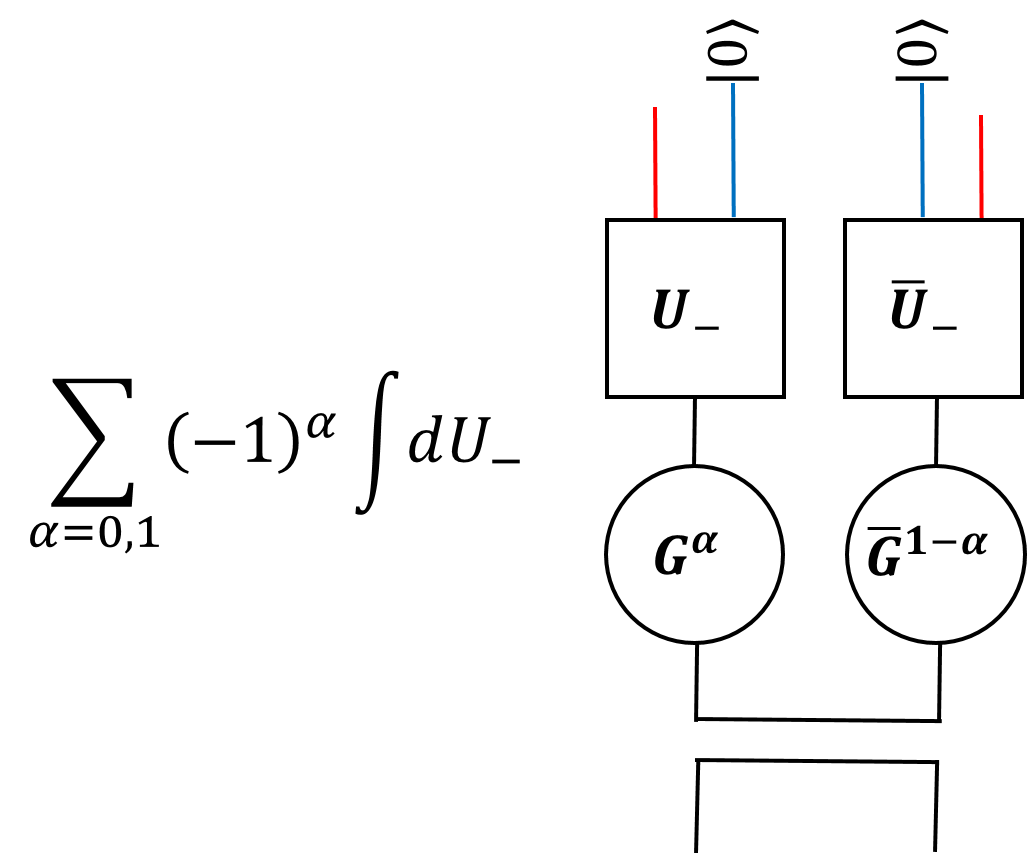}{0.3}.
\end{equation}
Similarly, by integrating all the unitary operators from the $2$-th site to the $n$-th site and ignore the trivial connected parts, the dangling legs  are contracted and Eq.~(\ref{derivative_mean_compact_u+}) becomes:
\begin{equation}
\label{derivative_mean_compact_u+_int}
    \ipic{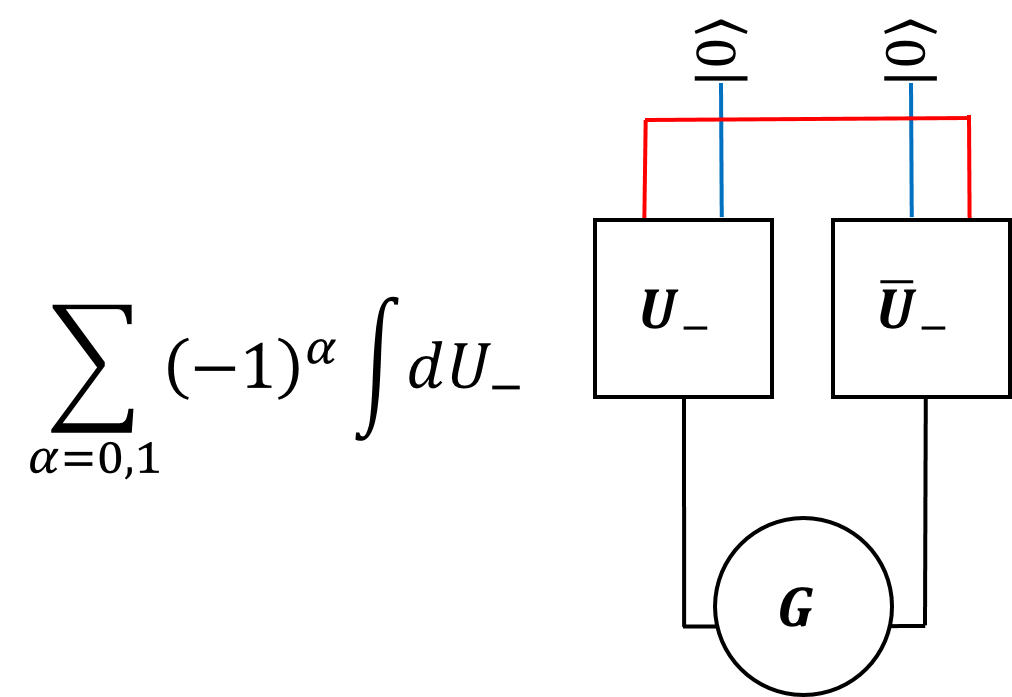}{0.3}.
\end{equation}
One can easily verify the $\alpha =0$ term and the $\alpha=1$ term cancel each other out, then Eq.~(\ref{derivative_mean_compact_u-_int}) equals to $0$. Thus,we conclude that the mean value of the derivative of the global loss function 
\begin{equation}\label{average_global_0}
\langle \frac{ \partial\mathcal{L}_g}{\partial \theta_k}\rangle =0,
\end{equation}
given that at least one of the two unitary operators  $\{U_-, U_+\}$ forms the unitary 1-design.




\subsection{Variance of the derivative of global loss function}
\label{variance_derivative_glf}
The variance of the derivative of the global loss function can be written as:
\begin{equation}
    \text{Var}(\partial_k \mathcal{L}_g) = \langle (\frac{ \partial\mathcal{L}_g}{\partial \theta_k})^2 \rangle  - \langle \frac{ \partial\mathcal{L}_g}{\partial \theta_k}\rangle^2.
\end{equation}
With the vanishing expectation value $\langle \frac{ \partial\mathcal{L}_g}{\partial \theta_k}\rangle=0$, the variance can be further reduced to
\begin{equation}
        \text{Var}(\partial_k \mathcal{L}_g) = \langle (\frac{\partial \mathcal{L}_g}{\partial \theta_k})^2 \rangle,
\end{equation}
which indicates that we only need to consider the mean square of the derivative of the loss function. For our convenient, we express the Weingarten calculus of the second moment operator of the Haar-random unitary in the following graph representation: 
\begin{equation}
\label{uniform_2design.png}
        \ipic{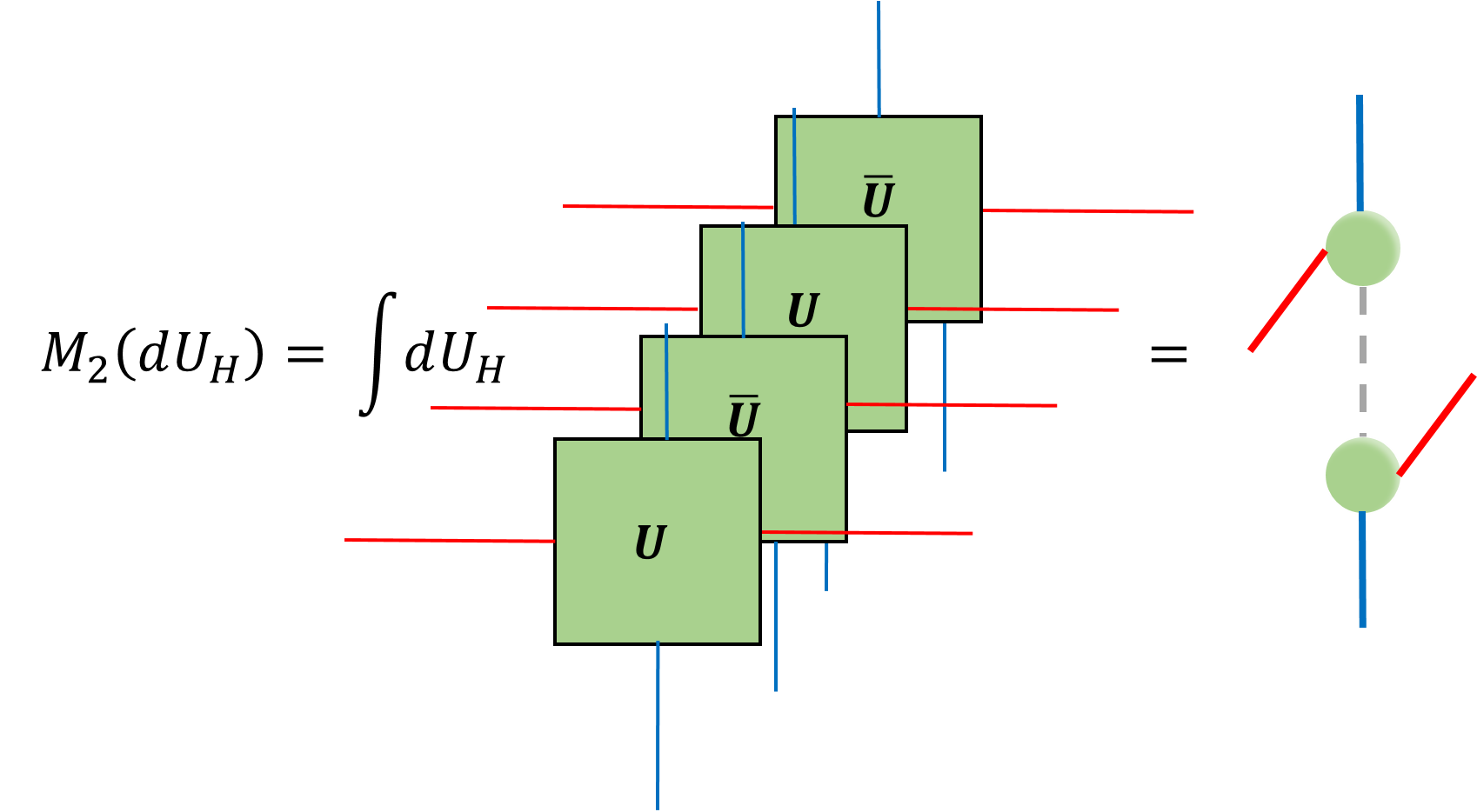}{0.3},
\end{equation}
where the red and blue lines on the right side of the equation represent a compact form of the four dangling legs with the bond dimension and the physical dimension respectively. The green circle denotes the sum of all
possible connections of the four dangling legs and has the following two cases (obtained from the calculation of 2-moment in Eq.~(\ref{graph-2-design}) )
\begin{equation}
        \ipic{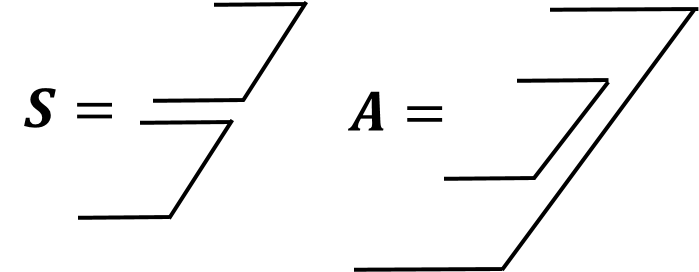}{0.3}.
\end{equation}
where $S$ denotes the case of the symmetric connection, and $A$ denotes denotes the case of anti-symmetric connection, so that $\ipicg{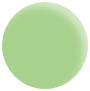}{0.32}=S+A$. The grey dashed line in Eq.~(\ref{uniform_2design.png}) corresponds to the weight of each connection, and has the following four cases:
\begin{equation}
\begin{aligned}
    \ipic{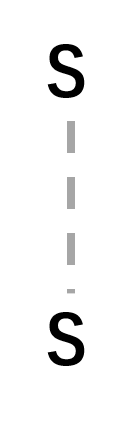}{0.3} &= \ipic{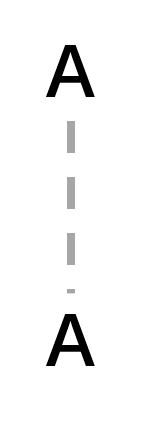}{0.3} = \frac{1}{(Dd)^2 - 1},\\
    \ipic{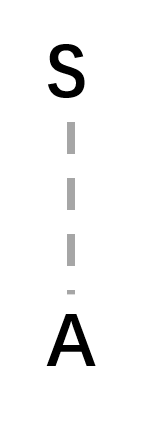}{0.3} &= \ipic{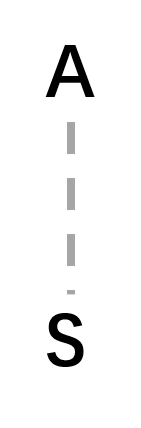}{0.3} = -\frac{1}{Dd\left[(Dd)^2 - 1\right]}.\\
\end{aligned}
\end{equation}


Now we calculate the variance of the derivative of the global loss function.  
The square of the  derivative of the loss function has the following formula:
\begin{equation}
    \ipic{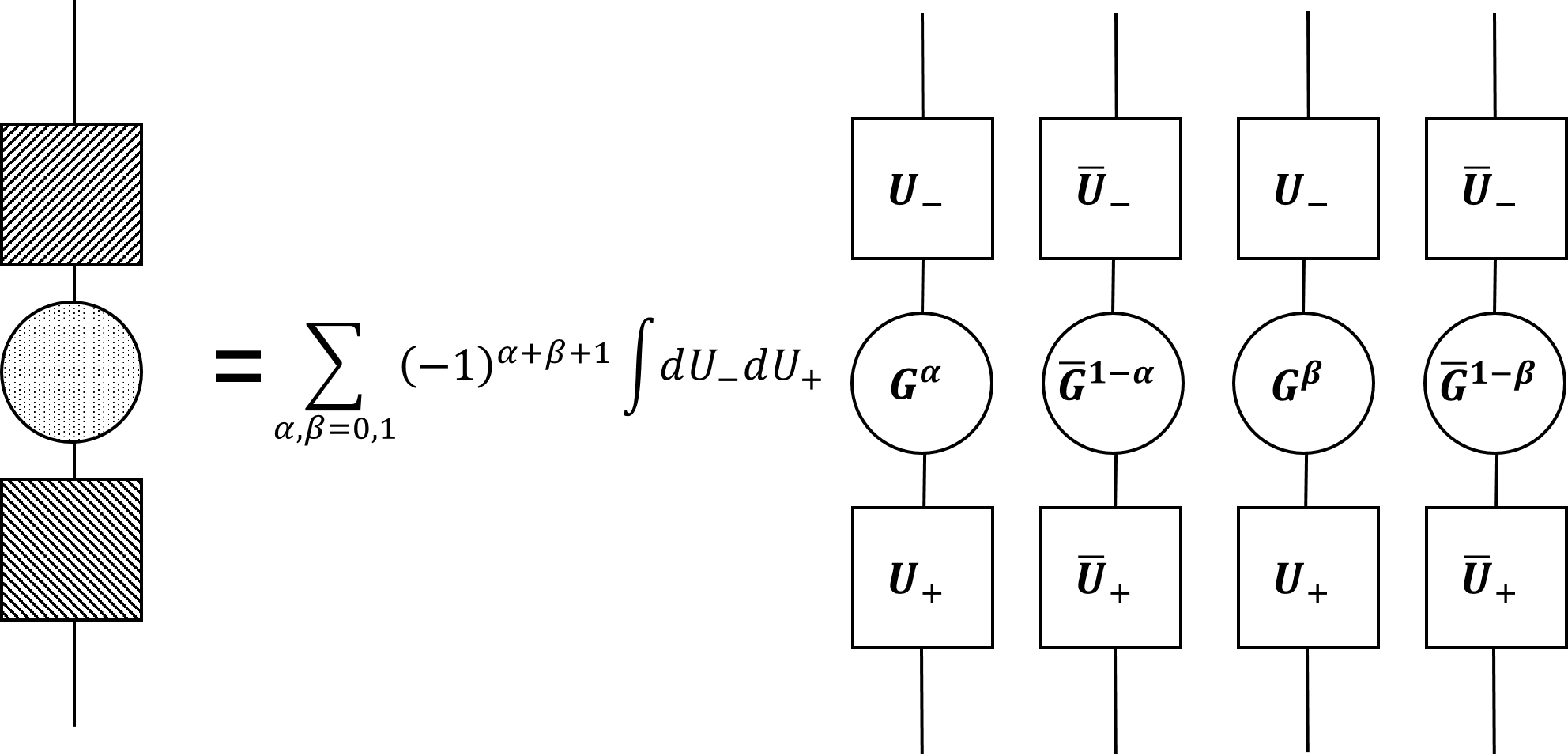}{0.3}.
\label{quadratic_derivative}
\end{equation}
In the following proof, based on the fact that the random unitary matrices form the approximate unitary $2$-design, here we consider the case that at least one of the two unitary operators $\{U_-,U_+\}$ in Eq.~(\ref{quadratic_derivative}) forms the unitary $2$-design.

To calculate the term $\langle (\frac{ \partial\mathcal{L}_g}{\partial \theta_k})^2\rangle$, 
we first consider the case that  $U_-$ forms the unitary $2$-design.  By integrating the randomly initialized unitary matrices on all sites, we obtain the following graphical result:
\begin{equation}
    \ipic{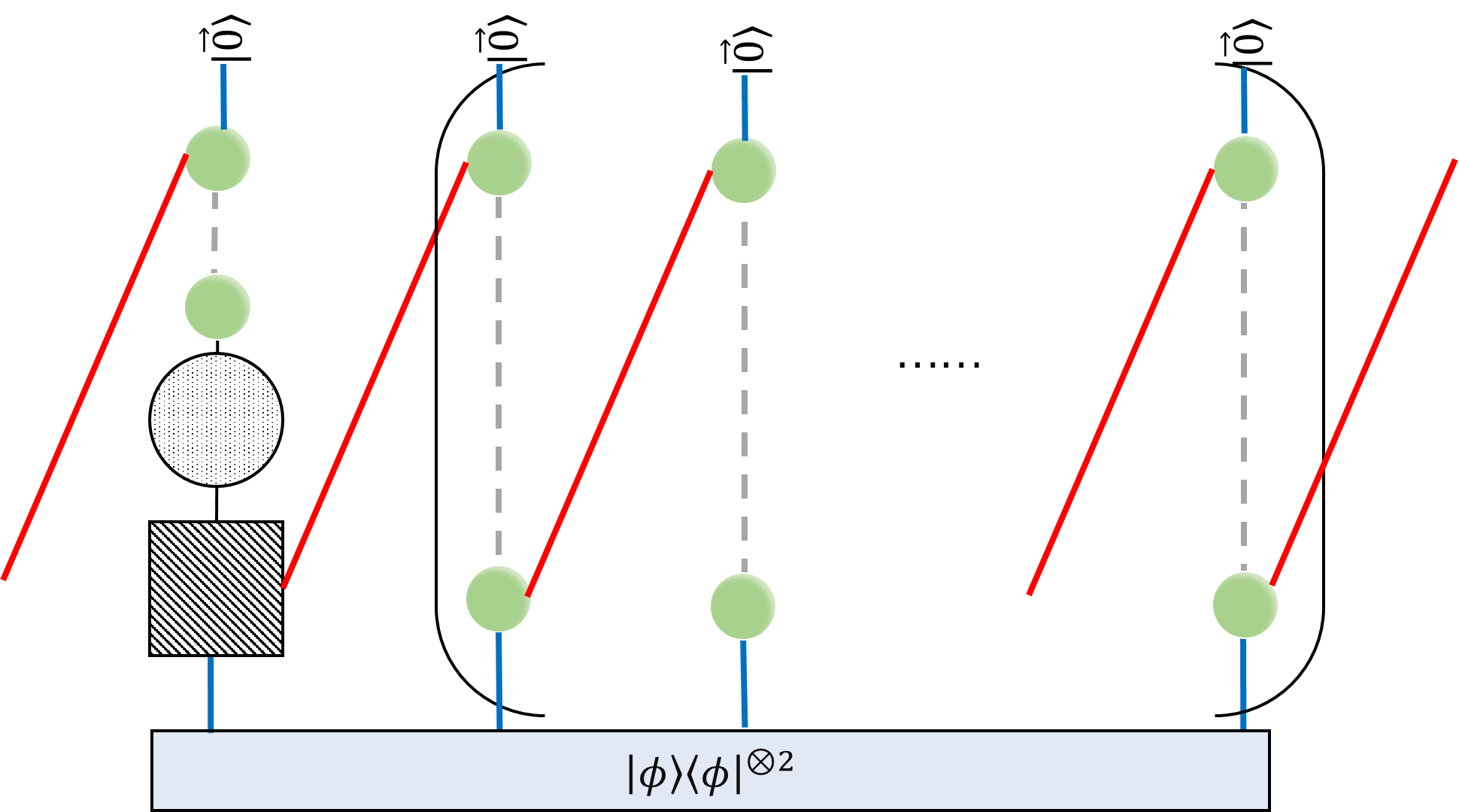}{0.3},
\label{graph_global}
\end{equation}
where the term on the derivative site takes the form:
\begin{equation}\label{u-symbol}
    \ipic{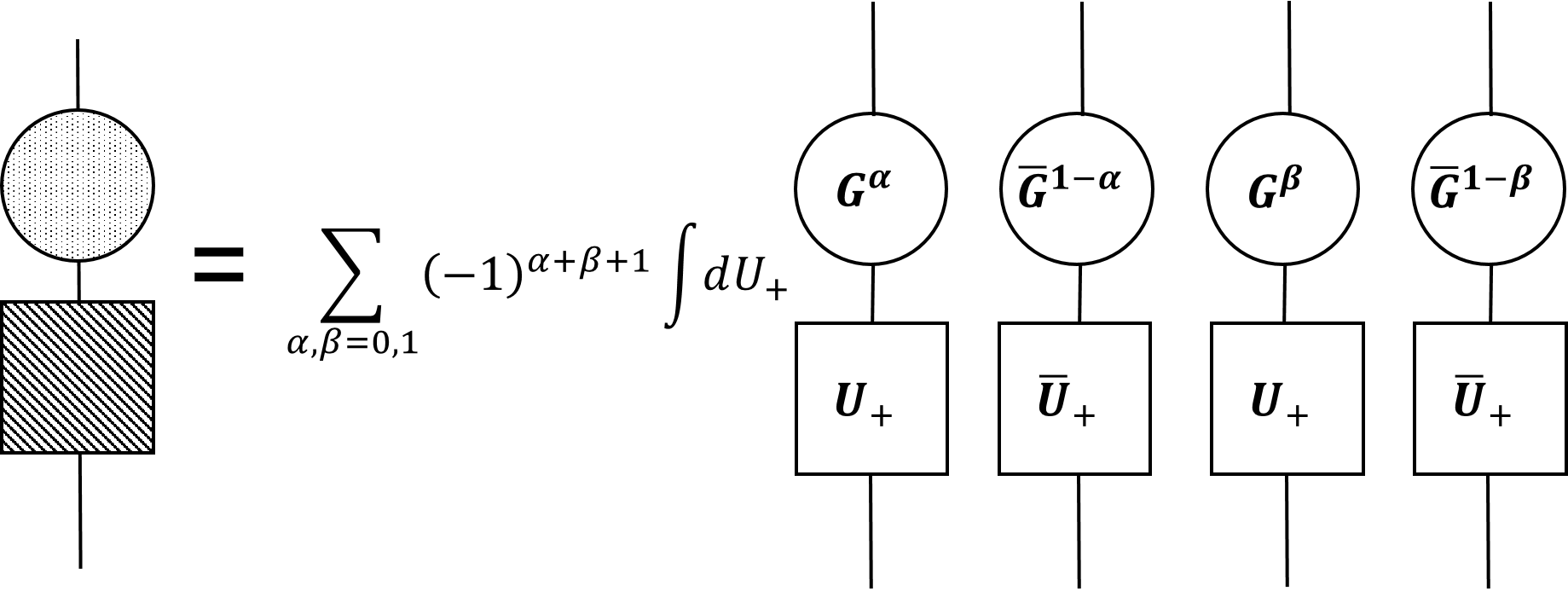}{0.3},
\end{equation}
and we denote the reference state by $|\vec{0}\rangle = |0\rangle \langle 0 | ^{\otimes 2}$ in Eq.~(\ref{graph_global}) based on the Choi-Jamiolkowski isomorphism \cite{jamiolkowski1972linear,choi1975completely}. 

With different connection of the dangling legs, we have the following relations:
\begin{equation}
\begin{aligned}
           \ipic{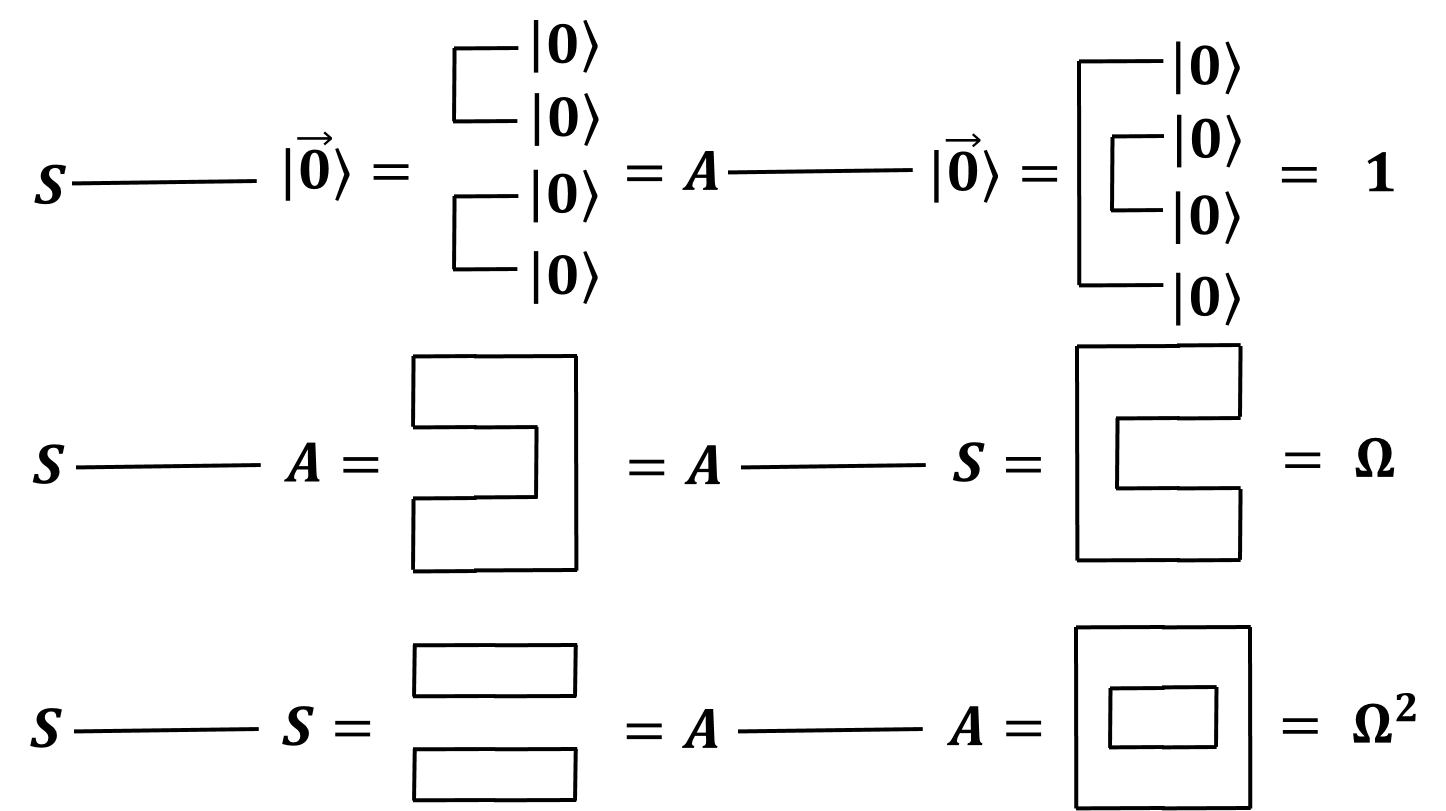}{0.3} ,
\end{aligned}
\end{equation}
where $\Omega$ is the dimension of the dangling leg. 

As was mentioned in Lemma~\ref{CS_Inequality}, the Cauchy-Schwarz inequality for tensor networks reads:
\begin{equation}
    \sum_{\Pi = {S,A}}|\langle \phi|^{\otimes 2}(\langle\bar{ \phi} |^{\otimes 2}) |\Pi\rangle| \leq 1,
\end{equation}
where $|\bar{\phi}\rangle$ is complex conjugation of $|\phi\rangle$, and $\Pi $ is an arbitrary choice of all the possible combinations of the $S$ and $A$ connections. 

With the above result of Cauchy-Schwarz inequality for tensor networks, we  easily obtain that the graphic structure of the integrated result is upper bounded by:
\begin{equation}\label{global_u-2design_1}
\ipic{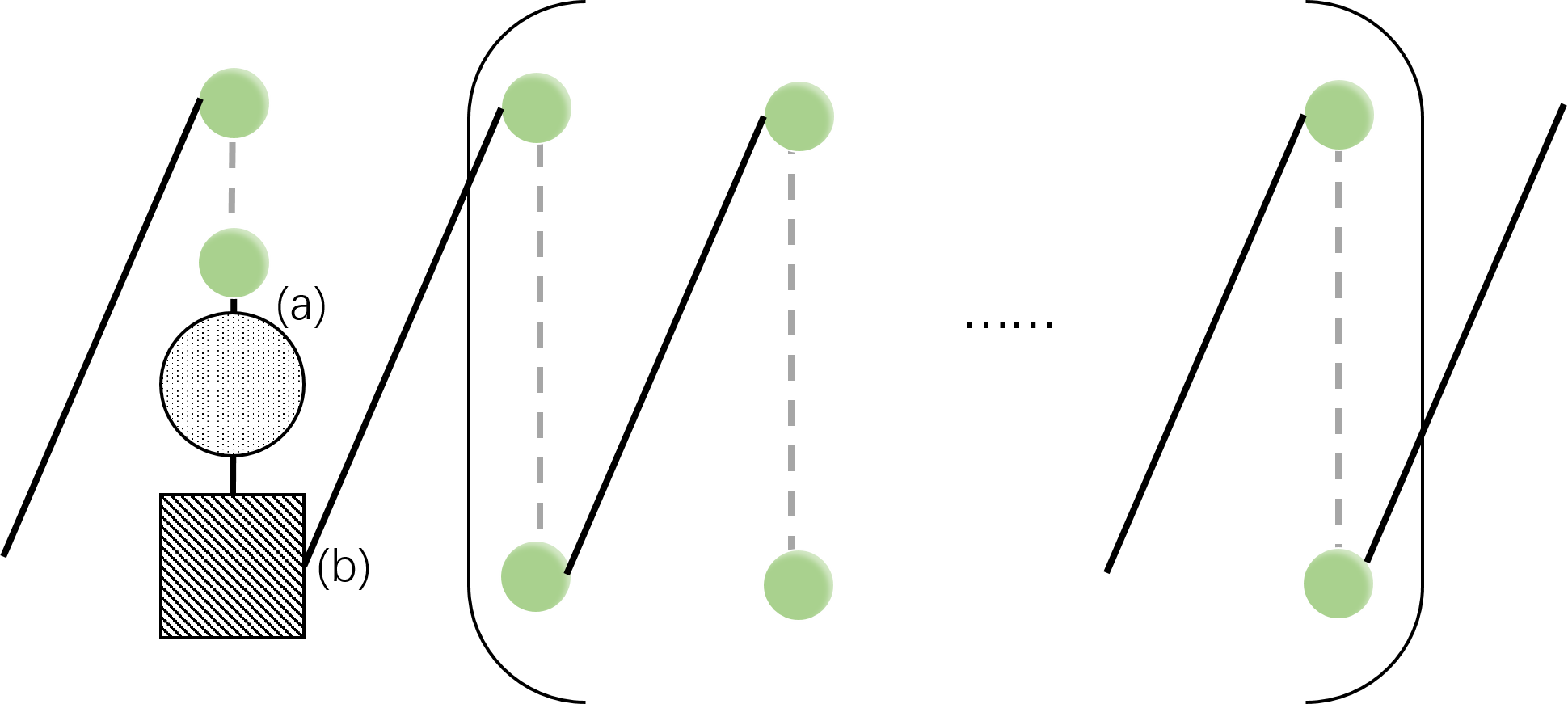}{0.3},
\end{equation}
where (a) denotes the possible connection of the four lines (noted in Eq.~(\ref{u-symbol})) between the green circle and the grey circle, and (b) denotes the possible connection of the four lines between the green circle and the black square.

Now we consider all of the possible cases for (a) and (b).

If (a) hosts the $S$ connection, we can graphically express the derivative site as:
\begin{equation}
\label{global_u-S_result}
\sum_{\alpha ,\beta= 0,1} (-1)^{\alpha + \beta +1} \int dU_+\ipic{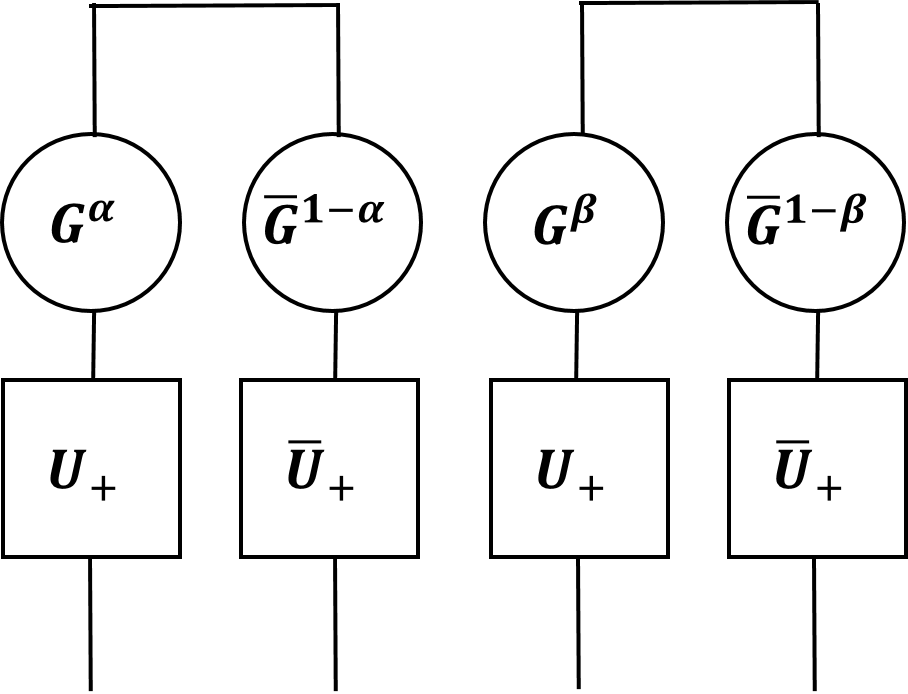}{0.3} =\sum_{\alpha ,\beta= 0,1} (-1)^{\alpha + \beta +1} \int dU_+ \ipic{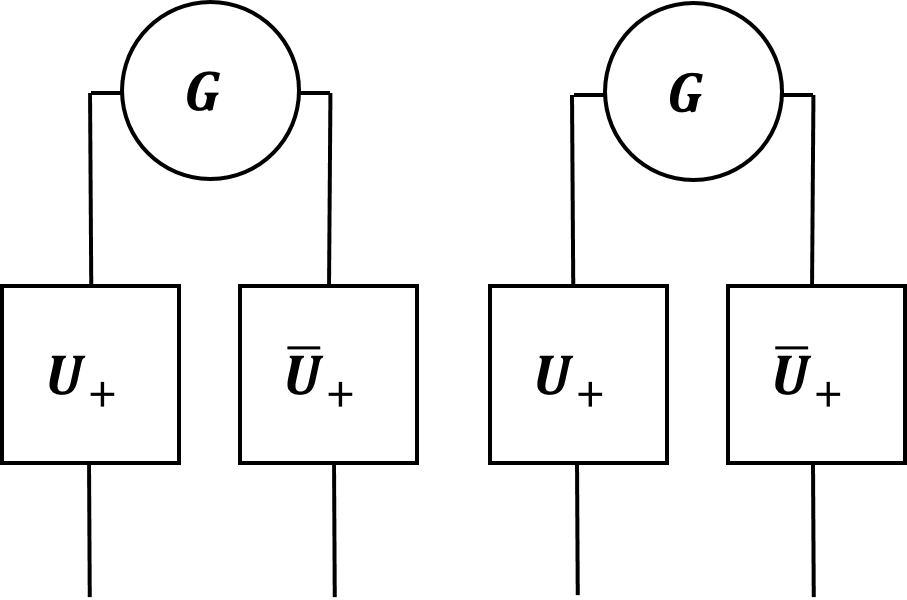}{0.3},
\end{equation}
where the integrated term is independent of the indices $\alpha$ and $\beta$. By summing over all the cases of $\alpha$ and $\beta$, we easily obtain that the value of Eq.~(\ref{global_u-S_result}) is $0$.

If (a) hosts the $A$ connection and (b) hosts the $S$ connection, the integration on the derivative site can be graphically expressed as:
\begin{equation}
\label{global_u-AS_result}
\sum_{\alpha ,\beta= 0,1} (-1)^{\alpha + \beta +1} \int dU_+\ipic{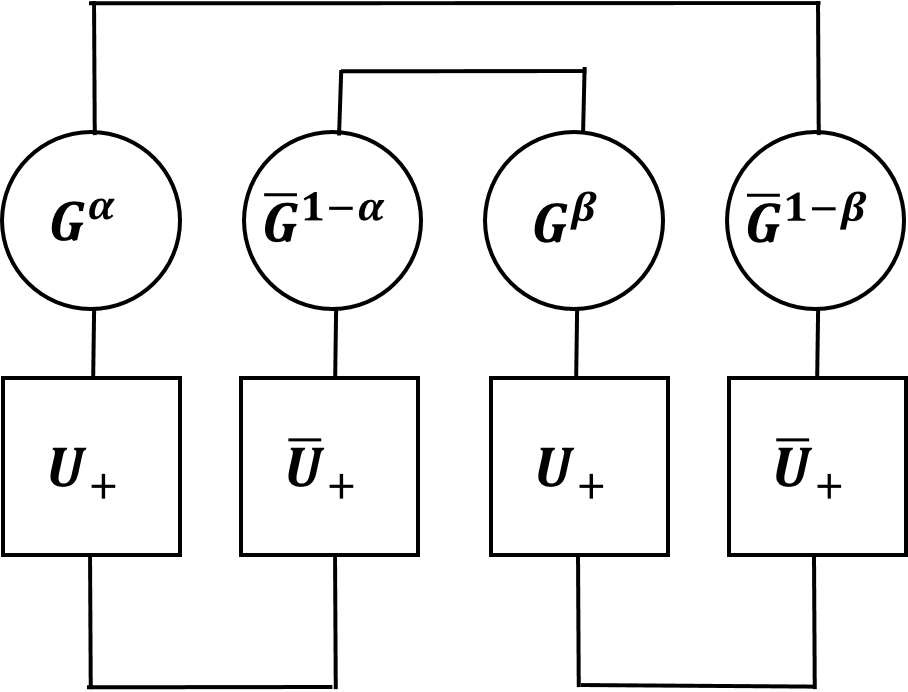}{0.3} =\sum_{\alpha ,\beta= 0,1} (-1)^{\alpha + \beta +1} \int dU_+ \ipic{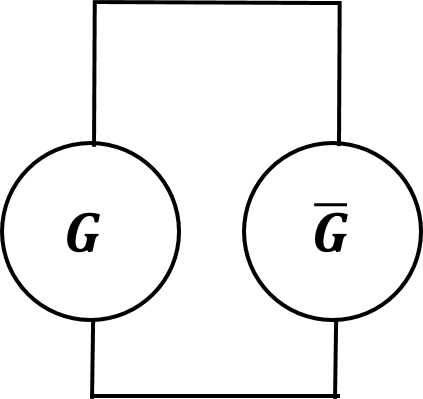}{0.3}.
\end{equation}
Similar to the case of Eq.~(\ref{global_u-S_result}), simple calculations show that the value of Eq.~(\ref{global_u-AS_result}) is $0$.

If both (a)  and (b) host the $A$ connections, the integrated graph on the derivative site is
\begin{equation}
\label{global_u-AA_result}
\sum_{\alpha ,\beta= 0,1} (-1)^{\alpha + \beta +1} \int dU_+\ipic{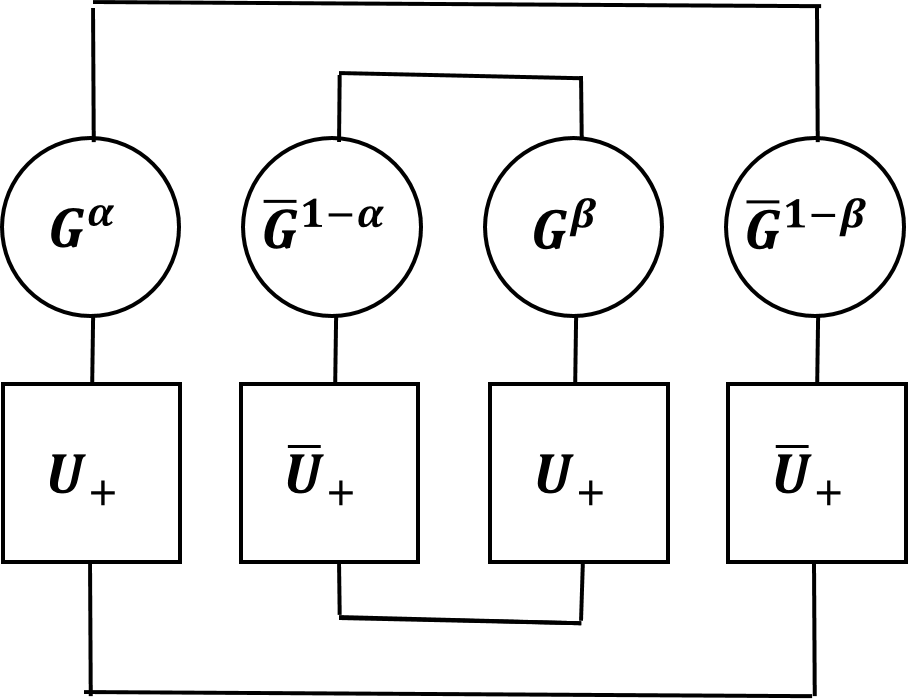}{0.3} =\sum_{\alpha ,\beta= 0,1} (-1)^{\alpha + \beta +1} \int dU_+ \ipic{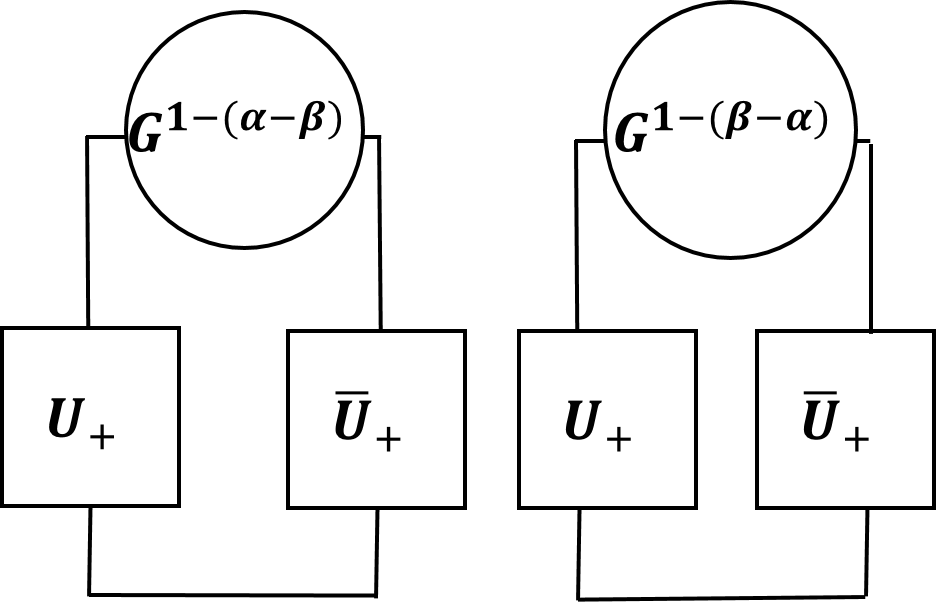}{0.3},
\end{equation}
which leads to a non-zero but constant term:
\begin{equation}
    C_{g1} = 2\text{Tr}(G^2) - 2\text{Tr}(G)^2. 
\end{equation}
Based on the results of Eqs.~(\ref{global_u-2design_1}, \ref{global_u-S_result}, \ref{global_u-AS_result}, \ref{global_u-AA_result}), and together with the relations  that $ A \rule[0.25\baselineskip]{0.75cm}{0.05em} A =  S \rule[0.25\baselineskip]{0.75cm}{0.05em} S  =  (Dd)^2 $ and $ S \rule[0.25\baselineskip]{0.75cm}{0.05em} A =  A \rule[0.25\baselineskip]{0.75cm}{0.05em} S  =  Dd $, then we can bound the variance $\text{Var}(\partial_k \mathcal{L}_g)$ by the following equations:
\begin{equation}\label{figure36}
\begin{aligned}
\text{Var}(\partial_k \mathcal{L}_g)=\langle (\partial_k \mathcal{L}_g)^2\rangle  \leq & \left|\ipic{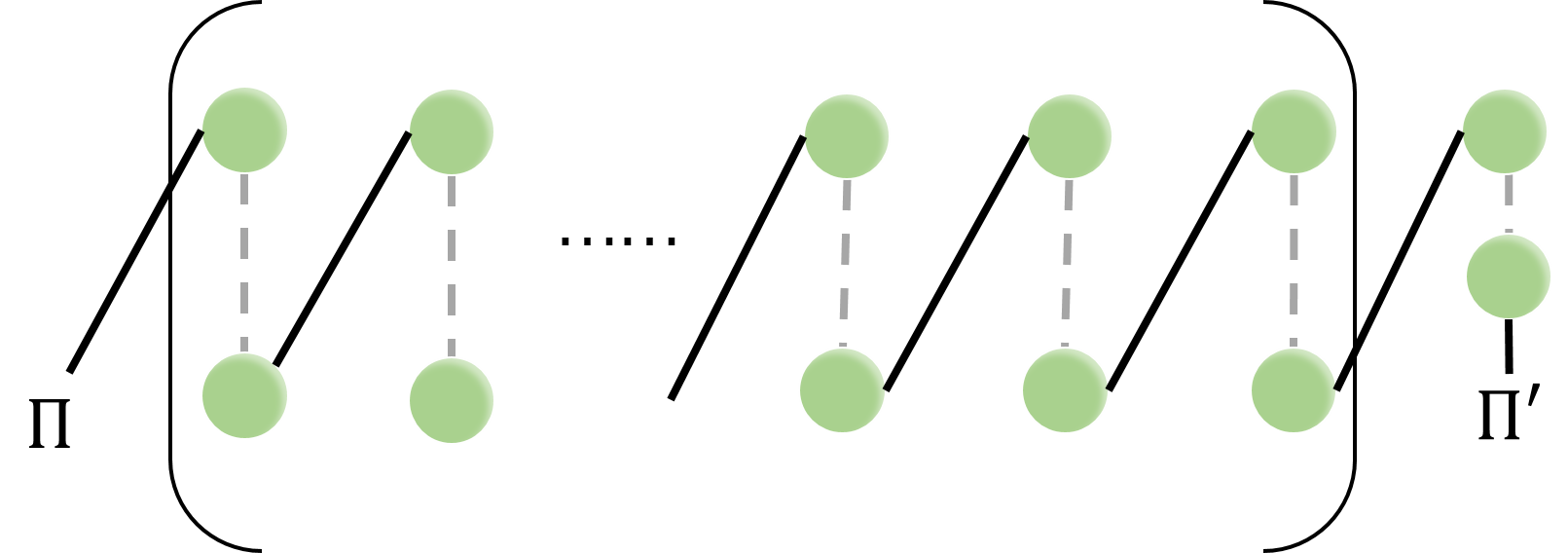}{0.3}\right| C_{g1}\\
\leq &2\frac{D d -1}{[(Dd)^2 -1]Dd} \frac{D^{2n}}{(d^2D^2-1)^n}\sum_{l=0}^{n}\sum_{r=0}^n{n\choose l}{n\choose r}D^{-l-r}d^{-r}C_{g1}  \\
= & 2 \frac{D d -1}{[(Dd)^2 -1]Dd} \frac{(1+\frac{1}{D})^{n-1}(1+\frac{1}{Dd})^{n-1}}{(d^2 - \frac{1}{D^2})^{n-1}} C_{g1},
\end{aligned}
\end{equation}
where $\Pi$ and $\Pi'$ are arbitrarily chosen from the $S$ and $A$ connections.

Now we consider the case that $U_+$ forms the unitary $2$-design. Following the Cauchy-Schwarz inequality in Lemma~\ref{CS_Inequality},  then the  Eq.~(\ref{graph_global}) is upper bounded by:
\begin{equation}
\ipic{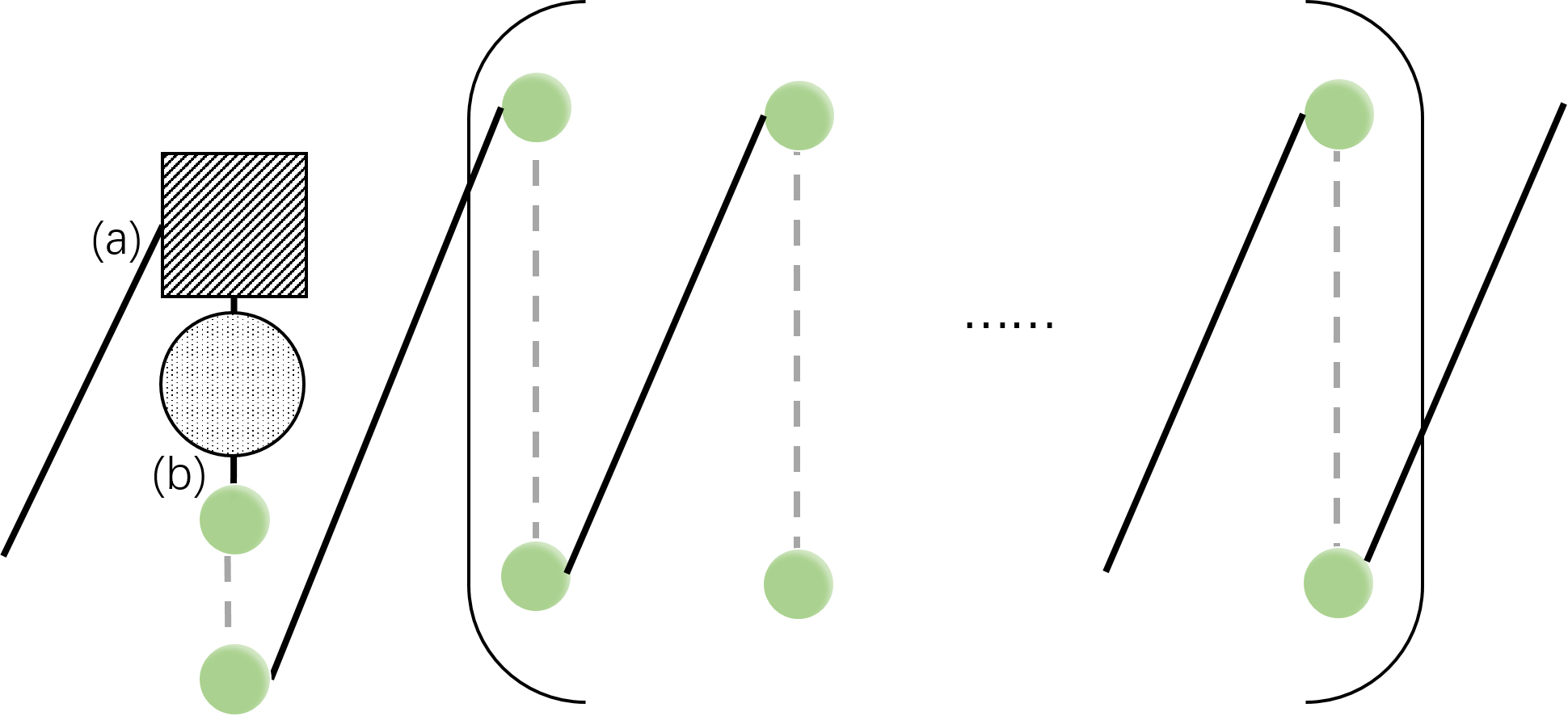}{0.3},
\end{equation}
where (a) denotes the possible connection of the four dangling lines  on the black square, and (b) denotes the possible connection of the four lines between the green circle and the grey circle.  The site hosting the derivative parameter is expressed as follows:
\begin{equation}\label{u+symbol}
    \ipic{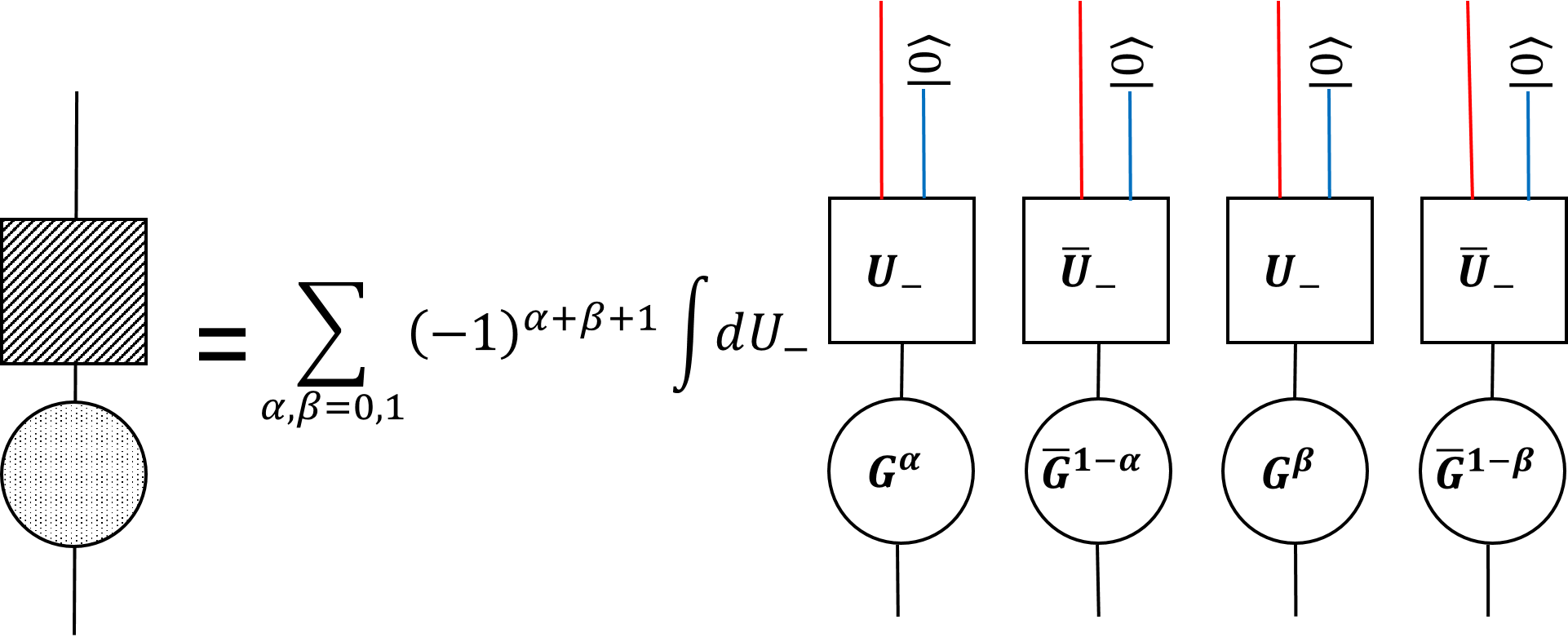}{0.3}.
\end{equation}
Similar to our previous discussion, when (a) hosts the $S$ connection  and (b) hosts the $A$ connection, the graph in Eq.~(\ref{u+symbol}) is depicted as:
\begin{equation}
    \ipic{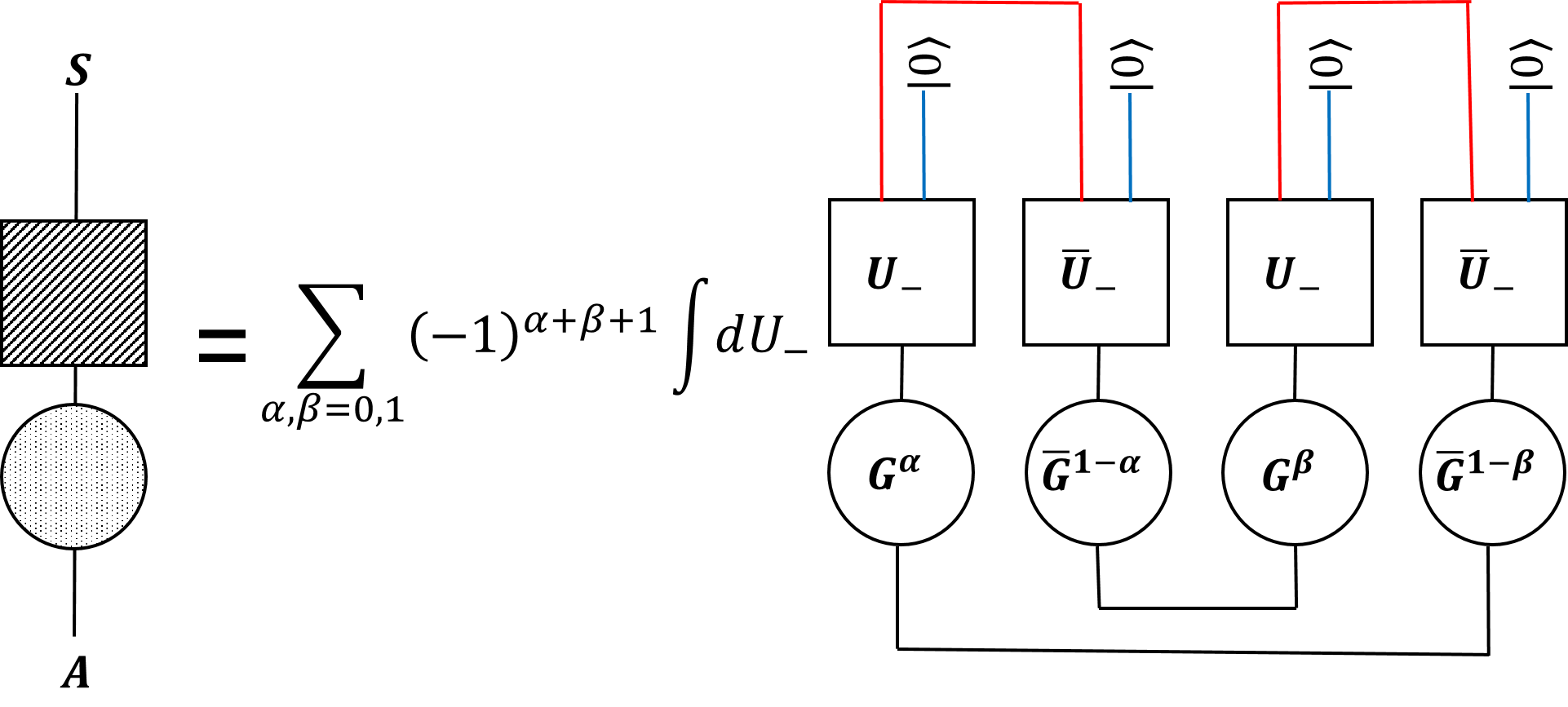}{0.3},
\end{equation}
which leads to a non-zero constant value:
\begin{equation}
    C_{g2} = \int dU_- 2[\text{Tr}(\rho G^2 \rho) - \text{Tr}((\rho G)^2)],
\end{equation}
where $\rho = U_- |0\rangle \langle 0| U_-^\dagger$.

When both  (a) and (b) hosts the $A$ connections, the graph in Eq.~(\ref{u+symbol}) is depicted as:
\begin{equation}
    \ipic{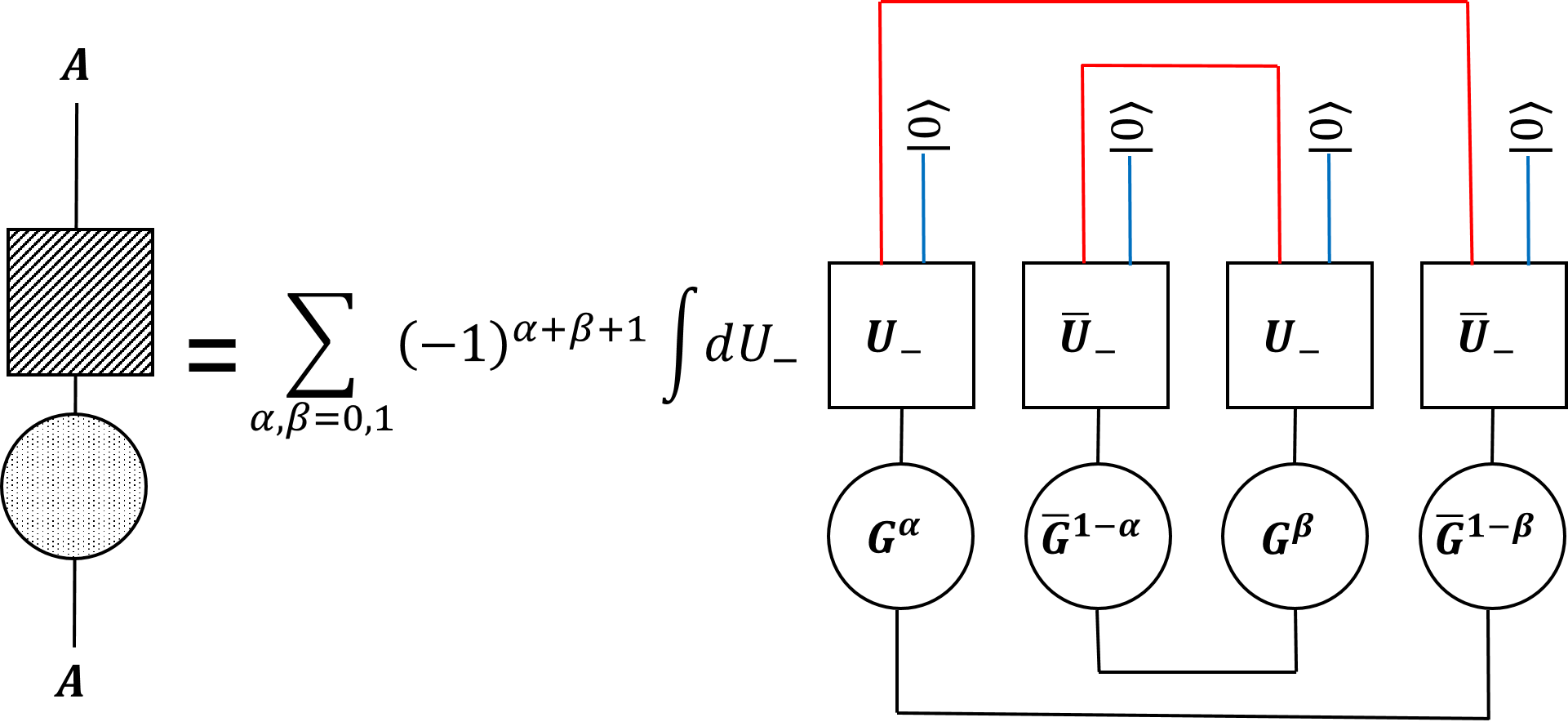}{0.3}
\end{equation}
which leads to a non-zero constant value:
\begin{equation}
    C_{g3} = \int dU_- 2[\text{Tr}(G^2 \rho) Dd - \text{Tr}(G\rho)^2].
\end{equation}
Thus the variance of the derivative of the global loss function  $\text{Var} (\partial_k \mathcal{L}_g)$  is bounded by:
\begin{equation}\label{var_global_g2_g3}
   \text{Var} (\partial_k \mathcal{L}_g)= \langle (\partial_k \mathcal{L}_g)^2\rangle \leq f(n)  \text{Max}(C_{g2}, C_{g3})
\end{equation}
with a $n$-dependent function 
\begin{equation}
    f(n) =  2 \frac{D d -1}{[(Dd)^2 -1]Dd} \frac{(1+\frac{1}{D})^{n-1}(1+\frac{1}{Dd})^{n-1}}{(d^2 - \frac{1}{D^2})^{n-1}}.
\end{equation}

In the last case, we consider that  both $U_-$ and $U_+$ form the unitary $2$-design. Following the Cauchy-Schwarz inequality in Lemma~\ref{CS_Inequality},  then the  Eq.~(\ref{graph_global}) is upper bounded by:
\begin{equation}
\ipic{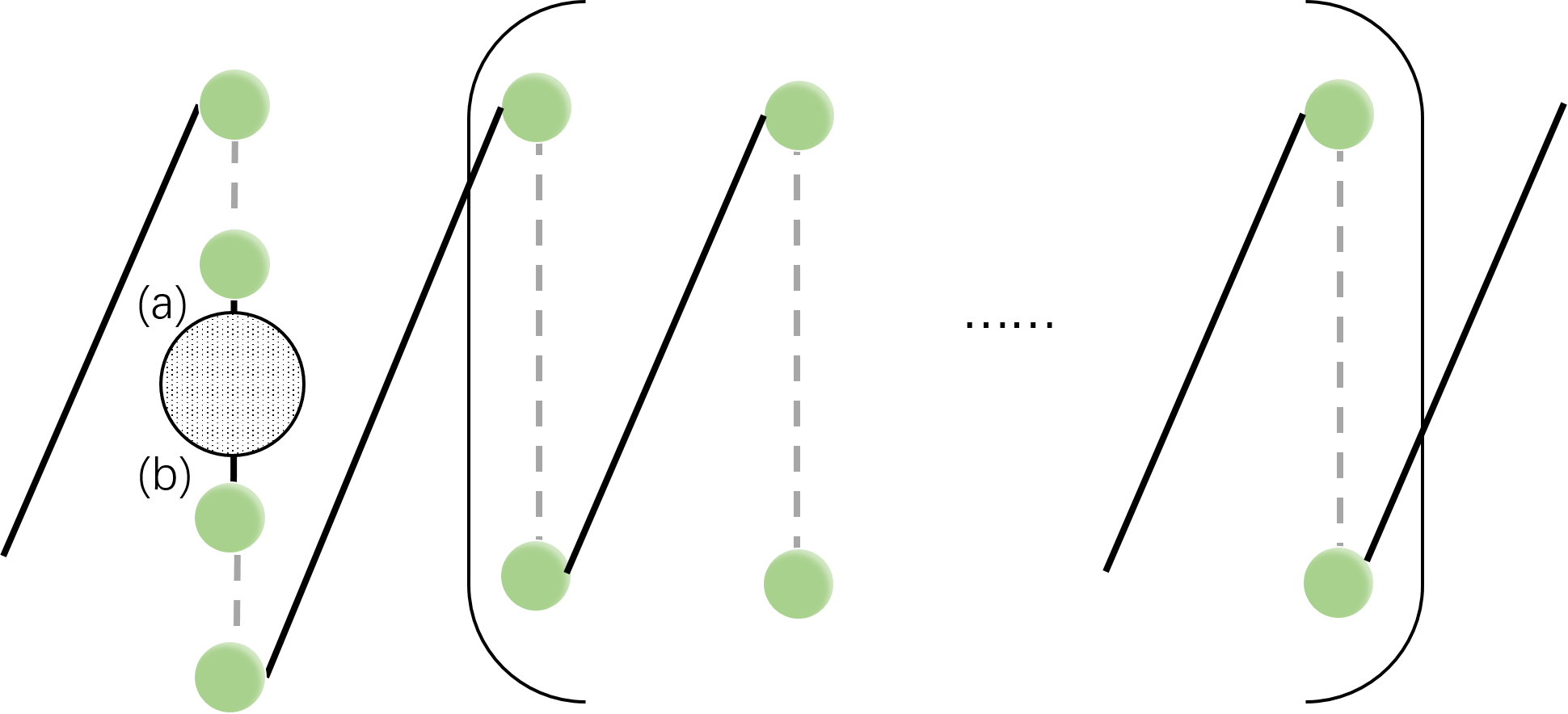}{0.3}
\end{equation}
where (a) denotes the possible connection of the four lines upon the grey circle, and (b) denotes the possible connection of the four lines upon the grey circle.

where the term on the derivative site is:
\begin{equation}
\ipic{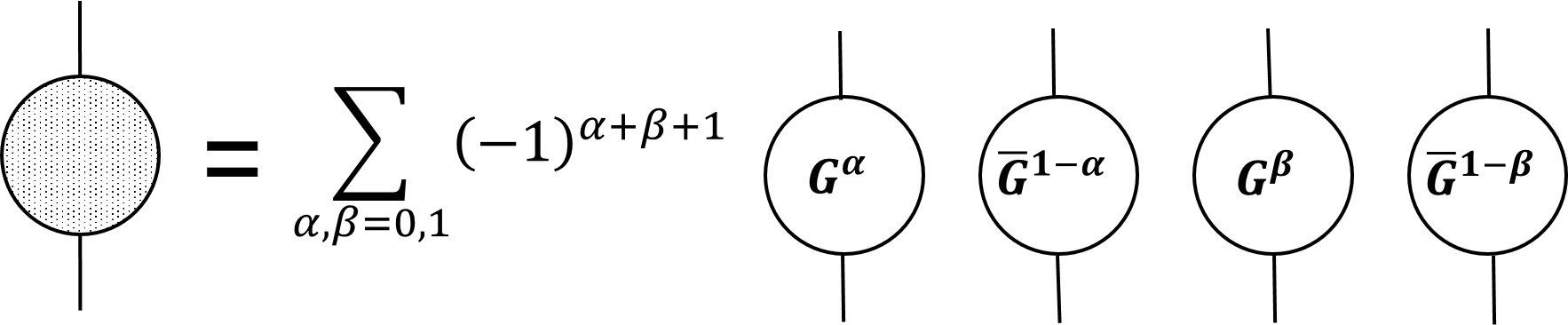}{0.3}.
\label{u+u-symbol}
\end{equation}
We find that only in the case that both (a)  and (b) take the $A$ connections, where the graph in Eq.~(\ref{u+u-symbol}) can be depicted as:
\begin{equation}\label{2designU+U-AA}
    \ipic{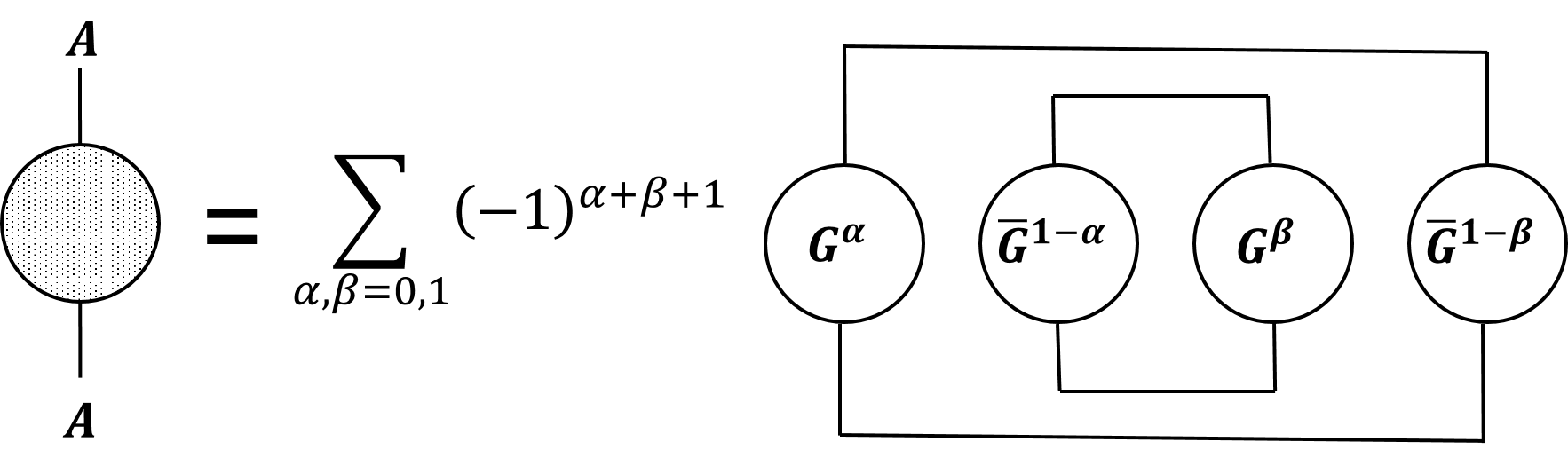}{0.3},
\end{equation}
the integration in Eq.~(\ref{2designU+U-AA}) takes the non-zero constant value:
\begin{equation}
    C_{g4} = 2\text{Tr}(G^2) Dd - 2 \text{Tr}(G)^2,
\end{equation}
and the variance  of the derivative of the global loss function  $\text{Var} (\partial_k \mathcal{L}_g)$  is bounded by
\begin{equation}\label{var_global_g4}
    \text{Var} (\partial_k \mathcal{L}_g)=\langle (\partial_\theta \mathcal{L}_g)^2\rangle \leq 2 
    \left[\frac{D d -1}{[(Dd)^2 -1]Dd}\right]^2 \frac{(1+\frac{1}{D})^{n-1}(1+\frac{1}{Dd})^{n-1}}{(d^2 - \frac{1}{D^2})^{n-1}} C_{g4}.
\end{equation}

 Based on the results in Eqs.~(\ref{figure36}, \ref{var_global_g2_g3}, \ref{var_global_g4}), where $C_{g(1,2,3,4)}$ are constant terms, one easily obtain that variance of the  of the derivative of the global loss function  $\text{Var} (\partial_k \mathcal{L}_g)$  is dominantly bounded  by the term $\frac{(1+\frac{1}{D})^{n-1}(1+\frac{1}{Dd})^{n-1}}{(d^2 - \frac{1}{D^2})^{n-1}}$, which decreases exponentially with respect to the system size $n$ for $d>1$ and $D>1$. 
 
 With the exponentially vanishing $\text{Var} (\partial_k \mathcal{L}_g)\sim \mathcal{O}(d^{-n})$, together with the average value of the term 
$\langle \frac{ \partial\mathcal{L}_g}{\partial \theta_k}\rangle =0$ in Eq.~(\ref{average_global_0}), we can show that
\begin{equation}\label{Prob:global}
{\rm Pr}\left(|\partial_k^{(i)}\mathcal{L}_g|>\epsilon\right)\leq \epsilon^{-2} \mathcal{O}(d^{-n})
\end{equation}
based on the Chebyshev's inequality, where $\rm Pr(\cdot)$ represents the probability. The result indicates the presence of the barren plateaus in the training process of  the tensor-network based machine learning models with respect to the global loss functions.

This completes the proof of the Theorem 1 in the main manuscript.

\section{Proof of theorem 2}
The general form of the local loss function can be written as:
\begin{equation}
    \mathcal{L}_l = \sum_i \langle \psi(\Theta) | \hat{O}_i |\psi(\Theta)\rangle, 
\end{equation}
where $\hat{O}_i$ is a local operator. For  simplicity, here we only consider the case that the local operator  acts on the $m$-th site of the lattice. Namely, the loss function can be written as:
\begin{equation}
\label{eq_loss}
    \mathcal{L}_l = \langle \psi(\Theta) | \hat{I}^{\otimes (m-1)} \hat{O} \hat{I}^{\otimes (n-m)} | \psi(\Theta) \rangle,
\end{equation}
where the lattice has $n$ sites. The graphical illustration of Eq.~(\ref{eq_loss}) is:
\begin{equation}
\ipic{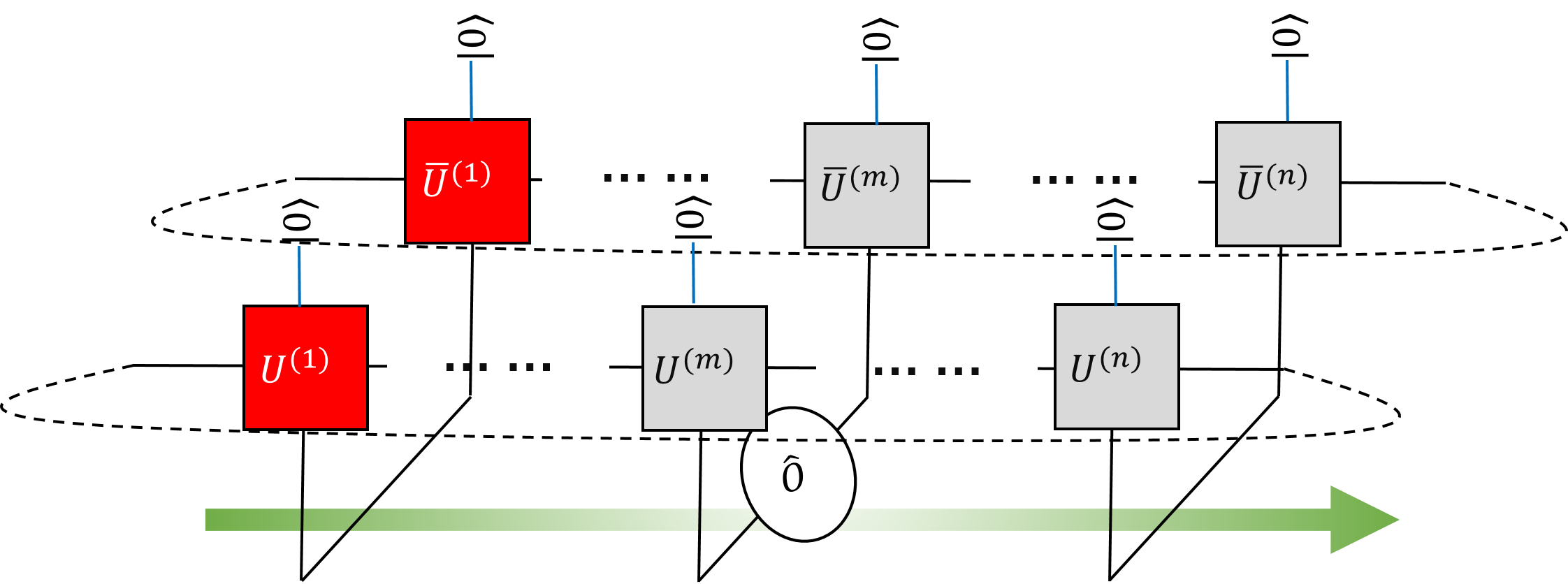}{0.3}.
\end{equation}
Without loss of generality, In our model, we consider the case that the first site hosts the derivative parameter (we denote such site by the derivative site) and the observable acts on the $m$-th site (we denote such site by the observable site). In the periodic condition, we take $m \leq \lfloor n/2 \rfloor$. We assume that the system size $n > 2$.  By using the unitary embedding techniques, our model can be identified as a quantum circuit that starts by the derivative site, passes through the observable site and loops over the sites in MPS states. 

In this proof, we divide  the calculation of the whole system  into three parts. The first part is the contraction over the derivative site. The second part is the contraction over the observable site, and the third part is the contraction over the self-connected sites. 

\subsection{Mean value of the derivative of local loss function}
\label{meanvalue_local}
Now we calculate the mean value  term $\langle \partial_k \mathcal{L}_l \rangle$. We first consider the off-site case (where the derivative site is not the same as the observable site).  Actually, to calculate the expectation value of the term $\langle \partial_k \mathcal{L}_l \rangle$,  we only need to care about the following term on the derivative site:
\begin{equation}
\begin{aligned}
\label{mean_local}
    \langle \partial_{\theta_k} (U(\Theta) U^\dagger(\Theta)) \rangle_{l_1,r_1} =& -\text{i}\langle U_+ G U_- |0\rangle \langle 0| U^\dagger_- U^\dagger_+\rangle_{l_1,r_1}\\
    & + \text{i} \langle U_+ U_- |0\rangle \langle 0| U^\dagger_- GU^\dagger_+\rangle_{l_1,r_1}.
\end{aligned}
\end{equation}
By integrating the nearest sites,  we obtain the graphical representation of Eq.~(\ref{mean_local}) as:
\begin{equation}
\label{mean_local_graph}
    \sum_{\alpha} (-1)^\alpha\int dU_- dU_+\ipic{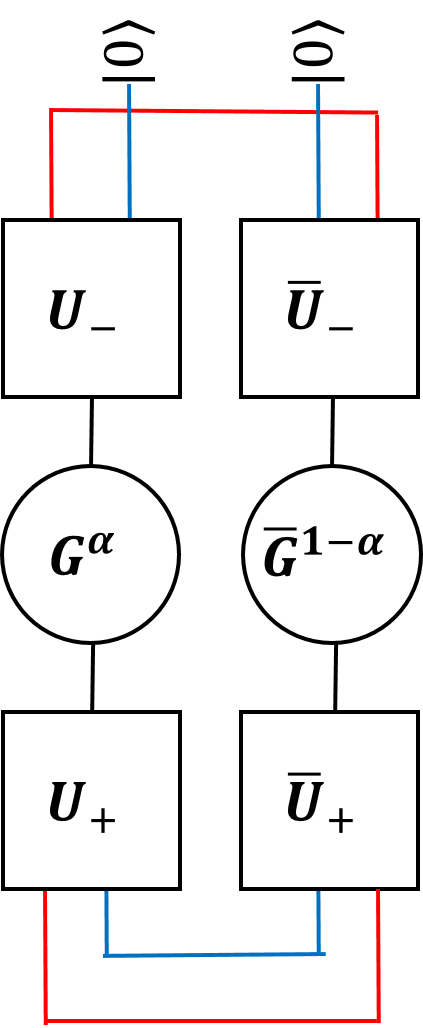}{0.3} = \int dU_- dU_+ \text{Tr}([G,\rho]) = 0,
\end{equation}
where we define $\rho = U_- |0\rangle \langle 0 | U^\dagger_-$. The above result indicates that the mean value  term $\langle \partial_k \mathcal{L}_l \rangle$ equals to zero for the off-site case. 

For the on-site case (where the derivative site is the same as the observable site),  we have the following expression on the derivative site:
\begin{equation}
\begin{aligned}
\label{mean_local_ob}
    \langle \partial_{\theta_k}( U(\Theta)\hat{O}U^\dagger(\Theta)) \rangle_{l_1,r_1} =& -\text{i}\langle U_+ G U_- |0\rangle \langle 0| U^\dagger_- U^\dagger_+ \hat{O} \rangle_{l_1,r_1}\\
    & + \text{i} \langle U_+ U_- |0\rangle \langle 0| U^\dagger_- GU^\dagger_+\hat{O}\rangle_{l_1,r_1}.
\end{aligned}
\end{equation}
By integrating the nearest sites, the graph representation of Eq.~(\ref{mean_local_ob}) is:
\begin{equation}
    \label{mean_local_ob_graph}
    \sum_{\alpha} (-1)^\alpha\int dU_- dU_+\ipic{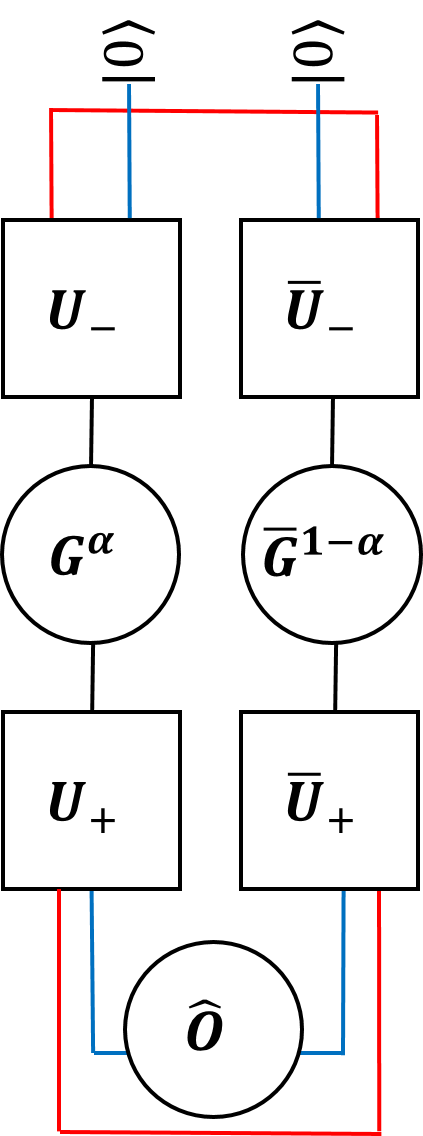}{0.3} .
\end{equation}
With the condition that either $U_-$ or $U_+$ forms the unitary $2$-design, similar to the calculation for global loss function, the sum of the integration with different values of $\alpha$ is $0$. So that the expectation value of derivative vanishes in the local loss function case. 

\subsection{Variance of the derivative of local loss function: off-site case}
\label{variance_local_offsite}


We first focus on the case where the derivative and the local operator are  applied on different site. Similar to our discussion in Sec.\ref{variance_derivative_glf}, in the derivative site, the variance of derivative can be graphically depicted as Eq.~(\ref{quadratic_derivative}).

Let us consider the case that the $U_-$ forms unitary $2$-design, where we integrated the $U_-$ part. We take the connection at the upper side of the generator $G$ as symmetric relations. In the graphical representation, the integration is:
\begin{equation}
\ipic{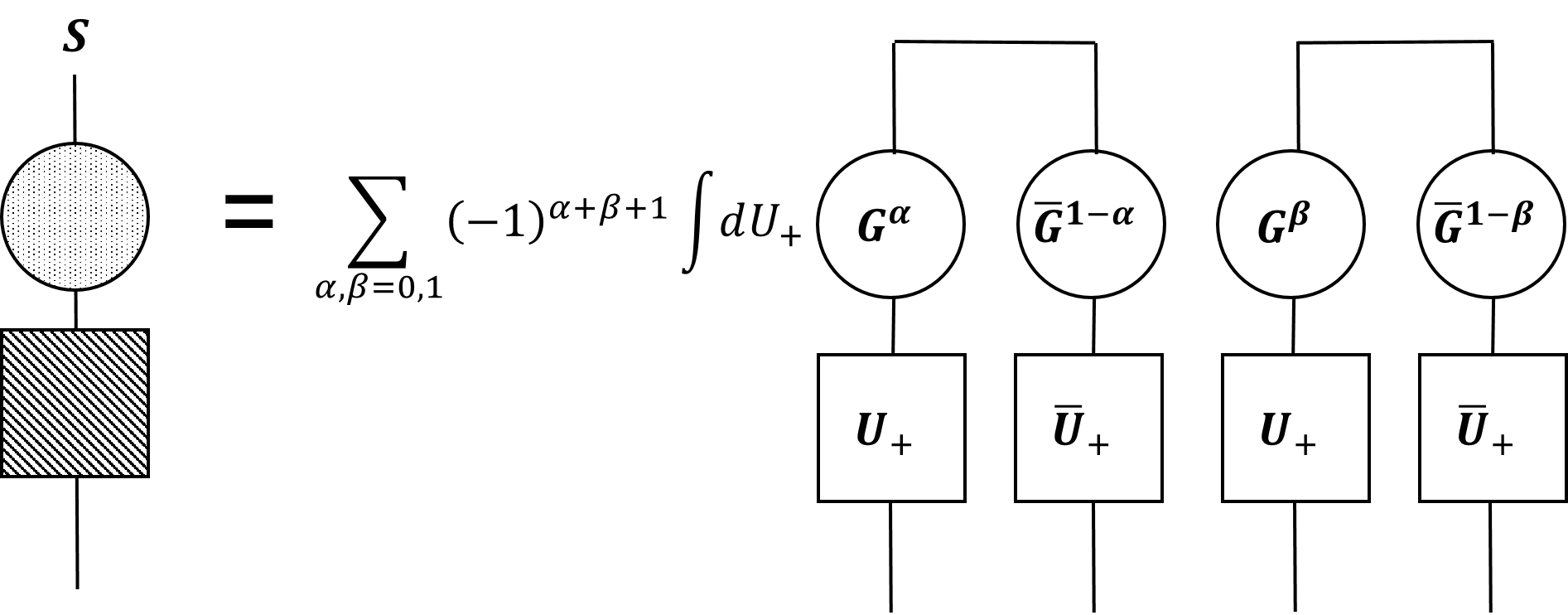}{0.3}.
\label{2design_upS}
\end{equation}
It is easy to see that for the integrands, the different combinations of the $\alpha$ and $\beta$ will generate the same structure, which will cancel with each other as sum over all possible values of $\alpha$ and $\beta$. 

The second case is the anti-symmetric connection at the upper index of the generator $G$, which will lead to the following two different types:
\begin{equation}
\ipic{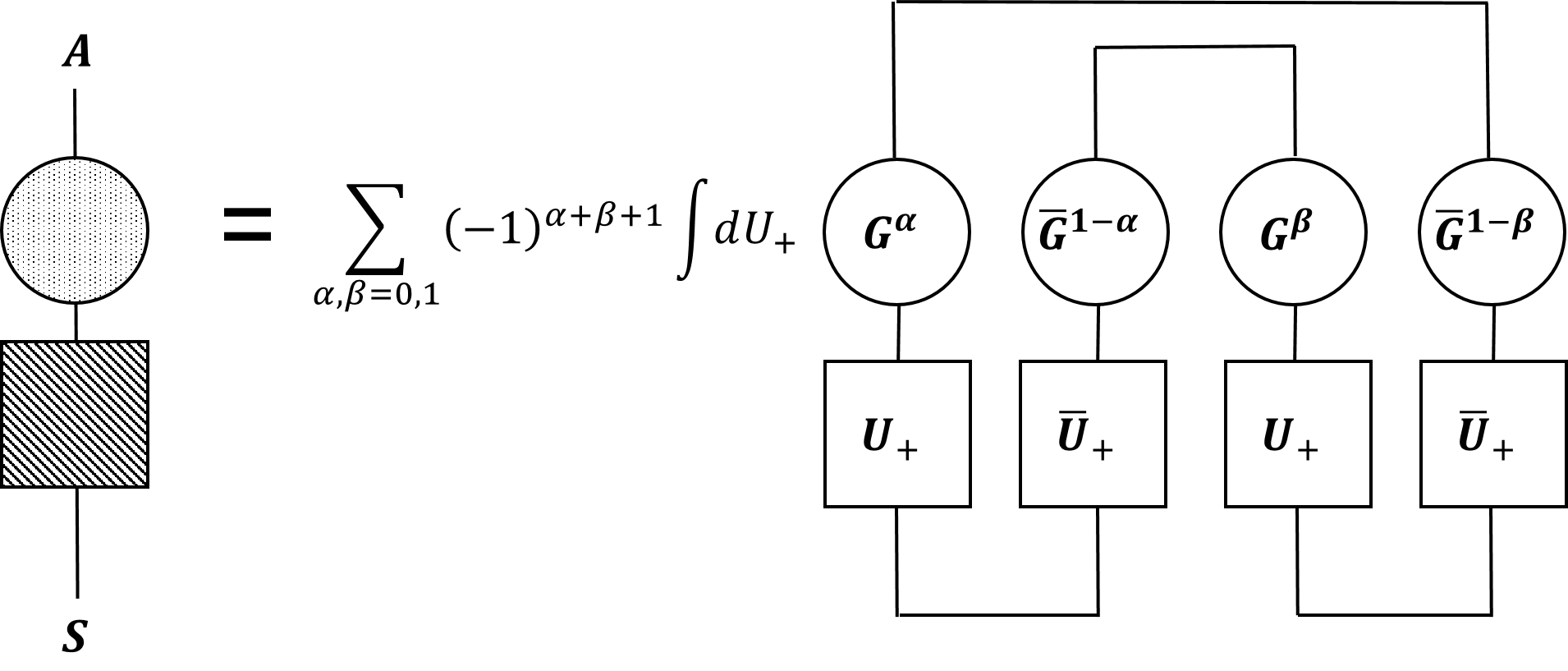}{0.3},
\end{equation}
and
\begin{equation}
\ipic{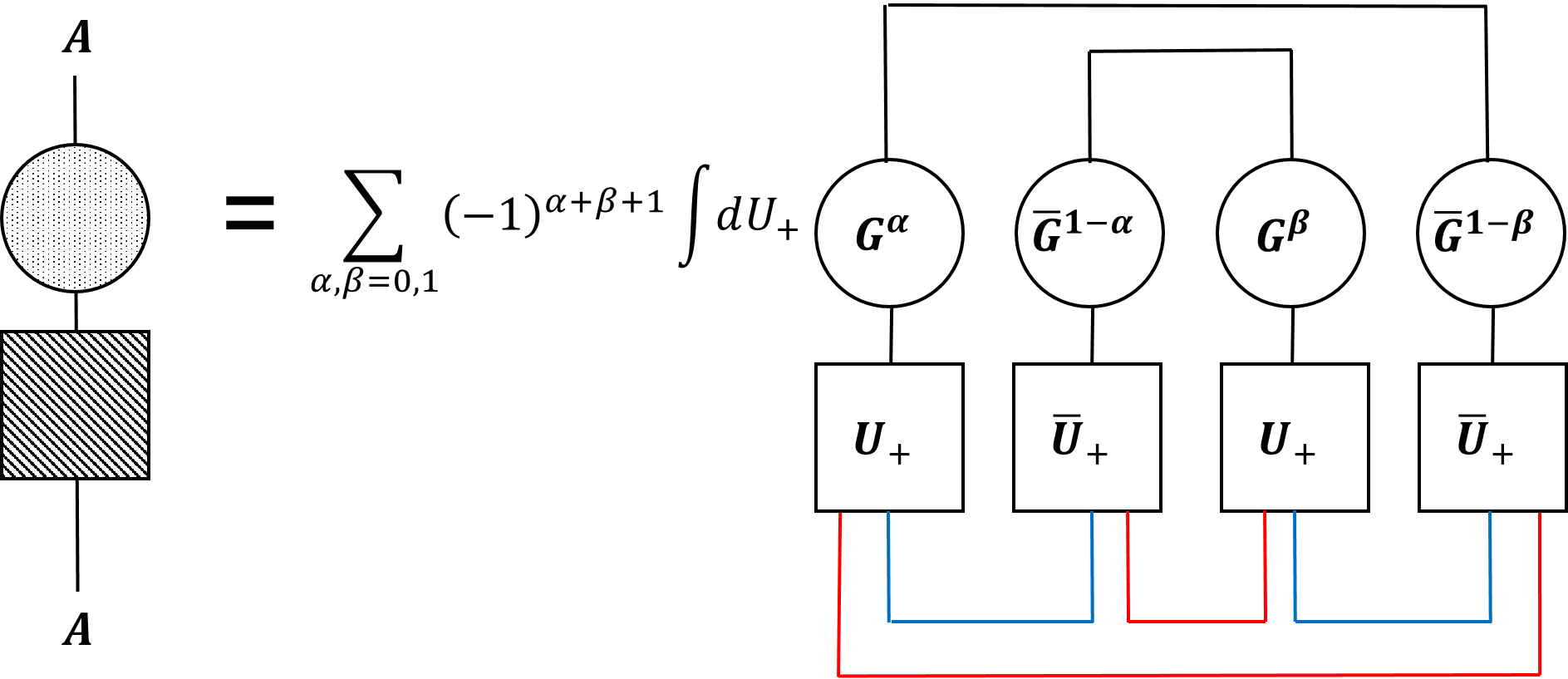}{0.3}.
\end{equation}
In the first connection type, if we
sum over the integrands with all possible values of $\alpha$ and $\beta$, the result will be:
\begin{equation}
    2\text{Tr}(G U_+ U^\dagger_+ G U_+U^\dagger_+) - 2\text{Tr}(G U_+ U^\dagger_+ U_+ U^\dagger_+ G) = 0.
\end{equation}

In the second type, if we sum over all the possible values of $\alpha$ and $\beta$, we obtain the contraction result is:
\begin{equation}
    C_1 = \int dU_+ 2[ D \text{Tr}_d(\text{Tr}_D(U_+ G^2 U^\dagger_+)) -  \text{Tr}_d(\text{Tr}_D(U_+ G U^\dagger_+)^2)].
\end{equation}

Then, let us consider the case that $U_+$ forms unitary $2$-design.  After integrated the $U_+$ part, we first take the symmetric connection at the bottom index of the generator $G$, which will generate the following structure:
\begin{equation}
\ipic{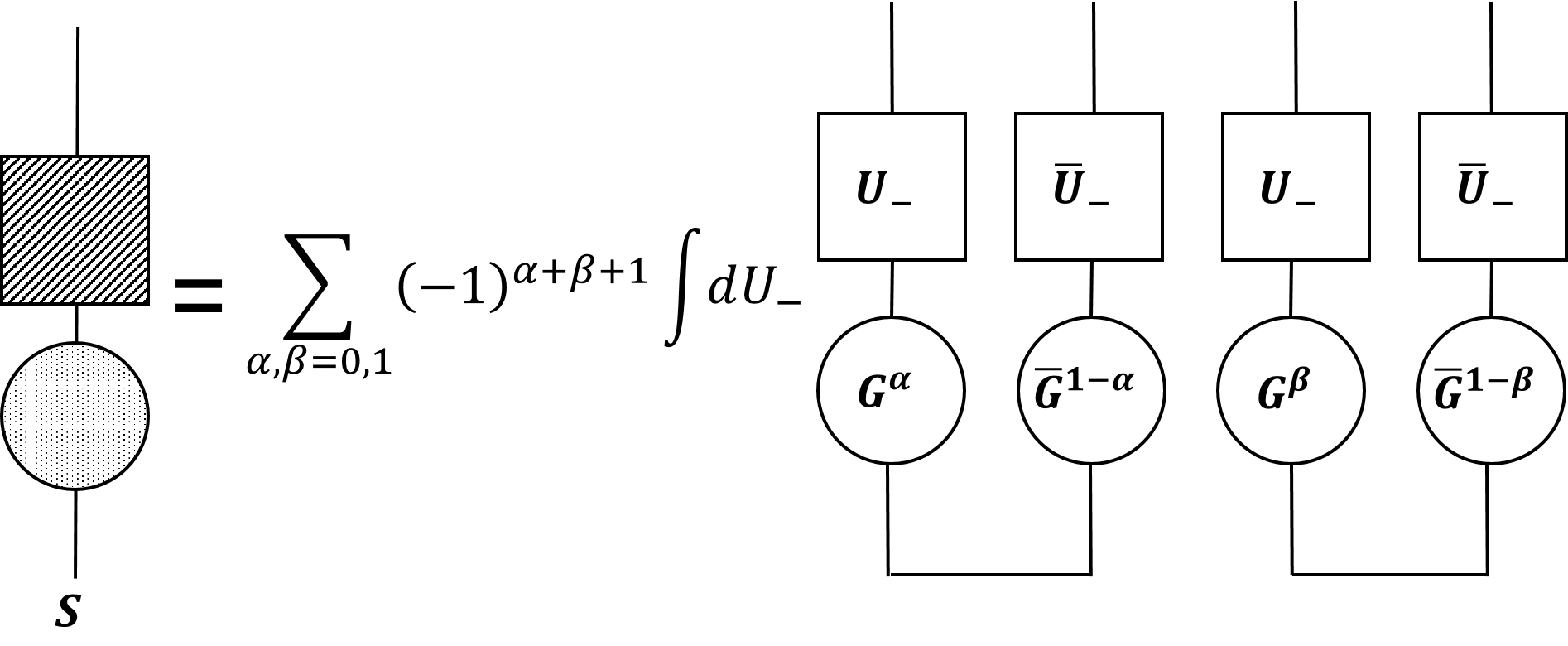}{0.3}
\end{equation}
same as our previous discussion in Eq.~\ref{2design_upS}, this term equals $0$ as we sum over the integrands with all possible values of $\alpha$ and $\beta$. 

And for the anti-symmetric connection at the bottom of the generator $G$, we can have the following structures:
\begin{equation}
\ipic{2design_U+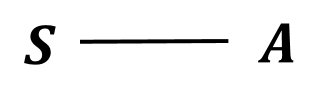}{0.3}
\end{equation}
After contraction of all the indexes, the output is:
\begin{equation}
    C_2 = \int dU_- 2[ \text{Tr}(\rho G^2 \rho)-\text{Tr}(\rho G \rho G)],
\end{equation}
where we denoted $U_- |0\rangle \langle 0 | U^\dagger_-$ as $\rho$. 
\begin{equation}
\ipic{2designU+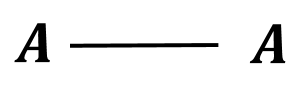}{0.3}
\end{equation}
we can obtain the output as:
\begin{equation}
    C_3 = \int dU_- 2 [Dd \text{Tr}(\rho G^2) - \text{Tr}(G \rho)^2].
\end{equation}

The third case is $U_-$ and $U_+$ are all form unitary $2$-design, at this case, symmetric connection of either side of indexes in the generator as:
\begin{equation}
\ipic{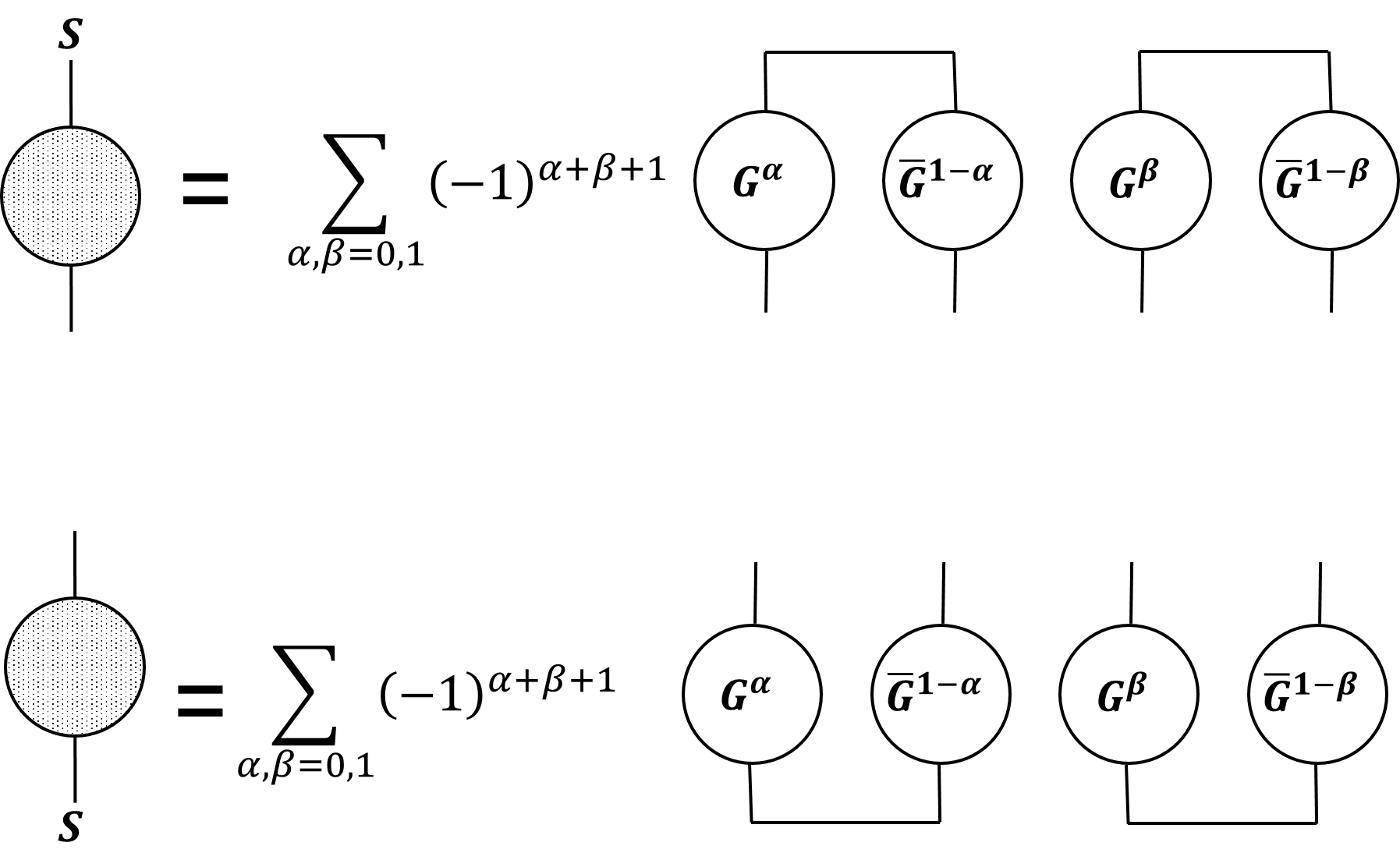}{0.3}
\end{equation}
will equal to $0$.

If the upper and bottom indexes of generator are all anti-symmetric connection, it will yield:
\begin{equation}
\ipic{2designU+U-AA.png}{0.3}
\end{equation}
with the output result:
\begin{equation}
    C_4 = 2[ \text{Tr}(G^2) D d -  \text{Tr}(G)^2]
\end{equation}
We find that only the anti-symmetric connection of the bottom indexes can obtain the non-zero solution.

Then, let us consider the observable site. The integration of the local operator site can be graphically draw as:
\begin{equation}
\ipic{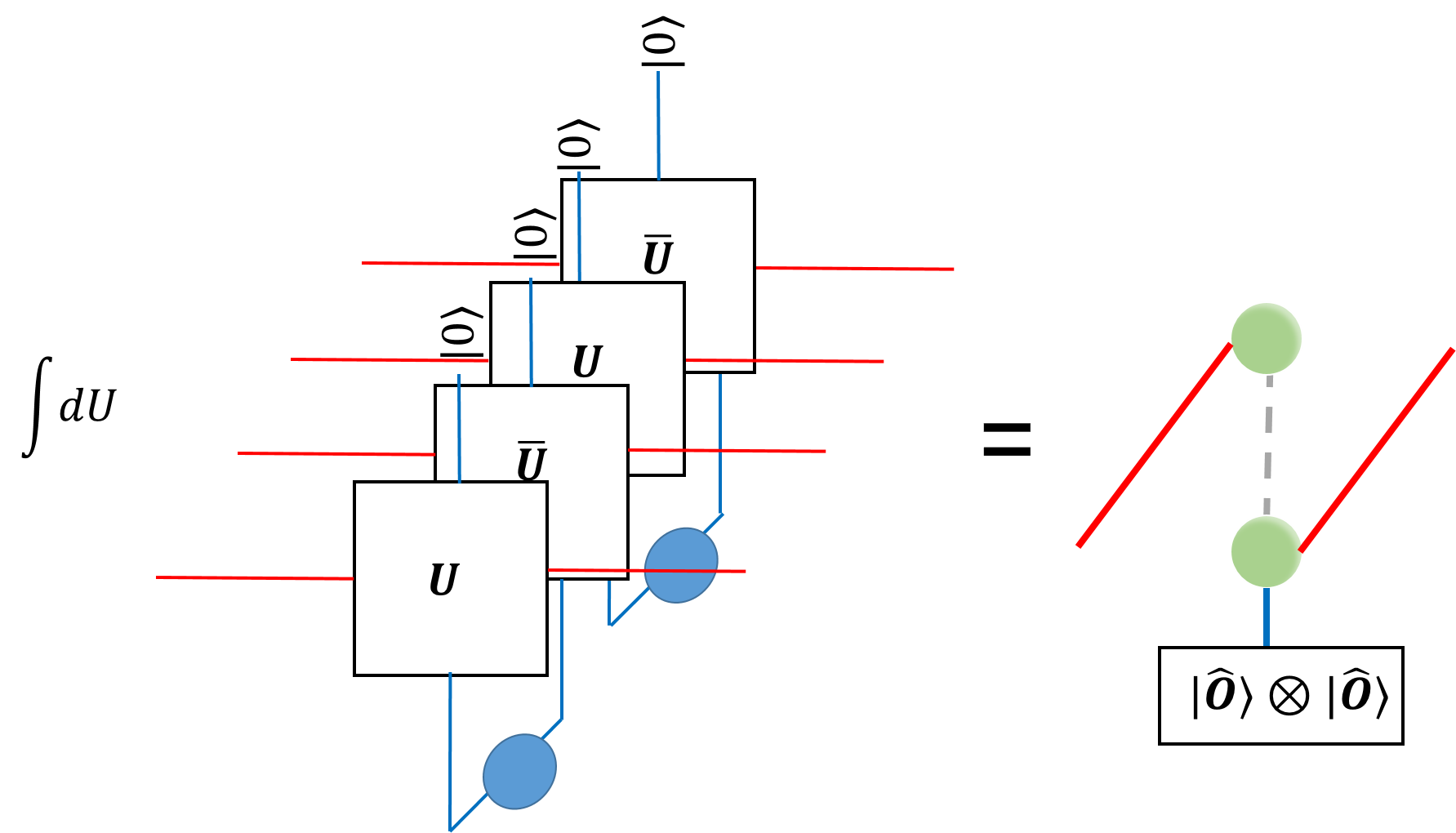}{0.3},
\end{equation}
where we use the choi-jamiolkowski isomorphism and graphically represent the local operator:
\begin{equation}
\ipic{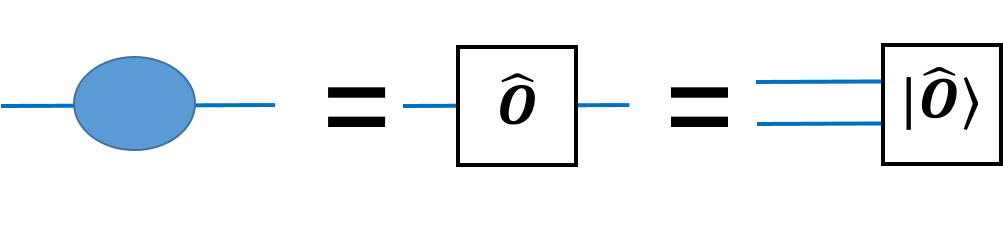}{0.3}.
\end{equation}

It is easy to verify the following results:
\begin{equation}
\ipic{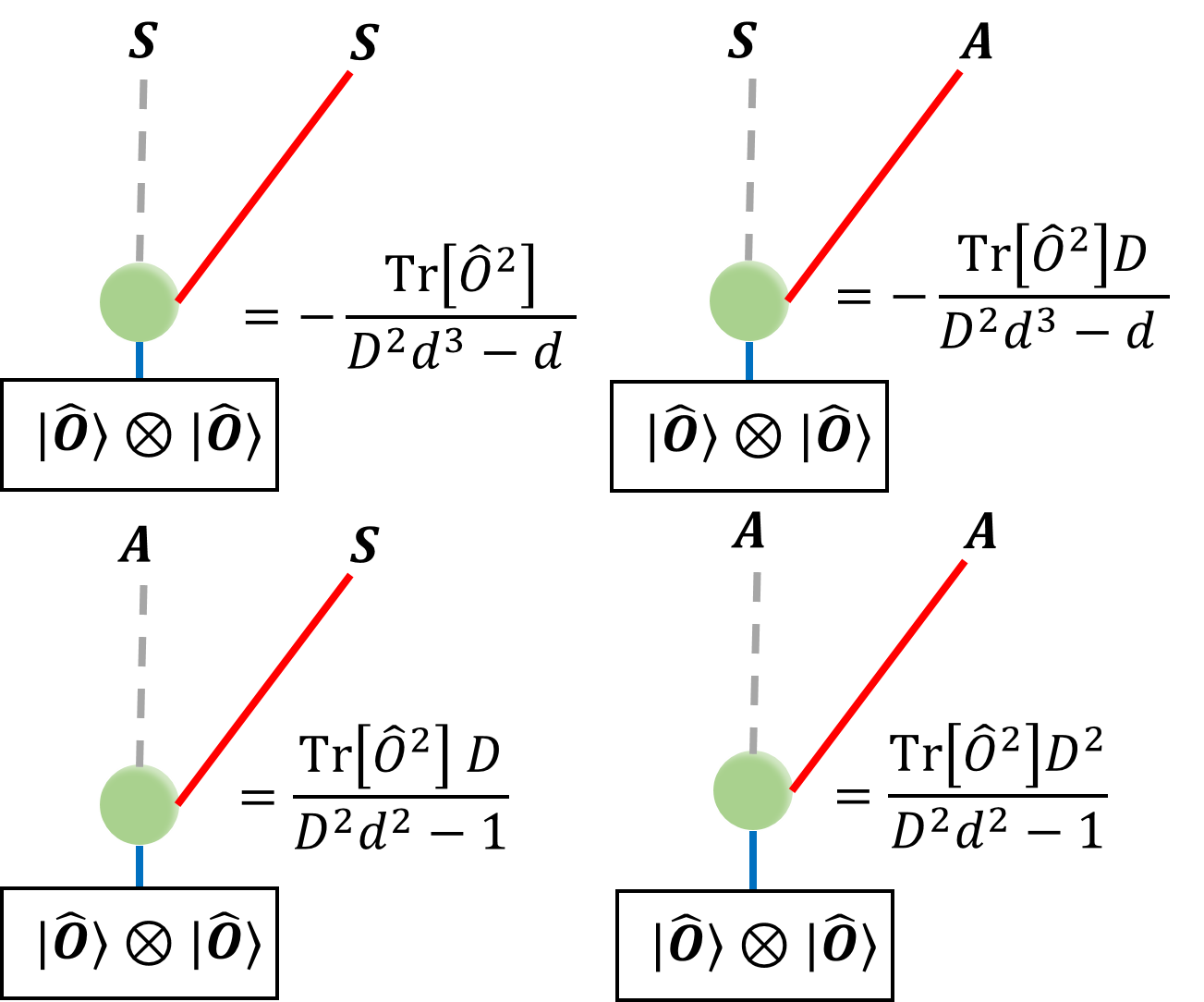}{0.3}
\end{equation}

For the self-connected sites, the integration will generate the form:
\begin{equation}
\ipic{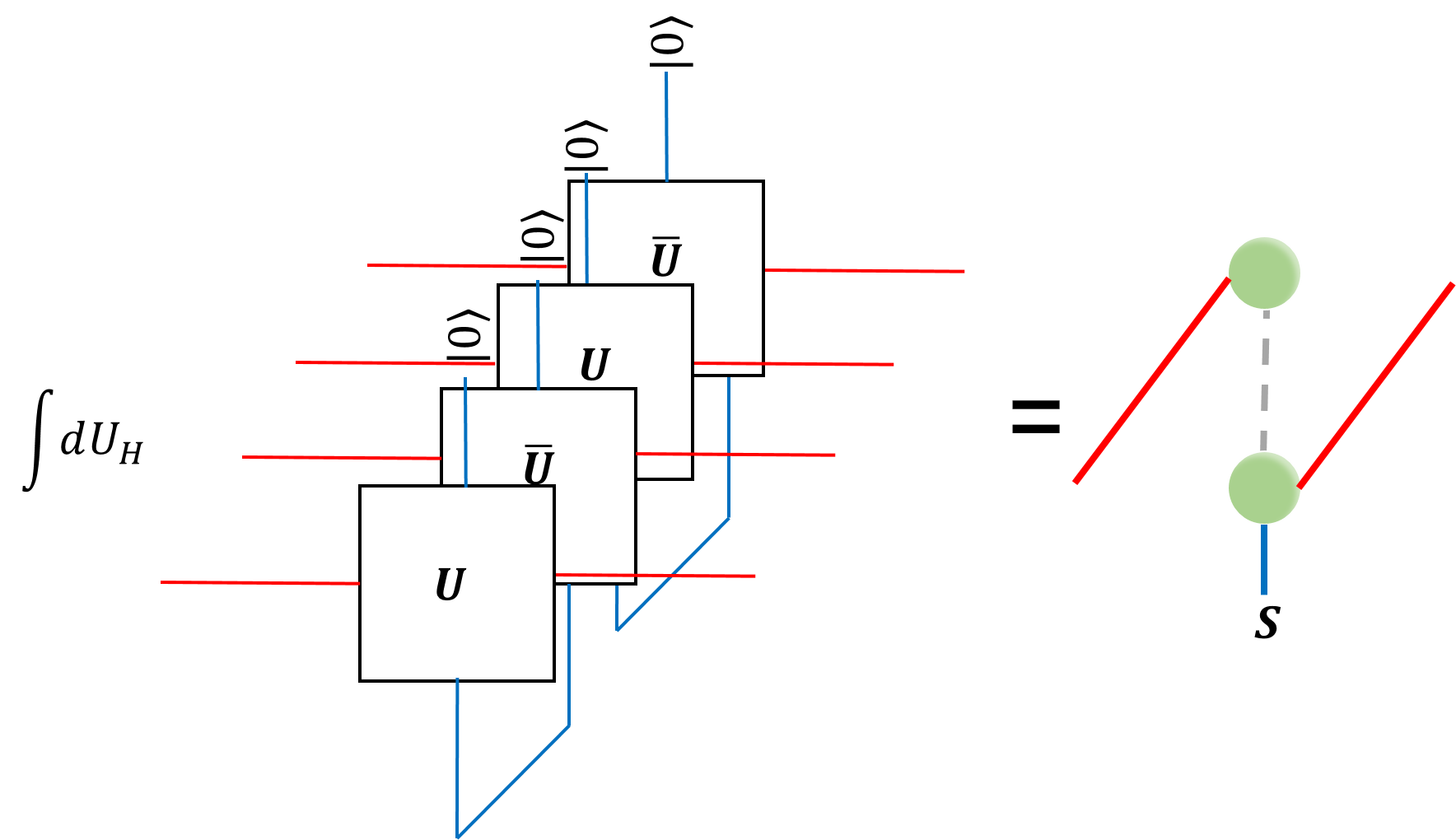}{0.3}
\end{equation}
 We have the following results:
\begin{equation}
\ipic{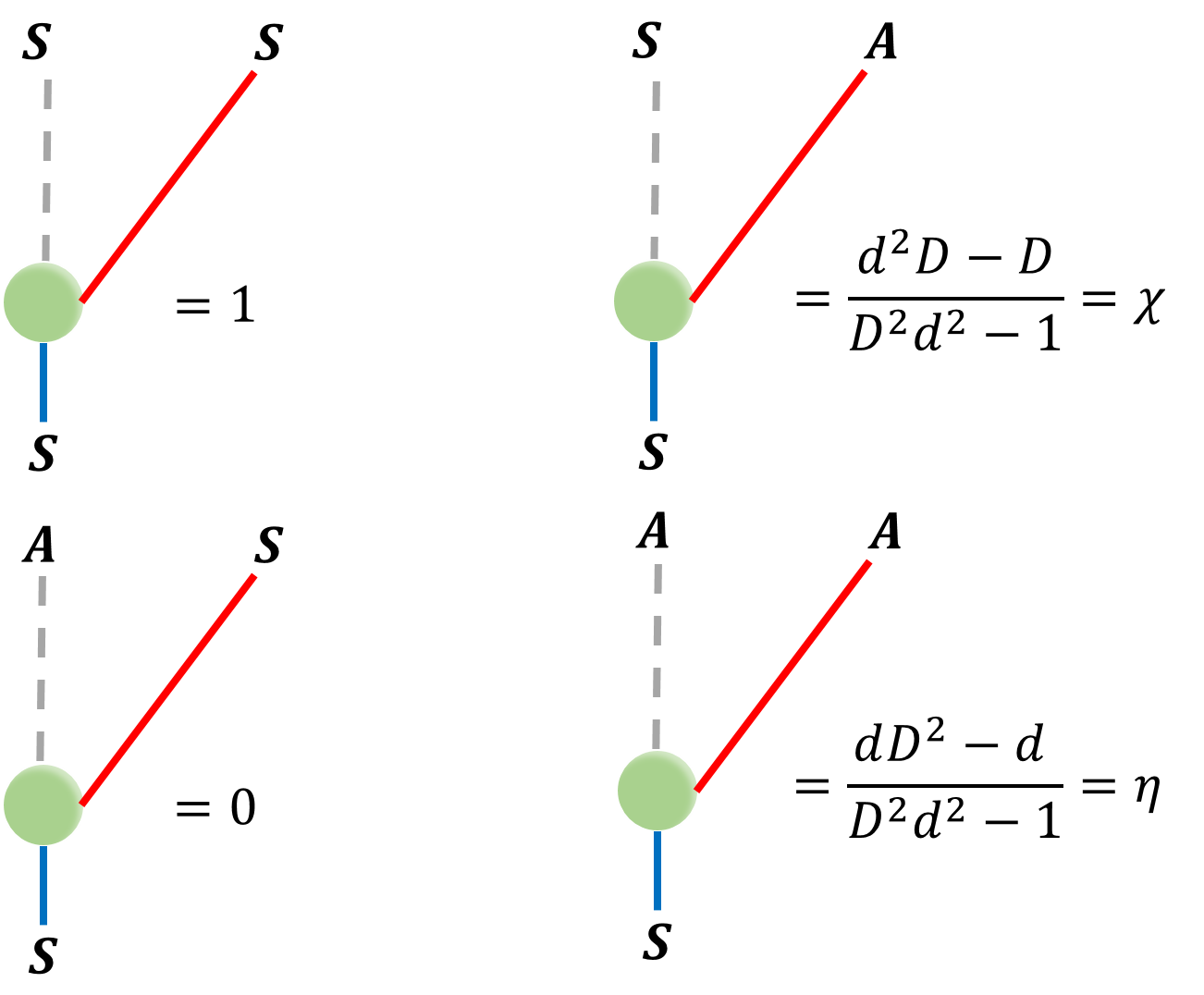}{0.3}.
\end{equation}
As we consider $L$ self-connected sites, the contraction result leads to:
\begin{equation}
\ipic{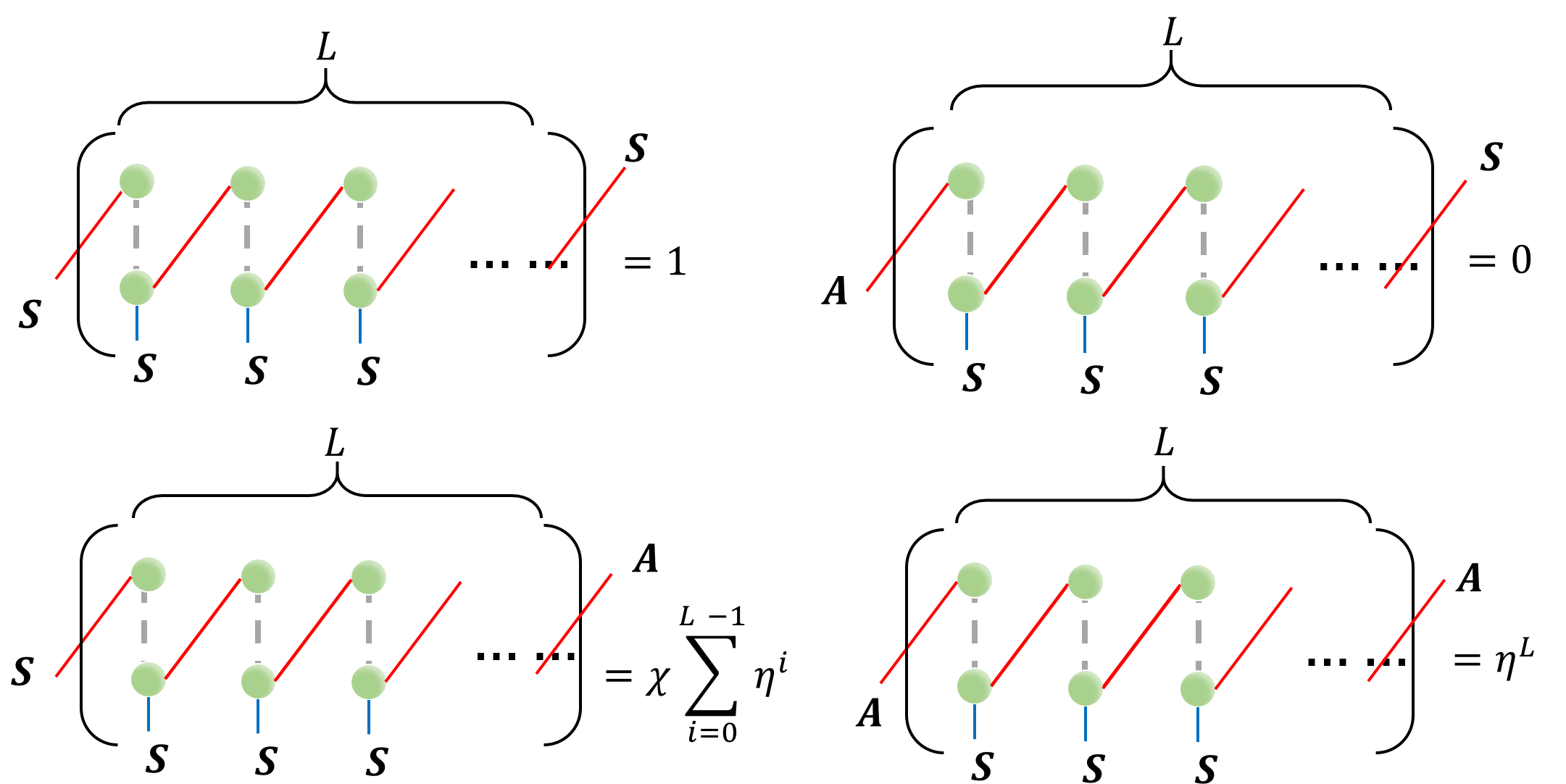}{0.3}
\end{equation}

Now, we combine the derivative site, observable site, and the self-connected sites together. In the case that $U_-$ forms unitary $2$-design, the variance of the derivative can be written as :
\begin{equation}
\label{overallu-}
    \ipic{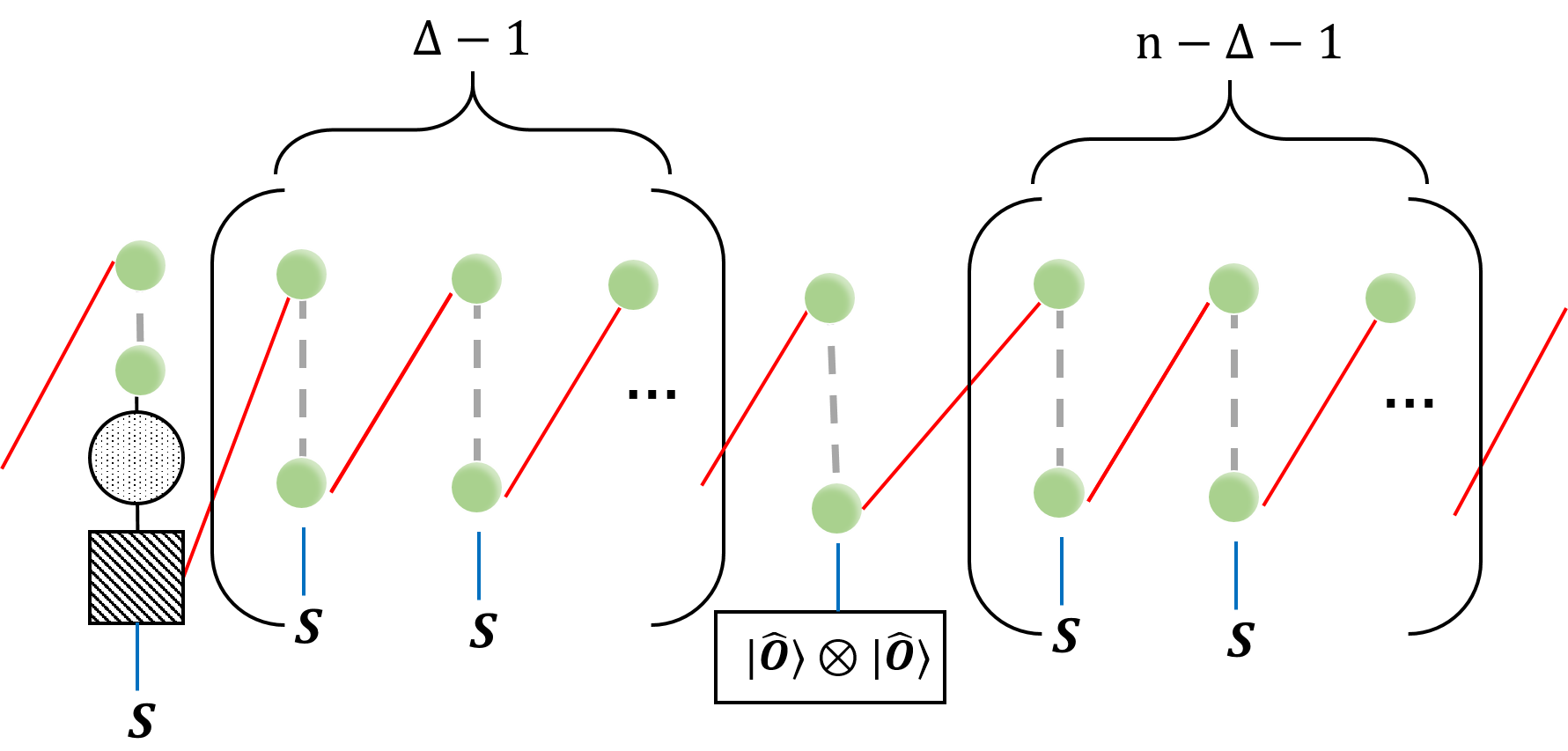}{0.3},
\end{equation}
where we define the distance $\Delta = m-1$.

According to our previous calculation, we obtain the following contraction result:
\begin{equation}
\label{outputoverallu-}
    \langle (\partial_\theta \mathcal{L}_l)^2\rangle =  C_1 \frac{\text{Tr}(\hat{O}^2)D \eta^{\Delta-1}}{(D^2 d^2-1)^2}
    \left[D \eta^{n-\Delta-1} + \frac{1-\eta^{n-\Delta-1}}{1-\eta} \chi -\frac{1}{D^2 d^2}\right].
\end{equation}


If $U_+$ forms unitary $2$-design,  the variance of the derivative can be written as:
\begin{equation}
    \ipic{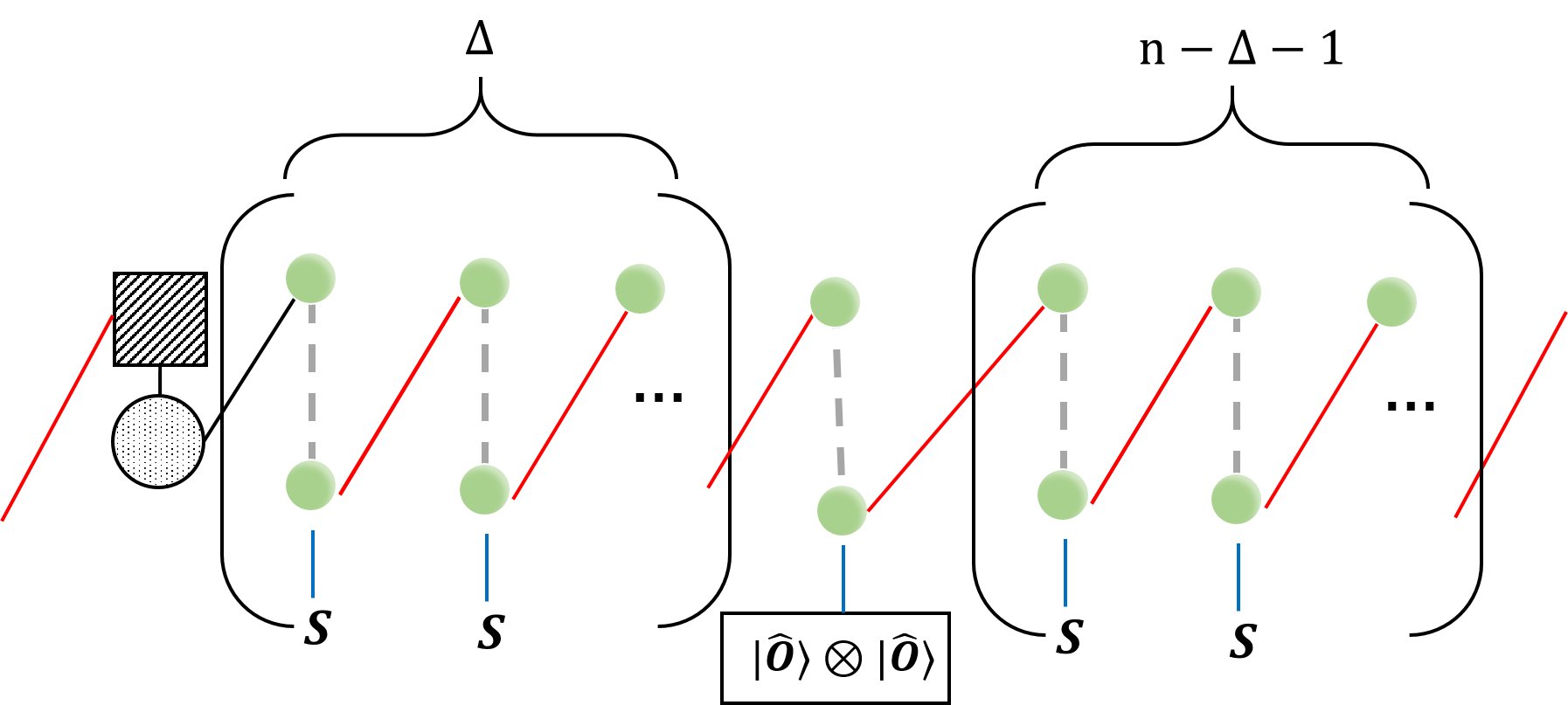}{0.3}.
\end{equation}
 The contraction result reads:
\begin{equation}
\label{outputoverallu+}
\begin{aligned}
       \langle (\partial_\theta \mathcal{L}_l)^2\rangle = & C_2\eta^\Delta\left[\frac{\text{Tr}(\hat{O}^2) D}{D^2 d^2 -1}\right]+\\
    &C_3\eta^{\Delta} \left[\frac{\text{Tr}(\hat{O}^2) D^2}{D^2 d^2 -1} \eta^{n-\Delta-1} + \frac{\text{Tr}(\hat{O}^2) D}{D^2 d^2 -1}\frac{1-\eta^{n-1}}{1-\eta}\chi\right].
\end{aligned}
\end{equation}

If both $U_+$ and $U_-$ form unitary $2$-design, the variance of derivative can be written as:
\begin{equation}
    \ipic{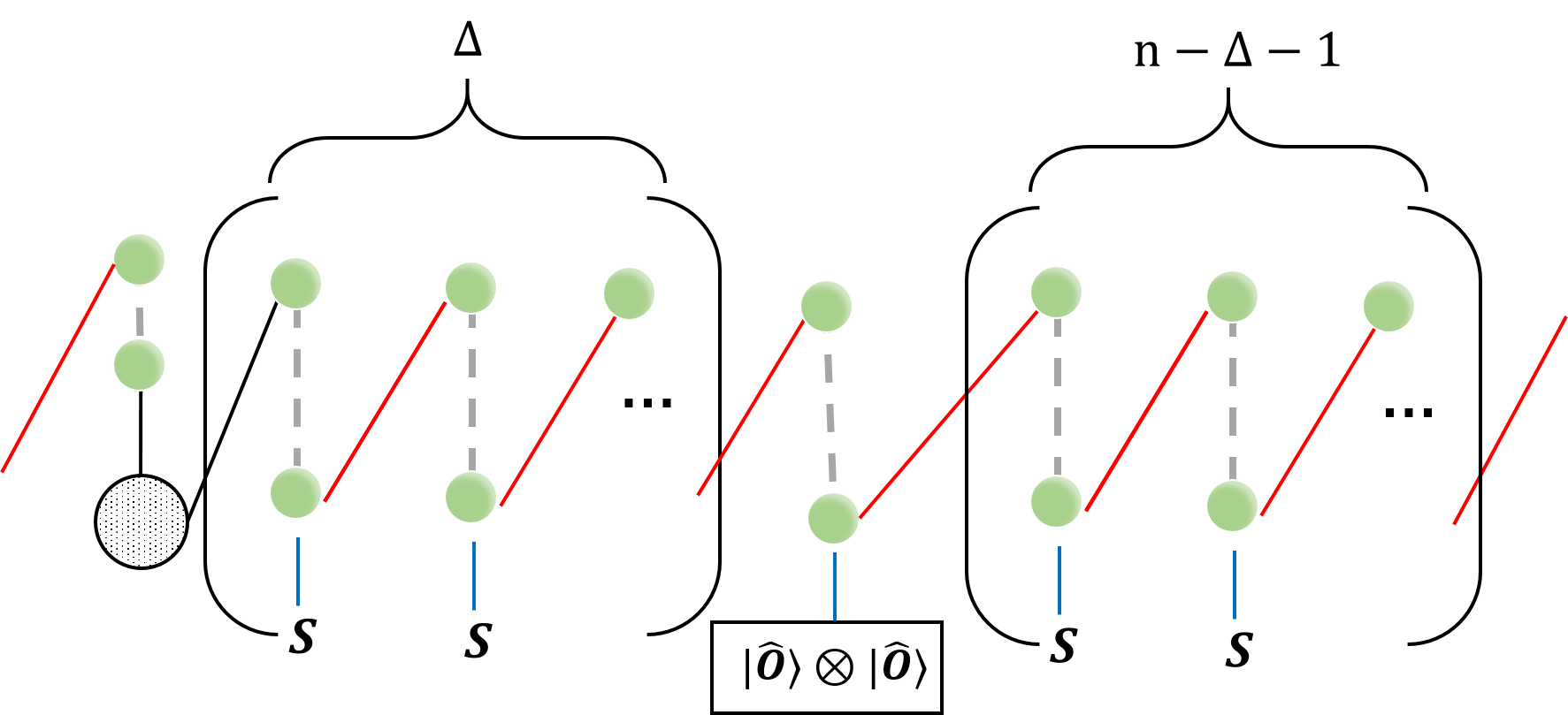}{0.3},
\end{equation}
and the contraction result reads:
\begin{equation}
\label{outputoverallu-u+}
     \langle (\partial_\theta \mathcal{L}_l)^2\rangle =  C_4 \frac{\text{Tr}(\hat{O}^2)D \eta^{\Delta}}{(D^2 d^2-1)^2}
    \left[D \eta^{n-\Delta-1} + \frac{1-\eta^{n-\Delta-1}}{1-\eta} \chi -\frac{1}{D^2 d^2}\right].
\end{equation}

As we can see that in Eqs.~(\ref{outputoverallu-}, \ref{outputoverallu+}, \ref{outputoverallu-u+}),  all of them show an exponentially decay with respect to the distance $\Delta$ between the observable site and the derivative site. As we set the distance $\Delta$ as constant, the increasing $n$ will eventually converge to a constant number.  

\subsection{Variance of the derivative of local loss function: on-site case}
\label{variance_local_onsite}
Then, let us consider the on-site case (where the derivative site is the same as the observable site). As we assume $U_-$ forms unitary $2$-design, the integration over $U_+$ has the form:
\begin{equation}
\ipic{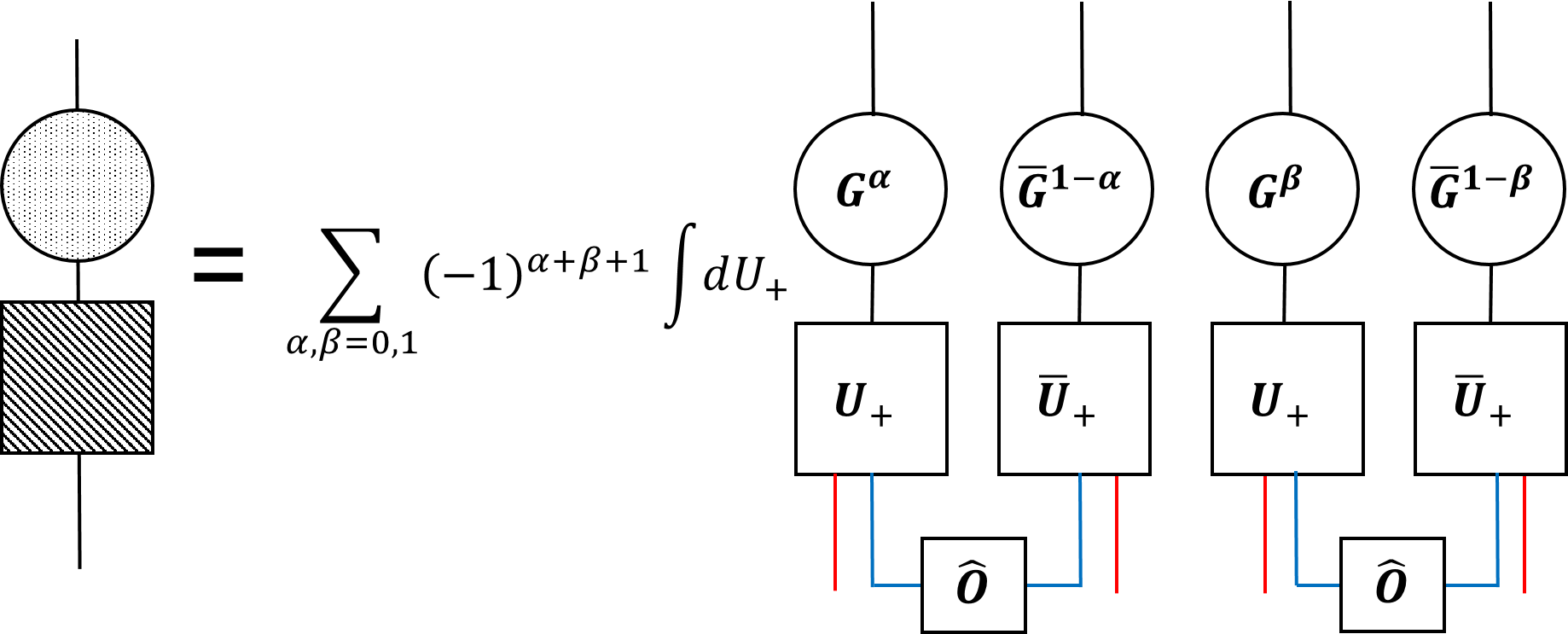}{0.3}
\end{equation}
Let us consider the four connection case as follows:
\begin{equation}
\ipic{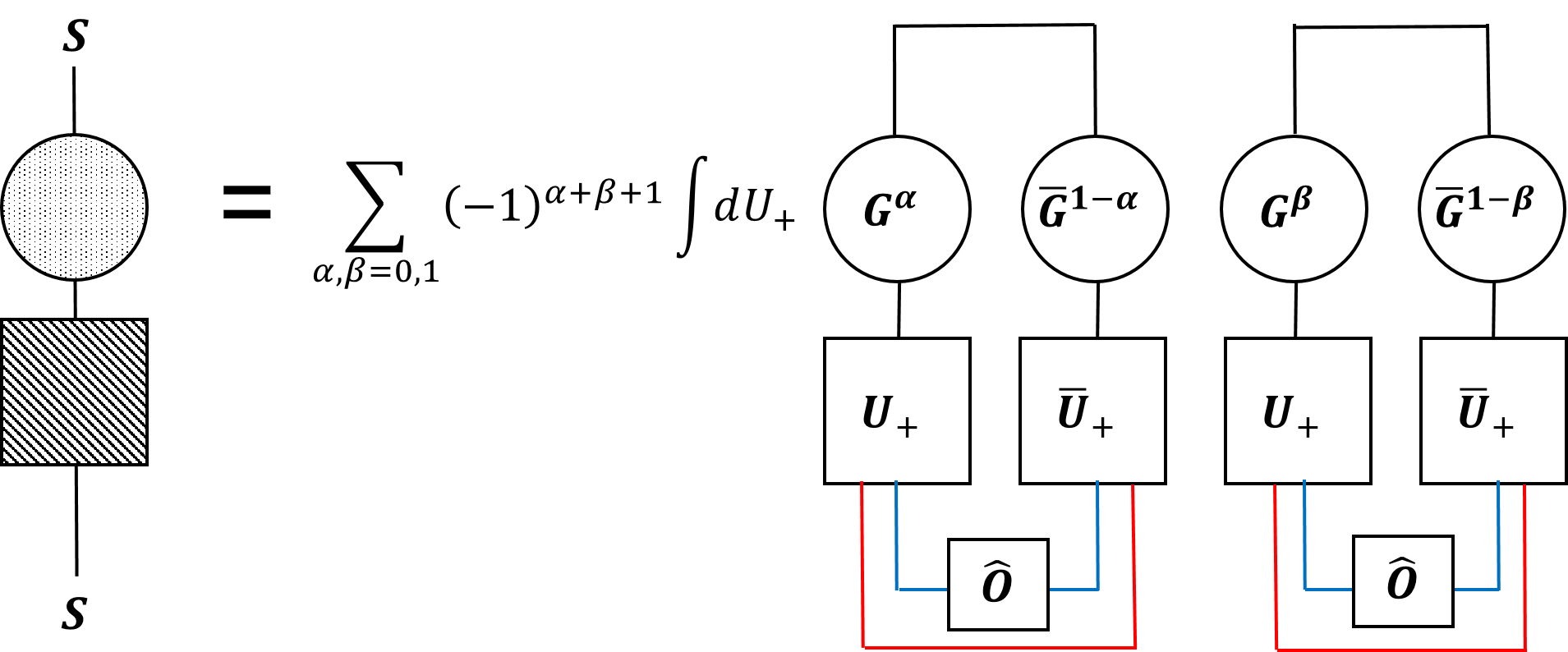}{0.3}
\end{equation}
\begin{equation}
\ipic{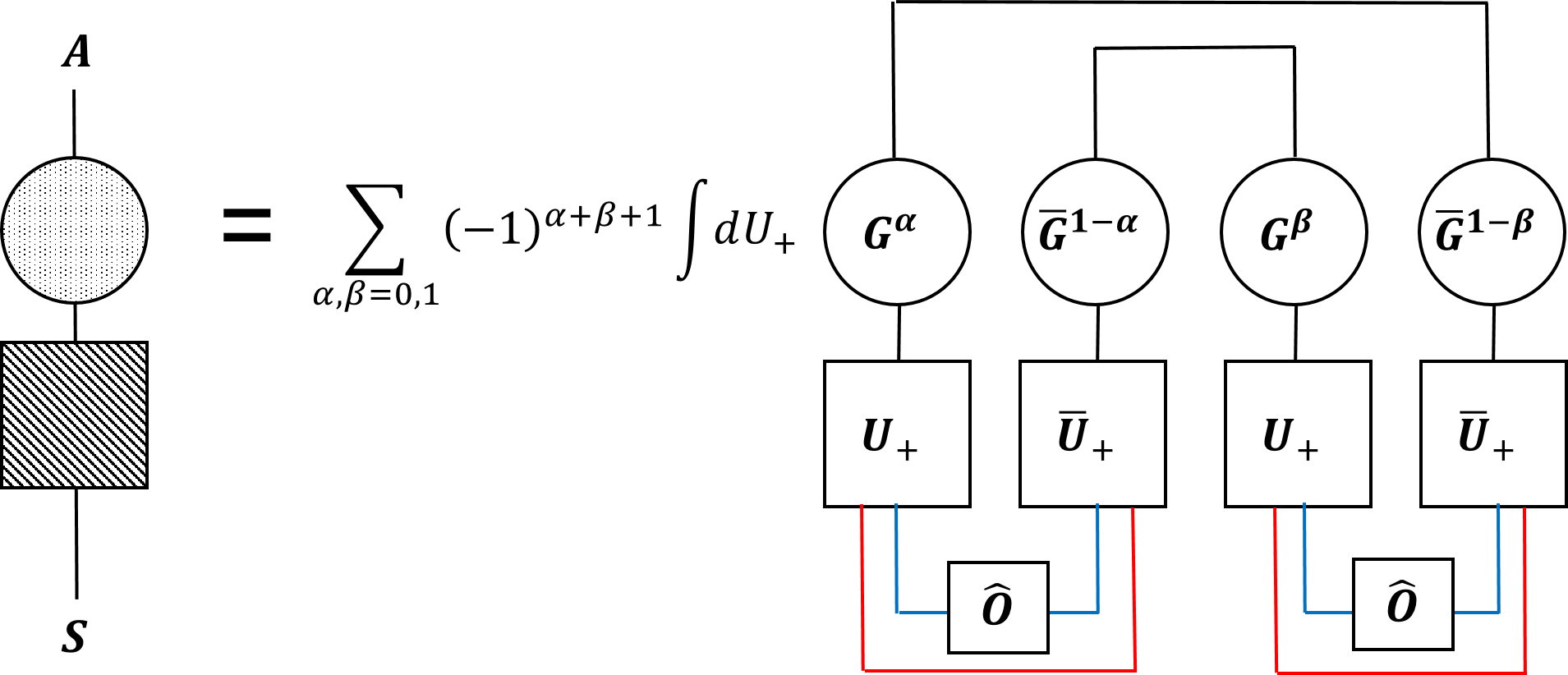}{0.3}
\end{equation}
\begin{equation}
\ipic{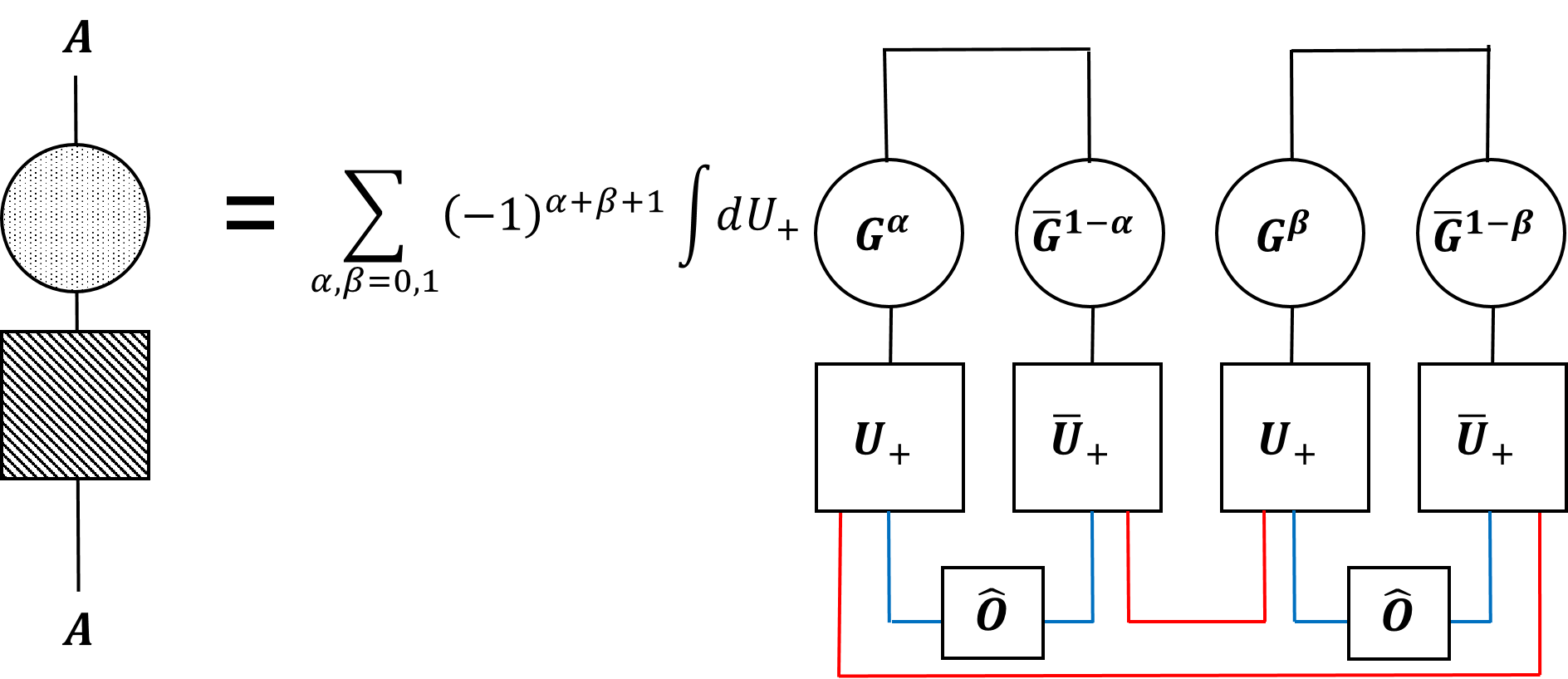}{0.3}
\end{equation}
\begin{equation}
\ipic{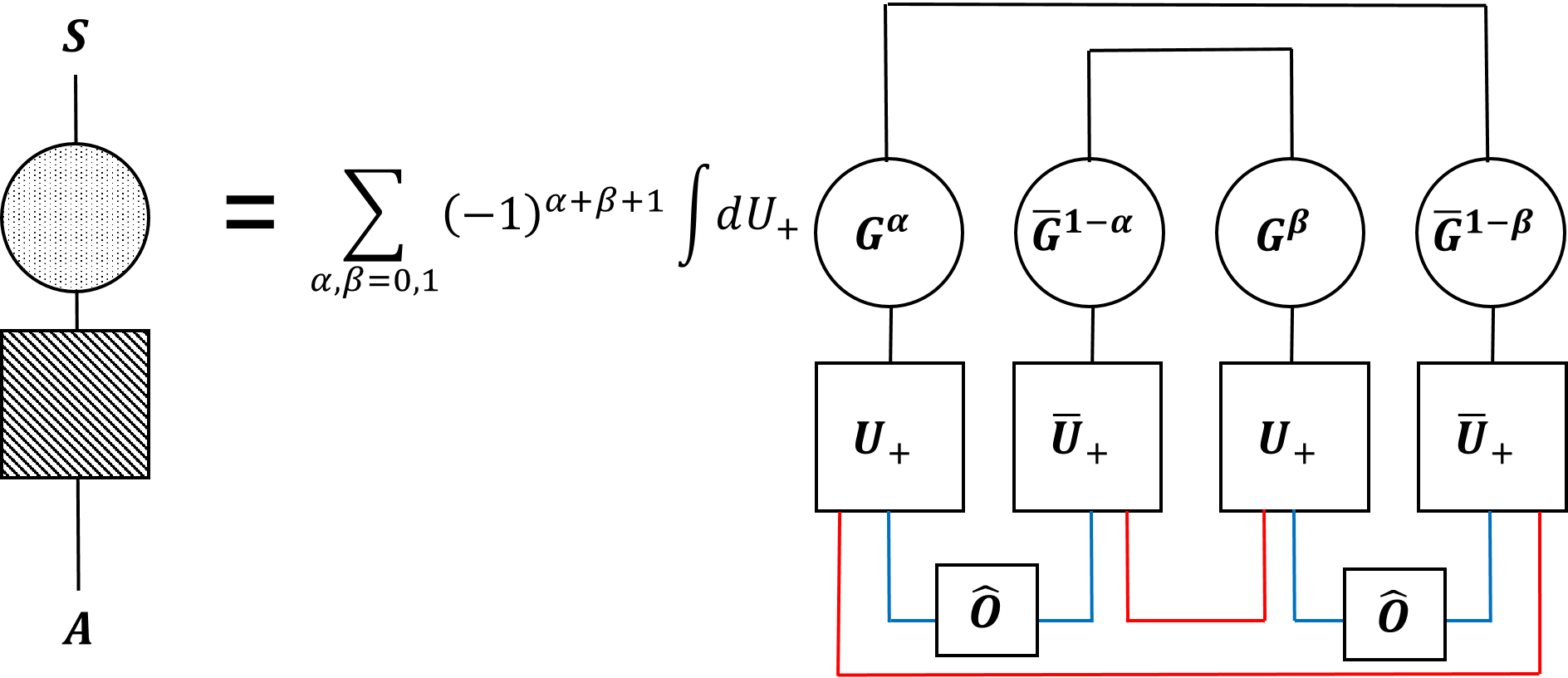}{0.3}
\end{equation}
The first and third connection cases are easy to see that it leads to $0$. The integration in the second connection case leads to:
\begin{equation}
    C_5 = \int dU_+ 2 [\text{Tr}(\sigma G^2 \sigma) - \text{Tr}(\sigma G \sigma G)]
\end{equation} 
where we denoted $\sigma = U_+ (\hat{O}\otimes I) U^\dagger_+$. In the last connection case, the integrated result leads to:
\begin{equation}
    C_6 = \int dU_+ 2 \text{Tr}_d\left\{D \text{Tr}_D\left[U_+ G^2 U^\dagger_+\right] \hat{O}^2-\left[\text{Tr}_D(U_+ GU^\dagger_+)\hat{O}\right]^2\right\}.
\end{equation}
The variance of derivative can be written as:
\begin{equation}
\ipic{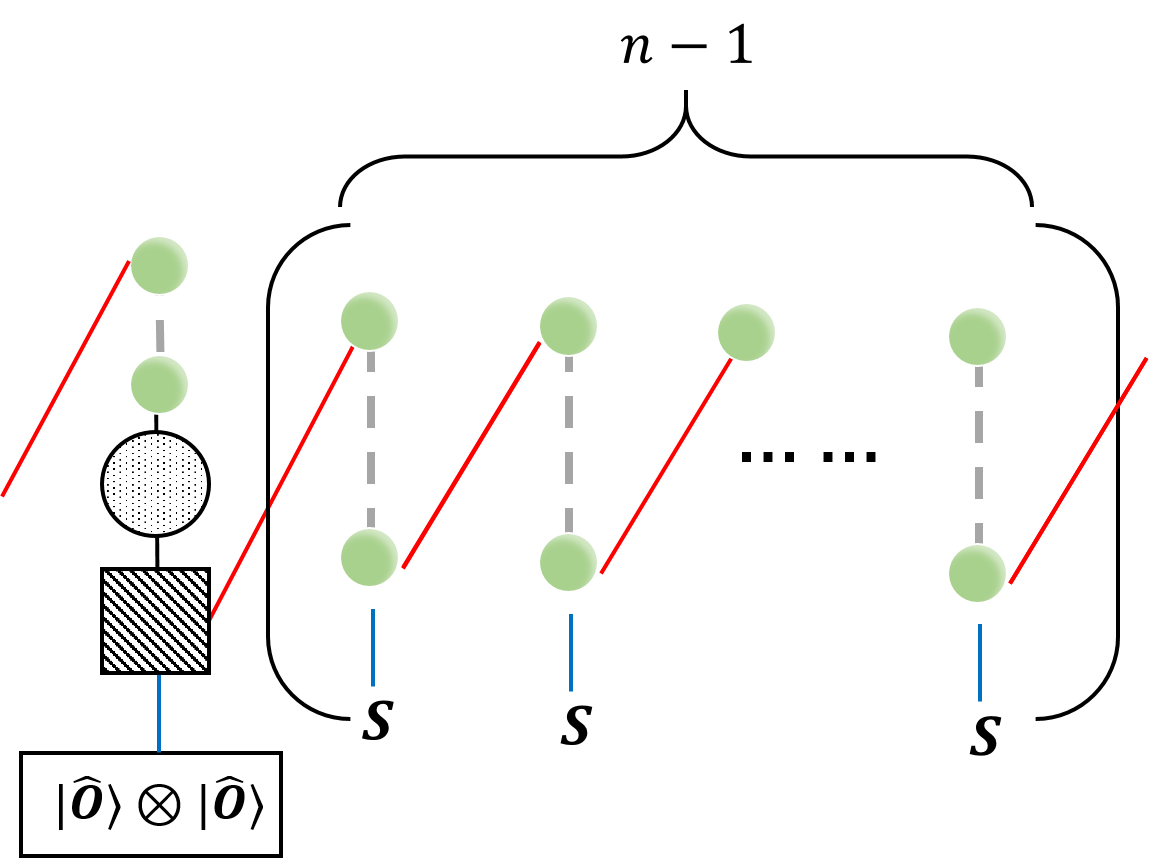}{0.3}
\end{equation}
and the contraction result reads:
\begin{equation}
\label{onsiteoverallu-}
    \langle (\partial_\theta \mathcal{L}_l)^2\rangle = \frac{C_5}{(Dd)^2-1}(\frac{1-\eta^{n-1}}{1-\eta}\chi - \frac{1}{Dd})+ \frac{C_6}{(Dd)^2 -1 }\eta^{n-1}
\end{equation}
For the case where only the $U_+$ forms unitary $2$-design, same as our previous discussion, the variance of derivative can be written as:
\begin{equation}
\ipic{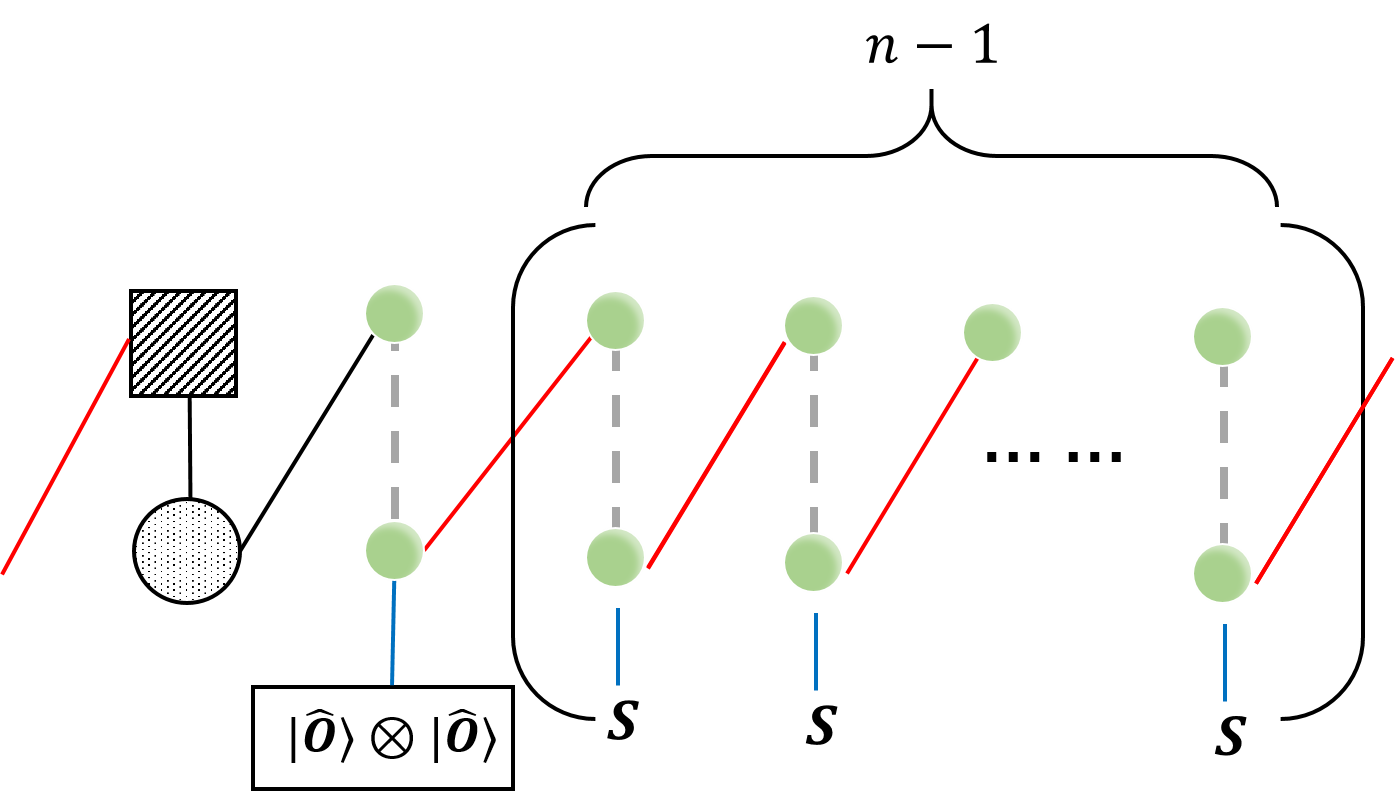}{0.3},
\end{equation}
where we can see that it has the similar structure in Eq.~(\ref{outputoverallu+}) with $\Delta = 0$, so that the contraction result reads:
\begin{equation}
\label{onsiteoverallu+}
  \langle (\partial_\theta \mathcal{L}_l)^2\rangle = C_2\left[\frac{\text{Tr}(\hat{O}^2) D}{D^2 d^2 -1}\right]+
    C_3 \left[\frac{\text{Tr}(\hat{O}^2) D^2}{D^2 d^2 -1} \eta^{n-1} + \frac{\text{Tr}(\hat{O}^2) D}{D^2 d^2 -1} \frac{1-\eta^{n-1}}{1-\eta}\chi\right].
\end{equation}

For the case where both $U_+$ and $U_-$ form unitary $2$-design, the variance of derivative can be written as:
\begin{equation}
\ipic{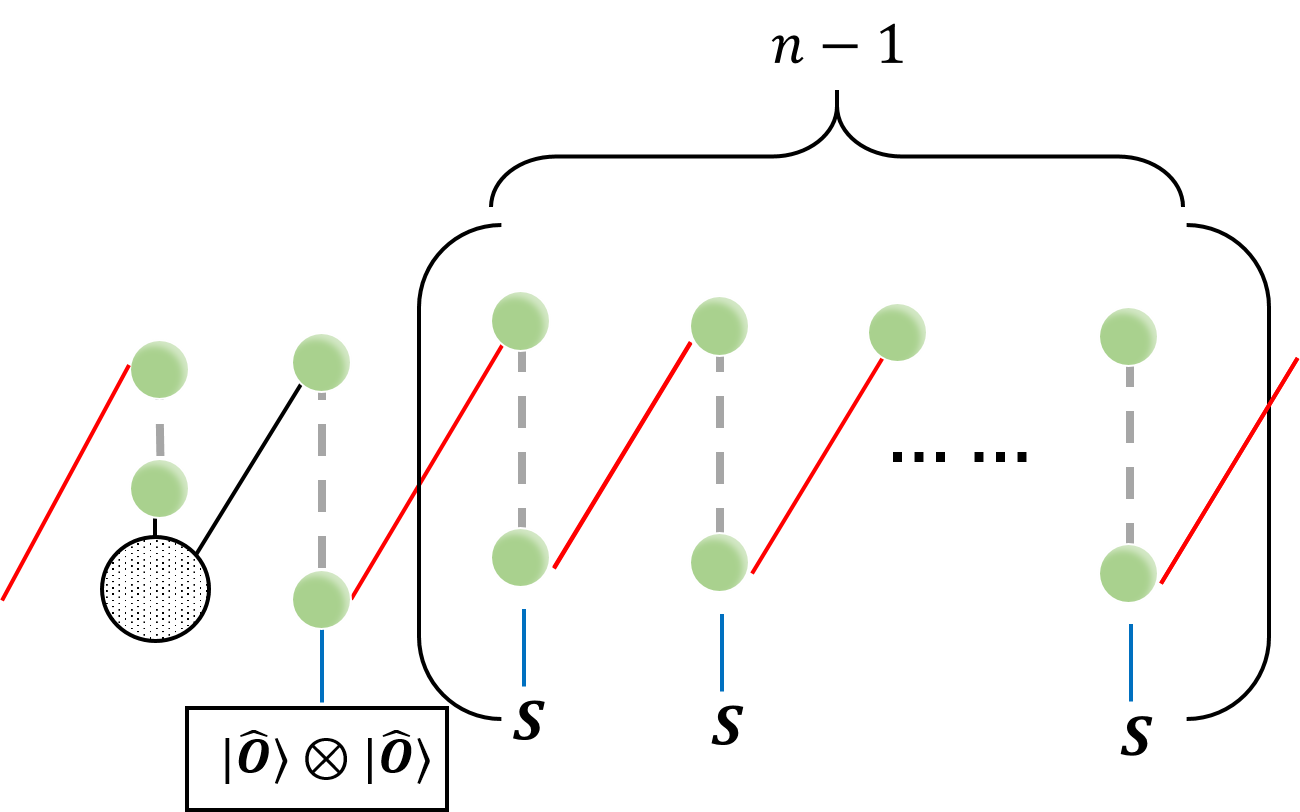}{0.3},
\end{equation}
where we can see that it has the similar structure in Eq.~(\ref{outputoverallu-u+}) with $\Delta = 0$, so that the contraction result reads:
\begin{equation}
\label{onsiteoverallu-u+}
         \langle (\partial_\theta \mathcal{L}_l)^2\rangle =  C_4 \frac{\text{Tr}(\hat{O}^2)D}{(D^2 d^2-1)^2}
    \left[D \eta^{n-1} + \frac{1-\eta^{n-1}}{1-\eta} \chi -\frac{1}{D^2 d^2}\right].
\end{equation}
As we can see that in Eqs.~(\ref{onsiteoverallu-}, \ref{onsiteoverallu+}, \ref{onsiteoverallu-u+}), for all of them, the increasing $n$ will eventually converge to a constant number that only depends on the virtual bond dimension $D$, the physical dimension $d$, and the trace of the square of the local observable $\text{Tr}(\hat{O}^2)$.

We combine the results in the off-site case in Eqs.~(\ref{outputoverallu-}, \ref{outputoverallu+}, \ref{outputoverallu-u+}) and the on-site case in Eqs.~(\ref{onsiteoverallu-}, \ref{onsiteoverallu+}, \ref{onsiteoverallu-u+}). As we fix the distance $\Delta$ (where $\Delta = 0$ in the on-site case),  the variance of $\partial_k \mathcal{L}_l$ scales as:
\begin{equation}
    \text{Var}(\partial_k \mathcal{L}_l) \sim \mathcal{O}\left(\text{Tr}(\hat{O}^2)\frac{P(D,d)}{Q(D,d)} \right),
\end{equation}
where $P(D,d)$ and $Q(D,d)$ are certain polynomials of $D$ and $d$ with constant degrees. This result indicates an absence of the barren plateaus in training process of the tensor-network based machine learning model in the local loss function case.

As the distance $\Delta > 0$, the variance 
\begin{equation}
    \text{Var}(\partial_k \mathcal{L}_l)\leq \mathcal{O}(d^{-\Delta}),
\end{equation}
which indicates that the derivative is upper bounded by an exponentially small number with respect to the distance $\Delta$, and only the derivatives with respect to nearby parameters play a role in the training process. 

With the discussions in this section, we complete the proof of Theorem 2 in the main manuscript.

\section{Numerical Details}

We implement the derivative of the loss function by the automatic differentiation package in the TensorFlow library \cite{Abadi2016TensorFlow}. The contraction of tensors is implemented by using the tensor network library \cite{roberts2019tensornetwork}. To speed up the computation, we use the GPU version of TensorFlow.

For the overlap global loss function:
\begin{equation}
     \mathcal{L}_g =1-\frac{|\langle\psi(\Theta)|\phi\rangle|^2}{\langle\psi(\Theta)|\psi(\Theta)\rangle \langle \phi | \psi \rangle },   
\end{equation}
we define the parameterized matrix product states $| \psi(\Theta) \rangle $ as:
\begin{equation}
    |\psi(\Theta)\rangle=\sum_{j_1, \dots, j_n}\text{Tr}\left[A_{j_1}^{(1)}(\Theta_1)A_{j_2}^{(2)}(\Theta_2)\dots A_{j_i}^{(i)}(\Theta_i)\dots A_{j_n}^{(n)}(\Theta_n)\right]|j_1,\dots j_n\rangle,
\end{equation}
where the tensor $A_{j_i}^{(i)}(\Theta_i)$ has the dimension $d \times D^2$ with the virtual bond dimension $D=2$ and the physical dimension $d = 2$.  The constant state
\begin{equation}
    |\phi\rangle = \sum_{j_1, \dots, j_n} \text{Tr}\left[C^{(1)}_{j_1}C^{(2)}_{j_2}\dots C^{(n)}_{j_n}\right]|j_1, j_2, \dots, j_n\rangle
\end{equation}
with each element in $C^{(i)}_{j_i}$ is fixed as $1$. 

For the KL divergence loss function
\begin{equation}
    \mathcal{L}_{g,2} = D_{KL}(Q(\phi)||P(\phi,\psi)),
\end{equation}
we setup a binary classification task (with the label $0$ for accept and label $1$ for reject), where we only consider one input data $|\phi\rangle$. The accept (reject) probability in hypothesis is $P(\phi, \psi)_{\text{accept}} = |\langle \phi | \psi(\Theta)\rangle|/\sqrt{Z_g}$ ($P(\phi, \psi)_{\text{reject}} = 1 - |\langle \phi | \psi(\Theta)\rangle|/\sqrt{Z_g}$), and the accept (reject) probability in data is $Q(\phi)_{\text{accept}} = 1$ ($Q(\phi)_{\text{reject}} = 0$).

In the local loss function, we define the loss function as:
\begin{equation}
        \mathcal{L}_l = \frac{\langle \psi (\Theta) |\hat{O}_m |\psi(\Theta)\rangle}{\langle \psi (\Theta) |\psi(\Theta)\rangle},
\end{equation}
where we take the local operator as the single site Pauli operator $\sigma_x$ which is acted on the $\lfloor \frac{n}{2} \rfloor$-th site. The virtual bond dimension and the physical dimension are the same as the global loss function case.

As for calculating the variance of $\partial_i^{(k)} \mathcal{L}$, we use the Monte Carlo method to generate $n_s$ number of  $\partial_i^{(k)} \mathcal{L}_l$. At each step, we randomly choose the elements of each parameterized tensor $A_{j_i}^{(i)}(\Theta_i)$  from the uniform distribution ranging from $-0.5$ to $0.5$. Then we calculate:
\begin{equation}
    \text{Var}(\partial_i^{(k)} \mathcal{L}) = \frac{1}{n_s}\sum_{s=1}^{n_s} \left\{\partial_k^{(i)}\mathcal{L}_l[\Theta(s)]\right\}^2 - \left\{\frac{1}{n_s} \sum_{s=1}^{n_s} \partial_k^{(i)}\mathcal{L}_l[\Theta(s)]\right\}^2.
\end{equation}
For every $10000$ steps, we check the convergence with the previous result. The relative convergence error is set as $10^{-3}$.

\end{document}